\documentclass[11pt]{article} 
\usepackage[margin=1in]{geometry}

\usepackage{float}

\usepackage{framed}

\usepackage{verbatim}

\usepackage{amssymb,amsfonts,amsmath,amsthm}

\usepackage{enumerate}

\usepackage[font=small,labelfont=bf]{caption}

\usepackage{color}

\usepackage{graphicx}
\usepackage{tabulary}
\usepackage{diagbox}
\usepackage[dvipsnames]{xcolor}
\usepackage{pgfplots}
\usepackage[colorlinks=true, linkcolor=blue, urlcolor=blue, citecolor=ForestGreen]{hyperref}
\usepackage{cleveref}

\usepackage{comment}

\usepackage{todonotes}

\usepackage{algorithm}
\usepackage[noend]{algorithmic}

\newcommand{\p}{\mathrm{par}}
\newcommand{\rank}{\mathrm{rank}}
\newcommand{\child}{\mathrm{child}}
\newcommand{\first}{\mathbf{first}}
\newcommand{\last}{\mathbf{last}}
\newcommand{\size}{\mathbf{size}}

\usepackage{setspace}

\newcommand{\poly}{\operatorname{poly}}
\newcommand{\polylog}{\operatorname{polylog}}

\newcommand{\MST}{\operatorname{{\bf MST}}}
\newcommand{\cost}{\operatorname{{\bf MST}}}

\newcommand{\R}{\mathbb{R}}

\newcommand{\BE}{\mathbb{E}}

\newcommand{\pos}{\mathrm{pos}}
\newcommand{\NEXT}{\mathrm{NEXT}}

\usepackage{footnote}

\newcommand{\cB}{\mathcal{B}}
\newcommand{\cC}{\mathcal{C}}

\newcommand{\cP}{\mathcal{P}}
\newcommand{\cQ}{\mathcal{Q}}
\newcommand{\cR}{\mathcal{R}}

\newcommand{\eps}{\varepsilon}

\newtheorem{theorem}{Theorem}
\newtheorem{lemma}{Lemma}
\newtheorem{corollary}{Corollary}
\newtheorem{observation}{Observation}

\newtheorem{proposition}{Proposition}

\newtheorem{claim}{Claim}
\theoremstyle{definition}
\newtheorem{definition}{Definition}

\usepackage[framemethod=TikZ]{mdframed}
\newcounter{Frame}

\title{Massively Parallel Algorithms for High-Dimensional Euclidean Minimum Spanning Tree}

\author{
Rajesh Jayaram \\
Google Research\\
\texttt{rkjayaram@google.com}
\and
Vahab Mirrokni\\
Google Research\\
\texttt{mirrokni@google.com}
\and 
Shyam Narayanan\\
MIT\thanks{Work done as a student researcher at Google Research.}\\ 
\texttt{shyamsn@mit.edu}
\and 
Peilin Zhong\\
Google Research\\
\texttt{peilinz@google.com}
}
\date{\today}

\begin{document}

\maketitle

\begin{abstract}

We study the classic Euclidean Minimum Spanning Tree (MST) problem in the Massively Parallel Computation (MPC) model. Given a set $X \subset \mathbb{R}^d$ of $n$ points, the goal is to produce a spanning tree for $X$ with weight within a small factor of optimal. Euclidean MST is one of the most fundamental hierarchical geometric clustering algorithms, and with the proliferation of enormous high-dimensional data sets, such as massive transformer-based embeddings, there is now a critical demand for efficient distributed algorithms to cluster such data sets. 

In low-dimensional space, where $d = O(1)$, Andoni, Nikolov, Onak, and Yaroslavtsev [STOC '14] gave a constant round MPC algorithm that obtains a high accuracy $(1+\epsilon)$-approximate solution. However, the situation is much more challenging for high-dimensional spaces: the best-known algorithm to obtain a constant approximation requires $O(\log n)$ rounds. Recently Chen, Jayaram, Levi, and Waingarten [STOC '22] gave a $\tilde{O}(\log n)$ approximation algorithm in a constant number of rounds based on embeddings into tree metrics. However, to date, no known algorithm achieves both a constant number of rounds and approximation. 

In this paper, we make strong progress on this front by giving a constant factor approximation in $\tilde{O}(\log \log n)$ rounds of the MPC model. In contrast to tree-embedding-based approaches, which necessarily must pay $\Omega(\log n)$-distortion, our algorithm is based on a new combination of graph-based distributed MST algorithms and geometric space partitions. Additionally, although the approximate MST we return can have a large depth, we show that it can be modified to obtain a $\tilde{O}(\log \log n)$-round constant factor approximation to the Euclidean Traveling Salesman Problem (TSP) in the MPC model. Previously, only a $O(\log n)$ round was known for the problem.

\end{abstract}

\thispagestyle{empty}
\newpage
\parskip 7.2pt 
\pagenumbering{arabic}

\section{Introduction}

The minimum spanning tree (MST) problem is one of the most 
fundamental problems in combinatorial optimization, whose 
algorithmic study dates back to the work of Boruvka in 1926~\cite{boruuvka1926jistem}. Given a set of points and distances between the points, the goal is to compute a tree over the points of minimum total weight. The MST problem has received a tremendous amount of attention from the algorithm design community, leading to a large toolbox of methods for the problem~\cite{Charikar:2002,IT03,indyk2004algorithms,10.1145/1064092.1064116,har2012approximate,AndoniNikolov,andoni2016sketching,bateni,czumaj2009estimating,CzumajEFMNRS05,CzumajS04,chazelle00,chazellerubinfeld}.

In the offline setting, where the input graph $G = (V,E)$ is known in advance, an exact randomized algorithm running in time $O(|V| + |E|)$ is known. The version of the problem where the input lies in Euclidean space has also been extensively
studied, see~\cite{eppstein2000spanning} for a survey. In this setting, 
the vertices of the graph are points in $\R^d$, and the 
set of weighted edges is (implicitly given by) the set of all ${n \choose 2}$ pairs of vertices and the pairwise Euclidean distances. Despite the implicit representation, the best known runtime bound for computing an Euclidean MST exactly is $n^2$, even for low-dimensional inputs. For approximations, the best known runtime bound is $O(n^{2-2/(\lceil d/2\rceil +1) +\eps})$ for obtaining a $(1+\eps)$-approximate solution~\cite{agarwal1990euclidean},  and is $n^{1+1/c^2 }$ for an $O(c)$-approximation~\cite{Har-PeledIS13}. 

The Euclidean variant of the MST problem is particularly important in light of the tremendous success of modern \textit{embedding} models in machine learning \cite{mikolov2013efficient,van2008visualizing,he2016deep,devlin2018bert}. Such models encode a data point, such as an image, video, or text, into a vector in high-dimensional space, so that the semantic similarity between two data points is accurately represented by the Euclidean distance between the corresponding embedding vectors. 
In this setting, computing an MST is a well-known and successful technique for clustering data \cite{lai2009approximate,wang2009divide,zhong2015fast,grygorash2006minimum,bateni2017affinity}, which is a core component of many ML pipelines for large embedding datasets.

To deal with the sheer size of these modern embedding datasets, the typical approach is to implement algorithms in massively parallel computation systems such as MapReduce~\cite{dean2004mapreduce,dean2008mapreduce}, Spark~\cite{zaharia2010spark}, Hadoop~\cite{white2012hadoop}, Dryad~\cite{isard2007dryad} and others. The \emph{Massively Parallel Computation (MPC)} model~\cite{karloff2010model,goodrich2011sorting,beame2017communication,AndoniNikolov} is a computational model for these systems that balances accurate modeling with theoretical elegance. The MST problem in particular has been extensively studied in this model \cite{AndoniNikolov,karloff2010model,andoni2018parallel,lattanzi2011filtering,ahanchi2023massively,chen2022new}. For the non-Euclidean case, when the input is a graph with $n$ vertices and $m$ edges, and each machine only has $n^{\varepsilon}$ space for some constant $\varepsilon<1$, algorithms that obtain a constant approximation in $O(\log n)$ rounds and linear $O(m)$ total space are known using connected components algorithms \cite{karloff2010model,andoni2018parallel,behnezhad2019near,coy2022deterministic}. However, improving the round complexity is unlikely, as such an algorithm would refute the longstanding \textsc{1-Cycle vs. 2-Cycle} conjecture~\cite{yaroslavtsev2018massively,roughgarden2018shuffles,lkacki2018connected,assadi2019massively}.

The Euclidean MST problem, on the other hand, is not as well understood in the MPC model as its graph-based counterpart. For low-dimensional Euclidean space, where $d$ is a constant, \cite{AndoniNikolov} gave a $(1+\varepsilon)$ approximate algorithm using only $O(1)$ MPC rounds. However, the landscape becomes much more mysterious and challenging in the high dimensional-setting. 

One prevalent approach for the high-dimensional setting is the \textit{spanner method}: one first constructs a $c$-approximate Euclidean spanner (i.e., a sparse graph over the points whose shortest path distance approximates the Euclidean distance to a factor of $c$) for some constant $c \geq 1$. Such spanners can be constructed with $O(n^{1+1/c})$ edges in $O(1)$ MPC rounds~\cite{epasto2022massively,cohen2022massively}. With the spanner in hand, one can simply run the aforementioned graph-based MST algorithm to obtain a constant approximation in $O(\log n)$ rounds. However, in the Euclidean variant we have the benefit of the metric-space structure, and the one-cycle two-cycle lower bound does not apply. Thus, settling for $O(\log n)$ rounds for Euclidean MST is undesirable and perhaps unnecessary. 

To date, the only known method for obtaining MPC algorithms with a $o(\log n)$ round complexity for high-dimensional MST is the \textit{tree-embedding method}~\cite{bartal1996probabilistic}. In this method, one probabilistically embeds the Euclidean points $X \subset \R^d$ into a (log-depth) tree-metric and then computes the optimal MST in the tree metric. The latter is a simple object that can be computed in the MPC model in $O(1)$ rounds and $\tilde{O}(n)$ total space. Indyk \cite{indyk2004algorithms} gave a tree embedding that obtains a $O(d \log n)$ approximation and can be computed in the MPC model in $O(1)$ rounds and $\tilde{O}(n)$ total space. This was later improved by \cite{chen2022new} to a $\tilde{O}(\log n)$ approximation, resulting in a $O(1)$ round $\tilde{O}(\log n)$ approximation algorithm. However, it is known that the tree-embedding method must suffer a $\Omega(\log n)$ approximation in the worst case \cite{bartal1996probabilistic}.

In summary, one can obtain constant approximations in $O(\log n)$ rounds via the spanner method, or a $\tilde{O}(\log n)$ approximation in constant rounds via the tree embedding method. However, both methods individually face hard barriers to removing the $O(\log n)$ factor from their round complexity or approximation (respectively). A natural question is whether this trade-off is intrinsic: namely if any algorithm for high-dimensional Euclidean MST must always use at least $\Omega(\log n)$ rounds or pay a $\Omega(\log n)$ approximation. Specifically, in this work we address the following question:

\begin{quote}
 \begin{center}
  {\it  Is it possible to compute an $o(\log n)$-approximate Euclidean minimum spanning tree in $o(\log n)$ rounds of the MPC model?}
 \end{center}
 \end{quote}

 An even stronger question is whether a $O(1)$ approximation is possible in $O(1)$ rounds. 
As a positive signal in this direction, a recent result~\cite{chen2023streaming} demonstrated that the \textit{cost} of the MST can be estimated to a constant factor in $O(1)$ MPC rounds. However, their algorithm is a sampling-based estimator that is far removed from a procedure that can actually compute an approximate MST. Moreover, separations between the complexity of estimating the cost of a solution and producing that solution are ubiquitous in high-dimensional geometry (e.g. for metric MST in the sublinear query model \cite{indyk1999sublinear,czumaj2009estimating}). For computing the MST in a distributed setting, such a result was only known in the more powerful Congested Clique \cite{jurdzinski2018mst} model, however implementing this algorithm in the MPC model would require $\Omega(n)$ space per machine.

In this work, we provide a positive resolution to the above question by designing a fully scalable\footnote{Meaning that each machine has only $n^{\eps}$ space for a constant $\eps < 1$, see Preliminaries \ref{sec:prelims}.} MPC algorithm in $\tilde{O}(\log\log n)$ rounds for a constant approximation of the Euclidean MST. In addition, the total space required by our algorithm is at most $O(n^{1+\varepsilon}+nd)$ where $\varepsilon>0$ can be an arbitrarily small constant. Our result makes substantial progress towards the stronger goal of a $O(1)$ approximation in $O(1)$ rounds. Specifically, our main result is:

\begin{theorem}[see Theorem \ref{thm:mpc_mst}]\label{thm:MSTMainIntro}
Given a set $X \subset \mathbb{R}^d$ of $n$ points, there is a MPC algorithm which outputs an $O(1)$-approximate MST of $X$ with probability at least $0.99$ in at most $O(\log\log(n)\cdot \log\log\log(n))$ rounds.
The total space required is at most $O(nd+n^{1+\varepsilon})$ and the per-machine space is $O((nd)^\varepsilon)$, where $\varepsilon>0$ is an arbitrarily small constant.
\end{theorem}

At a high-level, our algorithm bypasses the $\Omega(\log n)$ barriers intrinsic to the spanner and tree-embedding methods by combining the two approaches.
Specifically, our algorithm builds a spanner and attempts to compute the MST of that spanner. However, while doing so we do \textit{not} forget about the original metric structure.
Specifically, our algorithm will exploit the metric structure by generating a spanning forest using \textit{both} the edges of the spanner and geometric space partitions (which underpin tree-embedding methods).

\textbf{Eucledian TSP.} We leverage our algorithm for MST to develop the first $O(1)$-approximation to the Euclidean traveling salesman problem (TSP) in $o(\log n)$ rounds. 
At a first glance, one may think that  $O(1)$-approximate TSP should directly follow from a $O(1)$-approximate MST since any shortcutted traverse of the $O(1)$-approximate MST gives an $O(1)$-approximate TSP.
However, the approximate MST that we computed may have diameter $\Theta(n)$, and all existing fully scalable MPC algorithms require at least $\Omega(\log(\text{diameter}))=\Omega(\log n)$ rounds to compute a traversal of the tree.
To resolve this issue, we develop a new $O(\log\log(n))$-round fully scalable MPC algorithm to compute a traverse of our approximate MST by utilizing the information of a hierarchical decomposition of the point set that we generate while computing the approximate MST. This results in the following:

\begin{theorem}[see Corollary~\ref{cor:tsp}]
Given a set $X \subset \mathbb{R}^d$ of $n$ points, there is a TSP algorithm which outputs an $O(1)$-approximate TSP of $X$ with probability at least $0.99$ in at most $O(\log\log(n)\cdot \log\log\log(n))$ rounds.
The total space required is at most $O(nd+n^{1+\varepsilon})$ and the per-machine space is $O((nd)^\varepsilon)$, where $\varepsilon>0$ is an arbitrarily small constant.    
\end{theorem}

\subsection{Preliminaries}\label{sec:prelims}

\textbf{MPC Model.} In the Massively Parallel Computation (MPC) model, there are $p$ machines and each machine has local memory $s$; thus the total space available in the system is $p\cdot s$. 
The space is measured in words where each word has $O(\log(p s))$ bits. 
The input data has size $N$ and is distributed arbitrarily on $O(N/s)$ machines at the beginning of the computation.
If the total space satisfies $p\cdot s = O(N^{1+\gamma})$ for some $\gamma \geq 0$, and the local space satisfies $s = O(N^{\eps})$ for some constant $\eps \in (0,1)$, then the model is called the $(\gamma,\eps)$-MPC model~\cite{andoni2018parallel}.

The computation in the MPC model proceeds in rounds. 
In every round, each machine performs arbitrary local computation on the data stored in its memory.
At the end of each round, each machine sends some messages to the other machines.
Since each machine only has local memory $s$, the total size of messages sent or received by a machine in one round can not be larger than $s$.
For example, a machine can send a single message with size $s$ to an arbitrary machine, or it can send a size $1$ message to other $s$ machines, but it cannot send a size $s$ message to every machine in one round.
In the next round, each machine only holds received messages in its local memory.
At the end of the computation, the output is stored in a distributed way on the machines.
The parallel running time (number of rounds) of an MPC algorithm is the number of above computation rounds needed to finish the computation.

We consider $\eps\in(0,1)$ to be an arbitrary constant in this paper, i.e., our algorithms can work when the memory of each machine is $s=O(N^{\eps})$ for any constant $\eps\in(0,1)$.
Such algorithms are called \emph{fully scalable algorithms}~\cite{andoni2018parallel}.
Our goal is to develop fully scalable algorithms which only require a small number of rounds and a small total space.

\textbf{Basic Notation.} In the remainder of the paper, we use $X$ to denote a dataset of points that we wish to solve either Minimum Spanning Tree or Traveling Salesman Problem over. We use $n$ to denote the size of $X$ and $d$ to denote the dimensionality of $X$ (i.e., $n = |X|$ and $X \subset \R^d$). Additional notation is defined in the relevant sections.

\textbf{Euclidean Minimum Spanning Tree (MST) and  Travelling Salesman Problem (TSP)} In the Euclidean MST problem, the input is a set of $n$ points $X \subset \R^d$, where we assume that each coordinate $x_i$ of a point $x \in \R^d$ can be stored in a single word of space (i.e., $O(\log ps)$ bits). The points $X$ implicitly define a complete graph, where the vertices are $X$, and for any $x,y \in X$ the weight of the edge $(x,y)$ is $\|x-y\|_2$. Our goal will be to produce a spanning tree $T$ of this complete graph such that the weight of $T$, defined as $\sum_{(x_i, x_j) \in T} \| x_i - x_j \|_2$ is within a constant factor of the minimum spanning tree weight. We write $\MST(X)$ to denote the optimal MST weight.

In the Euclidean TSP problem, the input is the same as the MST problem.
But instead of outputting a spanning tree, we want to output a Hamiltonian cycle $C$ of $X$, i.e., each point appears on the cycle exactly once, such that the total length of the cycle $\sum_{(x_i,x_{j})\in C}\|x_i-x_{j}\|_2$ is minimized up to a constant factor.

\subsection{Technical Overview}\label{sec:techoverview}

\subsubsection{Approximate Euclidean MST}

Our starting point is a key structural fact about minimum spanning trees (observed in prior work on sublinear MST algorithms~\cite{chazellerubinfeld,czumaj2009estimating}) that links the cost of the minimum spanning tree of a weighted graph $G$ to the number of connected components in a sequence of auxiliary graphs. Namely, given a set of points $X \subset \R^d$ with pairwise distances in the range $(1,\Delta)$,\footnote{Note that we will later be able to assume that $\Delta \leq \poly(n)$.} and a distance threshold $t \geq 0$, we define the $t$-\emph{threshold graph} $G_t = (X,E_t)$ to be the graph where $(x,y) \in E_t$ if and only if $\|x-y\|_2 \leq t$. We write $\cP_t$ to denote the set of connected components of $G_t$.

Now consider the steps taken by Kruskal's MST algorithm: at the beginning, all vertices are in their own (singleton) connected component, and then at each step two connected components are merged by an edge of minimum weight. Thus, the number of edges added with weight in the range $(t,2t]$ is precisely $|\cP_t| - |\cP_{2t}|$, so
\begin{equation}\label{eqn:intro1}
    \cost(X) \leq  \sum_{i=0}^{\log (\Delta)-1}2^{i+1} (|\cP_{2^i}| - |\cP_{2^{i+1}}|)
    = n - \Delta + \sum_{i=0}^{\log (\Delta) } 2^i |\cP_{2^i}| \leq 2 \cost(X) 
\end{equation}
This suggests the following approach (which we call the \emph{ideal algorithm}): for each level $t = 2^i$ where $i = 1,2,\dots \log(\Delta)$, compute the set $\cP_t$ of connected components of $G_t$. Then, for each $t$ and every connected component $C \in \cP_t$, we contract the vertices in $C$ into a single super-node in the graph $G_{2t}$; call the resulting contracted graph $\bar{G}_{2t}$.   
Next, we run a \textit{unweighted} spanning forest algorithm on the $\bar{G}_{2t}$, and output every edge we found in this forest. Notice that this gives a valid spanning tree. Moreover, since each edge in $G_{2t}$ has weight at most $2t$, by the above this spanning tree will be a $2$-approximate MST. Because the sets $\cP_t$ are fixed (i.e., independent of the algorithm), the above procedure can be run in parallel for each value of $t$.

The first challenge to this approach is that the graph $G_t$ can be dense, namely, it may have $\Omega(n^2)$ edges. Since the total space available to our algorithm is only $O(n^{1+\eps})$, we will need to compress $G_t$. This is precisely what is accomplished by the \textit{spanner method}. Namely, for every $t$ one can create a graph $S_t$ that is an $O(1/\eps)$-approximate spanner of $G_t$ and has at most $n^{1+\eps}$ edges \cite{Har-PeledIS13}\footnote{\cite{Har-PeledIS13} actually achieves $\sim 1/\sqrt{\varepsilon}$ approximation. But we only consider $\varepsilon=O(1)$, and optimizing such dependence is not the focus of this paper.}; moreover, this construction is efficiently implementable in the MPC model (see, e.g., \cite{epasto2022massively,cohen2022massively}). Such a graph $S_t$ has the property that for every $x,y \in X$ with $\|x-y\|_2 \leq t$, there is a path of length at most $2$ between $x$ and $y$ in $S_t$, and for every edge $(w,z)$ in $S_t$ we have $\|w-z\|_2 \leq O(t/\epsilon)$. By adding the edges of $S_t$ into $S_{t'}$ for each $t' \geq t$, we ensure that the edges of the graphs $S_t$ are monotone increasing in $t$ (since we only consider $O(\log \Delta)$ values of $t$, this increases the number of edges by at most a $O(\log \Delta)$ factor). Let $\cP_t'$ be the set of connected components in $S_t$.  It follows from the spanner property that $|\cP_{O(t/\epsilon)}| \leq |\cP_t'| \leq |\cP_t|$. Thus, to obtain a $O(1/\epsilon)$ approximation, it will suffice to swap out $G_t$ with the spanner $S_t$ in the above ideal algorithm.

The second (more serious) challenge is computing the connected components of $S_t$ in the MPC model. Specifically, unless the one-cycle two-cycle conjecture is false, in general computing the connected components of a graph in the MPC model requires $\Omega(\log n)$ rounds. Thus, if we construct the spanner graphs $S_t$ and then forget about the original metric space that $S_t$ came from, then computing the connected components of $S_t$ in the MPC model will require $\Theta(\log n)$ rounds.\footnote{Note that this yields the complexity of the ``spanner method'' described earlier in the introduction.} Instead, our goal will be to run a connectivity algorithm while crucially using the metric structure of the original points $X \subset \R^d$. In what follows, we describe our approach to doing this. 

\paragraph{Approximately Computing Connected Components: Leader Compression with an Early Termination.}

Instead of running an MPC connectivity algorithm on $S_t$ as a black box, we will need to open the actual inner workings of the algorithm to analyze its interplay with the underlying metric structure of the graph. At a high level, our approach will be to cut off the execution of this algorithm early and show that the intermediate solution (set of connected subgraphs) that one obtains from this partial execution is good enough to compute a constant approximate MST. To this end, we will now describe a connectivity algorithm known as \textit{leader compression}. 

The leader compression algorithm proceeds in rounds. In each round, every vertex $u \in S_t$ flips a coin; the vertices that flip heads are called ``leaders'' and the vertices that flip tails are called ``followers''. Then, each follower vertex $u$ merges into a uniformly random leader vertex $v \in S_t$ such that $(u,v)$ is an edge. If no such edge to a leader vertex exists, $v$ is untouched on that step. Each edge $(u,v)$ that is merged in the process is contracted in the graph into a super-node, and then the process is repeated in the next round where each super-node flips a coin to be either a leader or a follower. Thus, at every time step, each super-node represents a connected sub-graph of $S_t$ which we may subsequently grow on later steps. Ultimately, each connected component in $S_t$ will be contracted into a single super-node. Since on each round, every vertex is merged into another vertex with probability at least $1/4$, the process will terminate after $O(\log n)$ rounds. However, since each round of leader compression takes $O(1)$ rounds in the MPC model, we cannot afford to run leader compression to completion. Instead, our approach will be to cut off the leader compression algorithm early and return the intermediate super-nodes obtained in the process.

To analyze the early stopping of leader compression, we observe two useful properties of this algorithm: firstly, the set of edges that are merged form a spanning forest of $S_t$, so we can use these edges for our approximate MST.\footnote{If $A,B \subset S_t$ are two super-nodes that are merged together during leader-compression, then we can pick any arbitrary edge between $A,B$ to be used for the spanning forest.} Secondly, after any $h \geq 1$ rounds of the process, for any connected component $C$ in $S_t$ with $m$ vertices, we expect there to be at most $m/2^{\Omega(h)}$ super-nodes in $C$. Call a super-node \textit{complete} if it is maximal, i.e. it contains its entire connected component. After $h$ rounds of leader compression, for every connected component $C$ we expect that either $C$ originally had size $2^{\Omega(h)}$, or $C$ is contained by a complete super-node.

Our main approach is then as follows:  we set $h = O(\log \log n)$, and run $h$ rounds of leader compression on $S_t$. We refer to a super-node remaining after $h$ rounds as an approximate connected component, and write $\hat{\cP}_t$ to denote the set of such approximate components at level $t$. We then attempt to run the \textit{ideal algorithm} described earlier, but using the set  $\hat{\cP}_t$  instead of the true set of connected components $\cP_t'$ of $S_t$. Since we terminated leader compression early, there may be many more approximate components than true connected components. However, for every true connected component $C \in \cP_t'$, if $C$ was not complete then leader compression at least reduced the number of vertices in $C$ by a factor of $2^{\Omega(h)}$. Our goal will be to use this fact to argue that for an \textit{incomplete} component $C$, the actual MST cost of $C$ is much larger than the cost we must pay for having under-merged $C$ (i.e., splitting $C$ into multiple approximate components).
To make this argument, we will make use of ideas from the tree-embedding literature. 

\paragraph{Using Tree-Embeddings to Handle Incomplete Components. }
Briefly, the idea of using tree-embeddings to generate approximate MSTs is as follows. First, for every $t=1,2,4,\dots,\Delta$, one can impose a randomly shifted hypergrid over $\R^d$ with side length $t/\sqrt{d}$. The random shift ensures that points $x,y$ that are much closer than $t/\sqrt{d}$ are unlikely to be split, and points $x,y$ with $\|x-y\|_2 > t$ will deterministically be split. If $n_t$ is the number of non-empty hyper-grid cells at level $t$ (i.e., cells that contain at least one point in $X$), then results from the tree-embedding literature~\cite{indyk2004algorithms,andoni2008earth,chen2022new} imply that with good probability:\footnote{Note that given a nested set of hyper-grid cells, one can easily compute an approximate spanning tree with cost at most $\sum_{t=2^i} t \cdot n_t$, see, e.g., \cite{indyk2004algorithms,chen2022new}.}

\begin{equation}\label{eqn:quadtree}
     \MST(X) \leq \sum_{t=2^i} t \cdot n_t \leq \polylog(n) \cdot \MST(X)
\end{equation}Equation $\ref{eqn:quadtree}$ is promising, as it relates the number of non-empty cells $n_t$ to the MST cost --- if we can show that leader compression returns significantly fewer approximate connected components than there were non-empty cells, then this would satisfy our earlier goal. We employ this result in the following way. First, we similarly impose a randomly shifted grid of size length $t/\sqrt{d}$, and perform an initial merging of all points inside of the same grid cell (i.e., points in the same grid cell are automatically merged together). Since the diameter of a grid cell is $t/\sqrt{d} \cdot \sqrt{d} = t$, these points will necessarily be in the same connected component in $S_t$.

After pre-merging points in the same grid cell, we next perform $h = O(\log \log n)$ rounds of leader compression, and again write $\hat{\cP}_t$ to denote the resulting set of approximate connected components. By the above, each connected component $C \in \cP_t'$ is either fully merged at this point, or we have reduced the number of super-nodes in $C$ by a $1/2^{\Omega(h)}$ factor. Since each grid cell was merged into a super-node before the start of leader-compression, it follows that $|\hat{\cP}_t| \leq |\cP_t'| + n_t/2^{\Omega(h)}$.

Now by Equation \ref{eqn:intro1} and the fact that the components $\cP_t'$ well-approximate $\cP_t$, we have $\MST(G) \approx \sum_{t= 2^i} t \cdot |\cP_t'|$. Thus, if we output $|\hat{\cP}_t|$ instead of $|\cP_t'|$ connected components, then at level $t$ we would be paying an additional cost of at most $t \cdot n_t/2^{\Omega(h)}$ in our spanning tree (note though that it is not clear yet that we can actually achieve this cost algorithmically, since we still need to find the edges to merge the components in $\hat{\cP}_t$). Then by Equation \ref{eqn:quadtree}, the fact that there are $n_t$ non-empty cells with side-length $t/\sqrt{d}$ implies that the actual MST cost $\MST(G)$ is at least $n_t \cdot t/(\sqrt{d} \cdot \polylog(n))$. 
Putting this together, the additional cost we pay is at most $\sqrt{d} \cdot \polylog(n)/2^{\Omega(h)} \cdot \MST(G)$ at level $t$. By standard dimensionality reduction for $\ell_2$, we may assume $d = \Theta(\log n)$, so for $h = O(\log \log n)$, this additional cost is at most $\MST(G)/(\log n)^{\Theta(1)}$. Since we only incur this additive cost for $O(\log n)$ geometrically increasing values of $t$, the total additive error incurred from under-merging components is still a small $\MST(G)/(\log n)^{\Theta(1)}$. 

\paragraph{Challenge: Maintaining Consistency of Approximate Connected Components.} If we ran the above procedure and simply counted the number of approximate connected components, by the above discussion this would be sufficient to obtain a constant approximation of the \emph{cost} of the minimum spanning tree. However, there are major issues in using this approach when trying to generate the actual tree. 
Specifically, to create a valid spanning tree when using the approximate components $\hat{\cP}_t$ in the \textit{ideal algorithm}, it is necessary that $\hat{\cP}_t$ is a refinement of $\hat{\cP}_{2t}$ for every $t$ a power of $2$; otherwise, we would be unable to generate the edges merging $\hat{\cP}_t$ into $\hat{\cP}_{2t}$ without creating cycles or leaving vertices disconnected. 

One possibility is to use the components in $\hat{\cP}_{t}$ as a starting point to generate the components in $\hat{\cP}_{2t}$. However, this would require us to create $\hat{\cP}_1, \hat{\cP}_2, \hat{\cP}_4, \dots$ sequentially, which would require $\log n$ rounds. Recall that the original \textit{ideal algorithm} did not have this issue, as the sets $\cP_t$ of connected components did not depend on the prior execution of the algorithm. Instead, our approach is to split the set of levels $\{1,2,4,8,\dots,\Delta\}$ into smaller chunks of $O(\log \log n)$ levels each, and show that the approximate components for each chunk can be computed in parallel. 
Specifically, we set $\alpha = (\log n)^{O(1)}$, and then define the ``checkpoint'' levels $\{1,\alpha,\alpha^2,\alpha^3,\dots,\Delta\}$. For every two checkpoint levels $t,t'$, we will compute $\hat{\cP}_t$  and $\hat{\cP}_{t'}$ independently and in parallel. The challenge is to do this while maintaining consistency between $\hat{\cP}_t$ and $\hat{\cP}_{\alpha t}$ for each 
checkpoint level $t$ .

 To ensure consistency between $\hat{\cP}_t$ and $\hat{\cP}_{\alpha t}$, we need to ensure that any pair $(u,v)$ of vertices merged in our leader compression algorithm $\hat{\cP}_t$ was also merged in $\hat{\cP}_{\alpha t}$. To this end, first recall that our leader compression algorithm for $S_t$ and $S_{\alpha t}$ started by merging all points in the same grid cell with side length $t/\sqrt{d}$ (resp., $\alpha t/\sqrt{d}$). Thus, if we enforce that our leader compression algorithm at level $t$ \textit{only} merges together pairs $u,v \in X$ that are in the same grid cell of side-length $\alpha t/\sqrt{d}$, then consistency will follow automatically. To enforce this condition, we simply modify the graph $S_t$ by removing any edge $(u,v) \in S_t$ which crosses the randomly shifted grid with side-length $\alpha t/\sqrt{d}$. Since $\alpha = (\log n)^{O(1)}$ is taken sufficiently large, because of the random shift any edge $(u,v)$ in $S_t$ is cut with probability at most $1/(\log n)^{O(1)}$, as we must have had $\|u-v\|_2 = O(t) = O( \alpha t / (\sqrt{d} \log^{O(1)} n))$. This allows us to show that the cost of omitting these cut edges in $S_t$ is small, and enables us to compute $\hat{\cP}_t$ independently of $\hat{\cP}_{\alpha t}$. Once we have computed $\hat{\cP}_t$ and $\hat{\cP}_{\alpha t}$ (for every checkpoint level $t$), we now hope to generate the intermediate sets $\hat{\cP}_{2t}, \hat{\cP}_{4t}, \dots \hat{\cP}_{\alpha t/2}$ while maintaining this consistency property. Since there are only $O(\log \alpha ) = O(\log \log n)$  intermediate sets, we can afford to compute them sequentially in $O((\log \log n)^2)$ rounds.

However, a challenge arises when attempting to compute the intermediate sets $\hat{\cP}_{2t}, \hat{\cP}_{4t}, \dots \hat{\cP}_{\alpha t/2}$. Namely, these sets still must be consistent with $\hat{\cP}_{\alpha t}$, so for any $2t \leq \tau\leq \alpha t/4$, we cannot merge two vertices $u,v \in S_{\tau} $ that are not in the same approximate connected component in $\hat{\cP}_{\alpha t}$. Unlike the situation with $S_t$, we cannot afford to cut edges in $S_\tau$ that cross the hyper-grid with side-length $\alpha t/\sqrt{d}$, since $\tau$ may no longer be significantly smaller than $\alpha t / \sqrt{d}$, so the probability that an edge in $S_\tau$ is cut is no longer small. Instead, one would need to cut every edge $(u,v)$ in $S_\tau$ such that $u,v$ were \textit{not} in the same approximate component in $\hat{\cP}_{\alpha t}$. Now if $\hat{\cP}_{\alpha t} = \cP_{\alpha t}$ were the true components in $S_t$, no such edges would exist. However, this will not be the case since we terminate leader compression early when constructing $\hat{\cP}_{\alpha t}$, and leader compression does \textit{not} have the property that each edge is equally likely to be contracted on a given round. 

In fact, removing such edges from $S_\tau$ which crosses the partition  $\hat{\cP}_{\alpha t}$ can significantly increase the number of connected components in $S_\tau$ beyond what we can afford. To see this, consider the following instance of ``parallel'' path graphs: set $k = (2^h \cdot \alpha)^{O(1)} = (\log n)^{O(1)}$ sufficiently large, and let $X_0 = \{e_1, 2  e_1, 3 e_1,\dots, k e_1 \}$ be a path graph, where $e_i \in \R^d$ is the standard basis vector for $i \in [d]$, and set $X_j = \{x  + \frac{\alpha}{\sqrt{2}} e_j \; | \; x \in X_0\}$, for each $j \in [k]$, and let $X = \cup_{j=0}^k X_j$. Now consider running leader compression at levels $\sqrt{\alpha}$ and $\alpha$; at level $\alpha$, each vertex $x \in X_j$ is connected to at $\Omega(k)$ other vertices in parallel paths $X_i$ for $i \neq j$, but is connected to at most $O(\alpha) \ll k/(\log n)^{O(1)}$ points in the same path $X_j$. It follows that on every step of leader compression, the probability that $x$ is merged with a point in the same path is $1/(\log n)^{O(1)} \ll 1/2^h$, so after $h$ rounds of leader compression at level $\alpha$ each point $x \in X_j$ will be in a super-node containing no other points from $X_j$ with probability $1-1/(\log n)^{O(1)}$; call such a point $x$ totally cut. Note that there will then be $\Omega(k^2)$ totally cut vertices. Now when running leader compression at level $\sqrt{\alpha}$, the graph $S_{\sqrt{\alpha}}$ only contains edges within the same path, so if we cut edges that were not merged at level $\alpha$ then every isolated vertex becomes a singleton in $S_{\sqrt{\alpha}}$, thus $|\hat{\cP}_{\sqrt{\alpha}}| = \Omega(k^2)$, so the cost of the MST produced will be at least $\sqrt{\alpha} \cdot |\hat{\cP}_{\sqrt{\alpha}}| = \Omega(\sqrt{\alpha} k^2) > \Omega(\sqrt{\alpha}) \cdot \MST(X)$, which is a bad approximation. Thus, it is not possible to obtain a constant factor approximation while using under-merged clusters.

\paragraph{The Solution to Inconsistency: Generate Over-Merged Clusters.}  
Our solution to the above issue is to generate approximate connected components that are  \emph{over-merged} instead of under-merged. In other words, each approximate component in the partition $\hat{\cP}_t$ that we output will contain a full connected component in the true partition $\cP_t$, and possibly more vertices as well.  Consider any checkpoint level $t$ --- when running leader compression on $S_t$, we 
have that guarantee that for every true connected component $C$, either we fully merge $C$ (i.e. $C$ is complete), or we split $C$ into at most $m / 2^{\Omega(h)}$ super-nodes, where $m$ was the number of hyper-grid cells with side length $t/\sqrt{d}$ that intersected $C$. Instead of simply outputting these under-merged super-nodes as-is, we first perform an over-merging step where we arbitrarily merge together \textit{every} ``incomplete'' super-node within in the same hyper-grid cell with the (larger) side length  $\alpha t/\sqrt{d}$.  Specifically, we can choose an arbitrary representative super-node $v$ that is incomplete within such a hyper-grid cell, and connect every other incomplete supernode $u$ in the same cell to $v$ via an arbitrary edge. Since $u,v$ were in the same hyper-grid cell, this edge will have weight at most $\alpha t$. Also, recall that we modified $S_t$ to remove edges crossing this larger hyper-grid; thus it follows that every resulting merged cluster fully contains at least one connected component in $S_t$ (i.e., we only over-merge). We then let $\hat{\cP}_t$ be this set of over-merged clusters. 

To argue that this over-merging of incomplete super-nodes does not increase the cost significantly, we note that for a super-node to be incomplete, it must have intersected at least $2^{\Omega(h)}$ hyper-grid cells of length $t/\sqrt{d}$ in expectation. By Equation \ref{eqn:quadtree}, if there are $\ell$ incomplete super-nodes in a cell, this means that $\MST(X) \geq t 2^h \ell / \poly(\log n)$. On the other hand, we pay a cost of $t \alpha$ to connect each incomplete super-node to the representative, for a total cost of $t \alpha \ell$.  But taking $h= O(\log \log n)$ sufficiently large, it follows that $t2^h \ell / \polylog(n) \gg t \alpha \ell$, thus we can afford this arbitrary over-merging step at level $t$. This handles the consistency between distinct checkpoint levels $1,\alpha,\alpha^2,\dots,\Delta$.

We must now consider maintaining consistency for the intermediate levels between checkpoints. Specifically, we run this over-merging algorithm to generate $\hat{\cP}_1, \hat{\cP}_\alpha, \hat{\cP}_{\alpha^2}, \dots$ in parallel. Then, for each $\alpha^k$, we generate the intermediate levels $\hat{\cP}_{2 \alpha^k}, \hat{\cP}_{4 \alpha^k}, \dots, \hat{\cP}_{\alpha^{k+1}/2}$ sequentially, where $\hat{\cP}_{2^j \alpha^k}$ is generated using $h = O(\log \log n)$ rounds of leader compression on top of the previous set of merged components in $\hat{\cP}_{2^{j-1} \alpha^k}$. Now by construction of the over-merging procedure that we used to generate $\hat{\cP}_{\alpha^{k+1}}$, the only way for an edge $(u,v)$ to be in a graph $S_\tau$, for any $\alpha^{k} < \tau < \alpha^{k+1}$ but \emph{not} have been merged in $\hat{\cP}_{\alpha^{k+1}}$ is if $(u,v)$ crossed the hyper-grid cell with side length $\alpha^{k+2} /\sqrt{d}$. But now this is acceptable, because the probability that such a cut occurs is at most $\sqrt{d}/\alpha = 1/(\log n)^{O(1)}$ over the random shift, so we can now safely remove these edges from $S_\tau$ to maintain consistency.

Note that the above algorithm result in $O\left((\log \log n)^2\right)$ rounds of MPC. However, we can further improve this by a careful application of binary-search. Namely, if we start with $\hat{\cP}_t$ and $\hat{\cP}_{\alpha t}$, we first generate $\hat{\cP}_{\alpha^{1/2} t}$, then in parallel generate $\hat{\cP}_{\alpha^{1/4} t}, \hat{\cP}_{\alpha^{3/4} t}$, and so on. This reduces the number of sequential rounds to $O(\log \log \alpha) = O(\log \log \log n)$, and since each round needs $O(\log \log n)$ rounds of leader compression, we use only $\tilde{O}(\log \log n)$ total rounds of MPC. However, at each intermediate step, we again need to over-merge for consistency reasons. This time, if we know $\hat{\cP}_{t/\gamma}$ and $\hat{\cP}_{t \cdot \gamma}$ and are trying to generate $\hat{\cP}_t$, we merge incomplete components in the same connected component of $\hat{\cP}_{t \cdot \gamma}$, so that we ensure consistency.

\paragraph{Generating the edges.} Given the above algorithm that maintains consistency, generating edges is now quite simple. 
Assuming that we have found the approximate connected components $\hat{\cP}_1, \hat{\cP}_2, \allowbreak\hat{\cP}_4,\allowbreak \dots$, we can generate edges merging $\hat{\cP}_t$ into $\hat{\cP}_{2t}$ for every $t = 1, 2, 4, \dots$, in parallel for each $t$.
For a fixed $t$, this can be done by performing $O(\log \log n)$ rounds of leader compression on $S_{2t}$
starting with $\hat{\cP}_t$. But this time, rather than just updating the connected components, we keep track of the edges we found.
While this will not generate all of the edges for this level, the same arguments as before will imply that leader compression finds $1 - \frac{1}{(\log n)^{\Theta(1)}}$ fraction of edges needed. Finally, to fully connect $\hat{\cP}_{2t}$, we can add arbitrary edges, and by the same argument used to analyze the tree embedding, the additional cost will not be too large.

\subsubsection{Approximate Euclidean TSP}
Recall that Euclidean TSP aims to find a cycle over the points such that each point is visited exactly once, and we want to minimize the total weight of the cycle.
A Euclidean MST is a $2$-approximation for the Euclidean TSP problem since a shortcut Euler tour of the MST gives a valid solution for TSP; here, an Euler tour of a tree is a directed circular tour on the tree such each undirectly tree edge $(u,v)$ appears exactly twice: once in each direction $(u,v)$ and $(v,u)$. Therefore, it suffices to compute an Euler tour of the approximate MST produced by our earlier algorithm.
However, this approximate MST could be a path with diameter $\Theta(n)$. Unfortunately, the best-known algorithm for computing an Euler tour algorithm requires $\Omega(\log(\text{diameter}))=\Omega(\log(n))$ rounds~\cite{andoni2018parallel}. Moreover, improvements on this are unlikely as an algorithm using $o(\log n)$ rounds algorithm would refute \textsc{1-Cycle vs. 2-Cycle} conjecture~\cite{yaroslavtsev2018massively,roughgarden2018shuffles,lkacki2018connected,assadi2019massively}.

Fortunately, in addition to the edges of approximate MST, our algorithm also outputs an $O(\log n)$-level hierarchical decomposition $\hat{\cP}_1,\hat{\cP}_2,\hat{\cP}_4,\hat{\cP}_8,\cdots,$ of the point set.
We will show that, though our approximate MST may have large diameter on the original point set, if we look at $\hat{\cP}_{t/2}$ for any fixed $t$, and regard each cluster in $\hat{\cP}_{t/2}$ as a node, then each cluster in $\hat{\cP}_t$ is merged from a subset of these nodes, and the tree introduced by the edges that we selected connecting these nodes (clusters in $\hat{\cP}_{t/2}$) has diameter at most $\polylog(n)$.

To see why this is true, consider how each cluster in $\hat{\cP}_t$ was constructed when merging clusters from $\hat{\cP}_{t/2}$.
Suppose $C\in\hat{\cP}_t$ is merged from $\cC = \{C_1,C_2,\cdots,C_k\} \subseteq \hat{\cP}_{t/2}$. 
To generate the edges, we ran $h=O(\log \log n)$ rounds of leader compression, and in each round we merged some subsets of $\cC$ by creating  edges between pairs $(C_i,C_j) \in \cC^2$.
Let $\hat{\cP}_{t}^{(i)}$ be the partition we obtained after $i$ rounds of leader compression.
Then each cluster in $\hat{\cP}_{t}^{(1)}$ corresponds to a star graph over a subset of $\cC$ (with the leader as the center), and similarly each cluster in $\hat{\cP}_{t}^{(2)}$ is a star graph over a set of clusters in $\hat{\cP}_{t}^{(1)}$, and so on.
Since the clusters in $\hat{\cP}_{t}^{(1)}$ correspond to a tree over $\cC$ with diameter $2$ (i.e. a star), this means that each merged tree over $C_1,C_2\cdots,C_k$ in $\hat{\cP}_{t}^{(2)}$ has diameter at most $3\cdot 2 +2$. 
Then by induction, the full tree over the vertex set $\cC$ in $\hat{\cP}_{t}^{(i)}$ has diameter $2\cdot 3^{i-1}-1$. Finally, after $O(\log\log n)$ rounds of leader compression, we create a star on the clusters in  $C_1,C_2,\cdots,C_k$ that were still unmerged, which will blow up the diameter by another factor of at most $2$. Thus, the tree $T = (\cC, E(T))$ of resulting from combining all these edges has diameter at most $3^{O(\log\log n)}=\polylog(n)$. We call this tree the \textit{super-node tree} over $\cC$

We will use this fact to obtain a Euler tour of this tree in $O(\log\log n)$ MPC rounds (e.g., by applying the Euler tour algorithm of~\cite{andoni2018parallel}, and using that the diameter is small). 
To this end, we now introduce a critical sub-problem, which we call the Euler Tour Join problem. 
Given any $t$ and $i \geq 0$, the input to the problem is the following:

\begin{enumerate}
\item A cluster $C \in \hat{\cP}_t$ that was merged from  $\cC = \{C_1,C_2,\cdots,C_k\} \in \hat{\cP}_{t/2^{2^{i}}}$, and an Euler tour $A$ of the super-node tree $T$ over $\cC$.  
\item For each $C_i$, suppose $C_i$ was merged from clusters $C_{i,1},C_{i,2},\cdots,C_{i,k_i} \in \hat{\cP}_{t/2^{2^{i+1}}}$. Then we are also given the super-node tree $T_i$ over $\{C_{i,1},C_{i,2},\cdots,C_{i,k_i}\}$, and an Euler tour $A_{i}$ of $T_i$.
\end{enumerate}
Given the above trees $T,T_1,T_2\cdots,T_k$ and tours $A,A_1,A_2,\cdots,A_k$, the goal of the of the Euler Tour Join problem is to compute (1) the super-node tree $T'$ over $\cup_{i=1}^k \{C_{i,1},C_{i,2},\dots,C_{i,k_i}\}$ which represents how $C$ was merged from the clusters in $\hat{\cP}_{t/2^{2^{i+1}}}$ and (2) an Euler tour $A'$ of $T'$. Notice that if we can solve this problem for any $i\geq0$, then we can run it for $i=0,1,2,\dots,\log \log (\Delta)$ sequentially, by pairing up groups of levels in $\{1,2,4,8,\dots,\Delta\}$, merging them, and the recusing on the merged groups. Thus, after $O(\log \log n)$ rounds, we would have a Euler tour of the full original dataset $X$. Thus, our main task now is to develop an $O(1)$ round fully-scalable MPC algorithm for the Euler Tour Join problem using total space $O(n)$.

We now describe the high level ideas behind our Euler Tour Join algorithm.
Let $C,C_1,\dots,C_k$, $T,T_1,\dots, T_k$, and $A,A_1,\dots,A_k$ be as above, and let $V=\cup_{i=1}^k \{C_{i,1},C_{i,2},\dots,C_{i,k_i}\}$. We will sometimes think of each $C_i$ as being composed of the the component sets $\{C_{i,1},C_{i,2},\dots,C_{i,k_i}\}$, and write $x \in C_i$ to denote that $x \in \{C_{i,1},C_{i,2},\dots,C_{i,k_i}\}$.
Now observe that for every edge $(C_i,C_j) \in E(T)$, there must have been a unique edge $(C_{i,a},C_{j,b})$, where $C_{i,a} \in C_i, C_{j,b} \in C_j$, used by the leader compression algorithm to connect $C_i,C_j$, between two clusters in $\hat{\cP}_{t/2^{2^{i+1}}}$. Thus, we can define a function $g:E(T) \to V^2$ that specifies this mapping (i.e., in the above example $g((C_i,C_j)) =(C_{i,a},C_{j,b})$).
It is easy to see that the tree $T' = (V,E(T'))$ with edges $\{g(C_i,C_j)\mid (C_i,C_j) \in E(T)\}\cup T_1\cup T_2\cup\cdots \cup T_k$ is a spanning tree of $V$. Thus, the challenge will be to compute an Euler tour of $T'$.
In what follows, any $x,y \in V$ such that $(x,y) \in E(T')$ and $X \in C_i, y \in C_j$ and $i \neq j$, we call both $x$ and $y$ \textit{terminal nodes}. In other words, terminal nodes are clusters from $V$ that connected two clusters from $\{C_1,\dots,C_k\}$.

Given the above, one natural idea to construct $A'$ is as follows.
We  follow along the Euler tour of $A$; each edge $(C_i,C_j)$ in $A$ corresponds to an edge $(u,v)=g(C_i,C_j) \in V^2$ between two terminals $u\in C_i,v\in C_j$ which means that our tour enters cluster $C_j$ via $v$ after leaving $C_i$.
Suppose the edge that follows $(C_i,C_j)$ in $A$ is $(C_j,C_l)$ and the corresponding terminals are $(x,y)$, then it means that our tour leaves $C_j$ via $x$ before entering $C_l$.
Then, it would be natural to simply plug the subtour $A_j$ that connects $v$ to $x$ in our final joined tour $A'$ to connect $(u,v)$ and $(x,y)$.
However, this approach fails, as it may use duplicated edges in $A_j$.
In fact, there exists a Euler tour $A$ such that if we visit terminals with respect to the ordering provided by $A$, we can never find a valid Euler tour $A'$.
As an example, suppose we have $5$ clusters 
\[C_1=\{i_1\},\quad C_2=\{i_2\},\quad C_3=\{i_3\},\quad C_4=\{i_4\}\quad ,C_5=\{i_{5,1},i_{5,2},i_{5,3},i_{5,4},i_{5,5}\}\]
where the inter cluster edges are $\{i_1,i_{5,1}\},\{i_2,i_{5,2}\},\{i_3,i_{5,3}\},\{i_4,i_{5,4}\}$, and $T_5$ has edges $\{i_{5,5},i_{5,2}\},\allowbreak\{i_{5,5},\allowbreak i_{5,3}\},\allowbreak\{i_{5,5},i_{5,4}\},\{i_{5,2},i_{5,1}\}$ with Euler tour $A_5:i_{5,1}\rightarrow i_{5,2}\rightarrow i_{5,5}\rightarrow i_{5,3}\rightarrow i_{5,5}\rightarrow i_{5,4}\rightarrow i_{5,5}\rightarrow i_{5,2}\rightarrow i_{5,1}$.
Let $T$ be a star with Euler tour $A:C_5 \rightarrow C_1\rightarrow C_5 \rightarrow C_4\rightarrow C_5\rightarrow C_2\rightarrow C_5\rightarrow C_3\rightarrow C_5$.
It is obvious that if one wants to follow the subpath $C_1\rightarrow C_5 \rightarrow C_4$, the inner path of $T_5$ must include $i_{5,1}\rightarrow i_{5,2}\rightarrow i_{5,5}\rightarrow i_{5,4}$ as a subsequence.
Similarly, if one wants to follow the subpath $C_2\rightarrow C_5\rightarrow C_3$, the inner path of $T_5$ must include $i_{5,2}\rightarrow i_{5,5}\rightarrow i_{5,3}$ as a subsequence.
Thus, we must use the (directed) $i_{5,2}\rightarrow i_{5,5}$ twice.

The main issue of the above approach is that the order of visiting the neighbors $(C_1,C_4,C_2,C_3)$ of $C_5$ in the Euler tour $A$ is not consistent with the order of visiting the terminals $(i_{5,1},i_{5,2},i_{5,3},i_{5,4})$ of $C_5$ in $A_5$.
Namely, $A_5$ suggests the order $(C_1,C_2,C_3,C_4)$.
Therefore, we need to compute a new Euler tour $\bar{A}$ such that the ordering of visiting the neighbors of each $C_i$ in $T$ is consistent with ordering of visiting terminals in $A_i$.

To this end, we develop a novel algorithm in the MPC model such that if each edge in $T$ is given a weight, and we are able to compute the total weight of every path from each node to the root of $T$, and we are also able to compute the size of each subtree of $T$, then we can efficiently compute the position of each directed edge of $T$ in the desired Euler tour $\bar{A}$ in parallel.
Fortunately, we show that above subtree sum problem and path weight sum problem can be solved efficiently using the known Euler tour $A$.
Therefore, can compute $\bar{A}$ in $O(1)$ rounds and $O(n)$ total space in the fully scalable setting.
Then, the remaining process of computing $A'$ becomes simple, we chop each Euler tour $A_i$ into paths between terminals.
Then we follow the ordering of edges in $\bar{A}$.
For a length $2$ path $C_i\rightarrow C_j\rightarrow C_l$  in $\bar{A}$, the relevant terminals are $(u,v)=g(C_i,C_j),(x,y)=g(C_j,C_l)$, then we insert the corresponding subsequence between $u$ and $x$ of Euler tour $A_j$ into the place between $(u,v)$ and $(x,y)$ in our final Euler tour $A'$.
Note that this sequence insertion subroutine can be efficiently implemented in the MPC model as shown by~\cite{andoni2018parallel}.

\subsection{Other Related Work}
The MST problem has been studied extensively in multiple models of sublinear computation, including streaming, distributed algorithms, and the sublinear query model.  
In the sublinear query model, the implicit input is the set of $\binom{n}{2}$ distances and the goal is to estimate 
the weight of the MST while making a sublinear number of queries to the distances. It is known that any algorithm which actually computes an approximate MST requires $\Omega(n^2)$ queries~\cite{indyk1999sublinear}, hence the focus on estimating the cost. To this end, Chazelle, Rubinfeld, and Trevisan~\cite{chazellerubinfeld} gave an algorithm based on estimating connected components, that gives a $(1+\epsilon)$-factor approximation for arbitrary graphs of maximum degree $D$ and edge weights in $[1,W]$ using 
at most $O(D W \epsilon^{-3})$ queries.
This result was improved by Czumaj and Sohler~\cite{czumaj2009estimating} for metric (e.g., Euclidean) MST, who
gave a $(1+\eps)$ approximation with $\tilde{O}(n \eps^{-8})$ queries. 

In the streaming model, the points $X \subset \R^d$ arrive in a stream, possibly with deletions, and the goal is to estimate the cost of the MST in small space (as outputting the MST would require $\Omega(n)$ space). The first algorithm for streaming Eucledian MST was due to Indyk~\cite{indyk2004algorithms}, who gave a $O(\log^2 n)$ approximation in $\polylog(n)$ space.  This was later improved by \cite{chen2022new} to a $\tilde{O}(\log n)$ approximation in $\polylog(n)$ space, and then a $\tilde{O}(1/\eps^2)$ approximation in $O(n^\eps)$ space by \cite{chen2023streaming}. The first two works employed tree embedding based approaches, whereas \cite{chen2023streaming} used a connected component-based estimator similar to those used in \cite{chazellerubinfeld,czumaj2009estimating}. For low-dimensional space, \cite{10.1145/1064092.1064116} gives a $(1+\eps)$ approximation with space exponential in the dimension, although the techniques in this paper rely heavily on the construction of exponentially sized $\eps$-nets.

The above streaming algorithms are linear sketches and therefore can be used in the MPC model to obtain constant round approximations of the cost of the MST (see e.g. \cite{AndoniNikolov} for a reduction from lienar sketching to the MPC). However, there is a substantial gap between estimating the cost of the MST and producing the MST. For instance, the estimators in the papers \cite{chazellerubinfeld, czumaj2009estimating,chen2023streaming} are based on sampling vertices and computing the size of their connected components, thus no edges or approximate tree structure can be derived from this approach. As described in Section \ref{sec:techoverview}, substantial challenges arise when attempting to construct a tree in a consistent (i.e., no cycles) and cost-effective way.

\section{Description of the MST Algorithm in the Offline Setting}\label{sec:offline_algorithm}

In this section, we give a description for generating an $O(1)$-approximate MST, that can be implemented in the MPC model. However, we defer the details of the MPC implementation to Section~\ref{sec:MPC_MST_implementation}.

We first assume WLOG that $d = \Theta(\log n)$, using the Johnson-Lindenstrauss lemma~\cite{JL}.
We will define two parameters $\alpha, \beta$ such that both $\alpha/\beta$ and $\beta$ are at least $(\log n)^C$ for some sufficiently large constant $C$. We also assume WLOG $\alpha = 2^{2^g}$ for some integer $g$ (which can always be done by replacing $\alpha$ with some $\alpha' \in [\alpha, \alpha^2]$).
Define $G$ to be the complete weighted graph on $X$, where $(x, y) \in X$ have an edge of weight $d(x, y)$. Define the \emph{threshold graph} $G_t$ to connect two points $x, y$ if $d(x, y) \le t$. We also define $\cP_t$ to be the partition of $X$ based on the connected components of $G_t$.

Finally, since we are only hoping for a constant-factor approximation, by Standard discretization methods (see, e.g., Proposition 1 in~\cite{chen2023streaming}), we may assume the aspect ratio is at most $\tilde{O}(n)$. More precisely, we may assume that the minimum distance between any two points is at least $\alpha^{100} = (\log n)^{O(1)}$, and the maximum distance between any two points is at most $O(n \cdot \alpha^{101}) \le n^2$ (for $n$ sufficiently large).
We shift the point set such that each coordinate of each point has value in $[0,\Delta]$ where $\Delta$ is a sufficiently large power of $2$ and $\Delta = O(n^2)$.

We also note a few definitions.

\begin{definition}
    Given a dataset $X$, a \emph{partition} of $X$ is a split of $X$ into one or more pairwise disjoint subsets of $X$, such that every element $x \in X$ is in exactly one of the subsets.

    Given two partitions $\cP, \cQ$ of $X$, we say that $\cP$ \emph{refines} $\cQ$ (or equivalently, $\cQ$ is \emph{refined by} $\cP$) if every partition component in $\cP$ is a subset of some partition component in $\cQ$. We use the notation $\cP \sqsupseteq \cQ$ (or $\cQ \sqsubseteq \cP$).
\end{definition}

\begin{definition}\label{def:adic_valuation}
    For a positive integer $n$, we define the \emph{2-adic valuation} $v_2(n)$ to be the largest nonnegative integer $k$ such that $2^k|n$.
\end{definition}

\paragraph{High level approach:} First, we will approximately generate $\cP_t$, the connected components of $G_t$, for $t = 1, \alpha, \alpha^2, \dots, \alpha^{H}$, for $H = \Theta(\log n/\log\log n)$. We will approximately compute each of these partitions $\cP_t$ in parallel, using $h = O(\log \log n)$ rounds. Next, given the approximate connected components for $G_{\alpha^k}$ and $G_{\alpha^{k+1}}$, we attempt to generate approximate connected components for $G_{2 \alpha^k}, G_{4 \alpha^k}, \dots, G_{\alpha^{k+1}/2}, G_{\alpha^{k+1}}$. Finally, we generate edges to form an approximately minimal spanning tree.

\paragraph{Quadtree and Spanner:}
We start off knowing a randomly shifted \emph{quadtree} $Q$. We recall the definition of a quadtree, along with some relevant notation.

\begin{definition}
    A \emph{randomly shifted Quadtree} is constructed as follows. 
    First, we choose a random vector $a = (a_1,a_2,\dots,a_d)\in\mathbb{R}^{d}$, where each coordinate $a_i$ is drawn uniformly from $[0,\Delta]$.
    For each $t$ that is a power of $2$ between $1$ and $2 \Delta$ (where $\Delta = \Theta(n^2)$ is a power of $2$), we use $a$ to generate a grid of side length $t$, which we call the \emph{grid at level $t$}. Specifically, there is a one-to-one correspondence between each grid cell in level $t$ and each vector $(c_1,c_2,\dots,c_d)\in\mathbb{Z}^d$, i.e., the corresponding cell denotes the set of points:
$
\{(x_1,x_2,\dots,x_d)\in\mathbb{R}^d \mid \forall i\in [d], x_i\in[c_i\cdot t-a_i,(c_{i}+1)\cdot t-a_i)\}\subset \mathbb{R}^d.
$
Finally, for any cell $c$ at level $t$, define $X_c$ to be the set of points in $X$ contained in the cell $c$.
\end{definition}

Note that the grids are nested, i.e., for every $t$, every cell of side length $t$ in the Quadtree is contained in a cell of side length $2t$.
Finally, for the largest grid length $2\Delta$, it is easy to see that there is a unique cell $c$ in level $\Delta$ such that $X_c = X$. This is because every $a_i \in [0, \Delta]$ and because we assume all data points have coordinates in $[0, \Delta]$.

Next, we will also assume we have a series of $2$-hop Euclidean spanners. We recall the definition of a $2$-hop Euclidean spanner.

\begin{definition}
    Given a dataset $Y \in \R^d$, a $C$-approximate $2$-hop Euclidean spanner of length $t$ is a graph on $Y$ with edge set $E$, with the following two properties.
\begin{enumerate}
    \item For any $p, q \in X$ such that $\|p-q\|_2 \le t$, either $(p, q) \in E$ or there exists $r \in X$ such that $(p, r), (r, q) \in E$.
    \item For any two points $p, q \in Y$ with $\|p-q\|_2 > t$, we must have $(p, q) \not\in E$.
\end{enumerate}
\end{definition}

In other words, any two points within distance at most $t$ are connected by a path of length at most $2$, but any two points of distance more than $t$ cannot be directly connected in the graph.

We will use the fact that for any dataset $Y$, any level $t$, and any constant $0 < \eps < 1$, there exists a $O(1/\eps)$-approximate $2$-hop Euclidean spanner of $Y$ with at most $O(|Y|^{1+\eps})$ edges.
More specifically, for each $t$ a power of $2$, let $k$ be such that $\alpha^{k-1} < t \le \alpha^k$. We generate $\tilde{G}_t(X_c)$ to be a $2$-hop Euclidean spanner of length $t$ generated on $X_c$, where $c$ is a cell at level $\alpha^{k+1}/\beta$ in the quadtree, where $\tilde{G}_t(X_c)$ has at most $O(|X_c|^{1+\eps})$ edges. In addition, let $\tilde{G}_t$ be the union of $\tilde{G}_t(X_c)$ across all cells $c$ at level $\alpha^{k+1}/\beta$. 
We also make sure that $\tilde{G}_t \subseteq \tilde{G}_{2t}$ for all $t$, simply by adding $\tilde{G}_t$ to $\tilde{G}_{2t}, \tilde{G}_{4t}, \dots$: it is clear that this does not violate the definition of an $O(1/\eps)$-approximate spanner.
Since there are at most $O(\log n)$ such levels, and since $\sum |X_c|^{1+\eps} \le (\sum |X_c|)^{1+\eps} = n^{1+\eps}$, every $\tilde{G}_t$ has at most $O(n^{1+\eps})$ edges.
Finally, we let $\tilde{\cP}_t$ be the partitioning of $X$ based on the connected components of $\tilde{G}_t$. Note that $\tilde{\cP}_t \sqsupseteq \tilde{\cP}_{2t}$ for all $t$, since $\tilde{G}_t \subseteq \tilde{G}_{2t}$.

We note that the assumptions on dimensionality reduction, aspect ratio, and our assumption that we have a Quadtree $Q$ and $2$-hop $O(1/\eps)$-approximate Euclidean spanners $\tilde{G}_t$ can all be achieved using $O(1)$ rounds of MPC. We discuss this in Section~\ref{sec:mpc_mst}.

\paragraph{Leader Compression Algorithm}

We need to first explain the leader compression algorithm. We assume that we start out with some starting partition $\cP$ of a dataset $X$, and are given some graph $H$ on $X$. The goal is to connect components that have edges between them in $H$, to form larger connected components.

Formally, we will define a single round of leader compression on $\cP$ with respect to $H$ as follows. 
For each connected component $C \in \cP,$ we assign every $x \in C$ some leader, which is a vertex $x^* \in C$. A round of leader compression works as follows. 
First, each leader will generate a random bit (uniformly 0 or 1), and will broadcast the value to all of its descendants (i.e., the rest of the vertices in $C$). So, every $x_i$ now has a value of 0 or 1, matching that of its leader. 
Next, every $x_i$ with a value of 1 will send a message to all of its neighbors in the graph $H$. Importantly, after these messages, every $x_j$ that has value 0 knows the set of $x_i$ that have value 1 (which means they do not have the same leader, since $x_j$ has value 0) and are connected to $x_j$ in $H$. For the purpose of MPC implementation, it will suffice for each such $x_j$ to know a single such $x_i$ (if one exists).
Simultaneously, every $x_j$ that has value 0 will choose such an $x_i$ (assuming such an $x_i$ exists), and $x_j$ then sends the $x_i$ value to its current leader $x_k$, in $O(1)$ rounds of MPC. Each leader $x_k$ with value 0 will choose a single descendant $x_j$ with such a $x_i$ (if such an $x_j$ exists: note $j$ could equal $k$). Then, $x_k$ will update its leader (as well as the leader of all of $x_k$'s descendants) to be the leader of $x_i$.

Finally, we can define $h$ rounds of leader compression on $\cP$ with respect to $H$ as follows. For the first round, we run a round of leader compression on $\cP$ with respect to $H$. We now have a potentially coarser partition of connected components, and for the second round we run leader compression on this new partition with respect to $H$. We repeat this for $h$ rounds, and at each round we use the updated partition and assignments of leaders.

We include pseudocode for a round of leader compression in Algorithm \ref{alg:leader-compression}.

\paragraph{Part 1: Approximately generating $\cP_t$, $t = \alpha^k$.}
We start with a randomly shifted grid from $Q$ of level $t/\beta$. We will automatically connect all points that are in the same cell. (Note that all such points must have distance at most $t \sqrt{d}/\beta$ from each other, so they should be in the same connected component at this level anyway). This forms an initial partitioning $\cP_t^{(0)}$ of $X$.

Now, each (nonempty) component $\cP_t^{(0)}$ starts with a ``leader'' point. 
We now perform $h = O(\log \log n)$ rounds of leader compression on $\cP_t^{(0)}$ with respect to $\tilde{G}_t$. At the end of these rounds, we will have connected components and thus have $\bar{\cP}_t$, our preliminary estimate for $\cP_t$.

We will next convert $\bar{\cP}_t$ into $\hat{\cP}_t$ by doing some additional merging. For each component $S \in \bar{\cP}_t$, we will check whether $S$ is a complete component. To do so, we look at all of the edges in $\tilde{G}_t$: for each edge we check whether its endpoints are in the same connected component in $\bar{\cP}_t$. If not, we send a message to both endpoints. Every vertex will then know if there is an edge leaving it into another connected component. Thus, every leader of a connected component will know if the connected component has any edges leaving it. We call such a component \emph{incomplete}. To form $\hat{\cP}_t$, we group together all incomplete components in $\bar{\cP}_t$ with their leaders in the same cell $c$ at level $\alpha^{k+1}/\beta$.

We include pseudocode for this step in Algorithm \ref{alg:part-1}.

\paragraph{Part 2: Approximately generating $\cP_t$, $\alpha^{k} < t < \alpha^{k+1}$, $t$ a power of $2$.}
Recall that $\alpha = 2^{2^g}$. We will generate $\hat{\cP}_t$ in decreasing order of $v_2(\log_2 t)$ (recall Definition~\ref{def:adic_valuation} for $v_2(\cdot)$).
In other words, we first generate $\hat{\cP}_t$ for $t = \alpha^{k + 1/2} = \alpha^k \cdot 2^{2^{g-1}}$, then for $t = \alpha^{k + 1/4} = \alpha^k \cdot 2^{2^{g-2}}$ and $t = \alpha^{k + 3/4} = \alpha^k \cdot 2^{3 \cdot 2^{g-2}}$ simultaneously, and so on. This will result in $g = \log_2 \log_2 \alpha$ iterations.
When generating some $\alpha^{k} < t < \alpha^{k+1}$, we define $\kappa = 2^{2^{v_2(\log_2 t)}}$. We note that $v_2(\log_2 (t/\kappa)), v_2(\log_2 (t \cdot \kappa)) > v_2(\log_2 t)$, so we have already generated $\hat{\cP}_{t/\kappa}$ and $\hat{\cP}_{t \cdot \kappa}$. We will use these two partitions to generate $\hat{\cP}_{t}$.

To do so, we start with the connected components from $\hat{\cP}_{t/\kappa}$. Then, we perform $h = O(\log \log n)$ rounds of leader compression on $\hat{\cP}_{t/\kappa}$ with respect to $\tilde{G}_t$, to create $\bar{\cP}_t$.
Finally, we convert $\bar{\cP}_t$ to $\hat{\cP}_t$, by performing merging in a similar way as in Part 1. Specifically, we check each edge in $\tilde{G}_t$ and see if the endpoints are in the same connected component in $\bar{\cP}_t$, and use this information to determine whether each component in $\bar{\cP}_t$ is complete or incomplete. (Namely, a connected component $C$ in $\bar{\cP}_t$ is \emph{incomplete} iff there exists an edge $\tilde{G}_t$ with exactly one vertex in $C$.) Finally, we group together all incomplete connected components with their leaders in the same connected component in $\hat{\cP}_{t \cdot \kappa}$.

We include pseudocode for this step in Algorithm \ref{alg:part-2}.

\paragraph{Part 3: Generating the edges.} We have listed out the approximate connected components $\hat{\cP}_t$ for each $t$ a power of $2$. Now, for any such $t$, we want to generate edges connecting $\hat{\cP}_{t/2}$ into $\hat{\cP}_t$.

We will again perform $h = O(\log \log n)$ rounds of leader compression on $\hat{\cP}_{t/2}$ with respect to $\tilde{G}_t$, but this time we keep track of all edges that we added. We will show that the partition we generate after doing leader compression still refines $\hat{\cP}_t$. So, to finish generating $\hat{\cP}_t$ along with the necessary edges, we connect all disconnected components in each partition of $\hat{\cP}_t$, using arbitrary edges. 

We include pseudocode for this step in Algorithm \ref{alg:part-3}.

\begin{algorithm}[tb]
   \caption{$\textsc{LeaderCompression}(\cP, \ell, H, \text{Edges})$: A single round of leader compression on $\cP$ with respect to $H$, where $\ell$ is the function mapping each node to its leader in the partition. $\text{Edges}$ is a boolean value that denotes whether we want to generate edges (which is only needed for Step 3).}
   \label{alg:leader-compression}
\begin{algorithmic}[1]
    \STATE $S = \{\ell(x): x \in X\}$ \COMMENT{$S$ is the set of leader nodes}
    \FOR{$x \in S$}
        \STATE $b_x \leftarrow \text{Unif}(\{0, 1\})$. \COMMENT{$b_x$ is the random bit for leader node $x$} 
    \ENDFOR
    \FOR{$x \in X$}
        \STATE $b_x \leftarrow b_{\ell(x)}$.
    \ENDFOR
    \FOR{$x \in X$ with $b_x = 1$}
        \FOR{$e = (x, y) \in H$}
            \STATE $x$ sends message ``$(x,\ell(x))$'' to $y$
        \ENDFOR
    \ENDFOR
    \FOR{$y \in X$ with $b_y = 0$}
        \STATE Select any message $(x,\ell(x))$ sent to $y$ (if it exists)
        \STATE $y$ sends message ``$(x, y, \ell(x))$'' to $\ell(y)$.
    \ENDFOR
    \FOR{$z \in S$ with $b_z = 0$}
        \STATE Select any message $(x, y, \ell(x))$ sent to $z$ (if it exists)
        \IF{$\text{Edges}$}
            \STATE $E \leftarrow E \cup \{(x, y)\}$. \COMMENT{$E$ is a set of edges, initialized to $\emptyset$}
        \ENDIF
        \STATE $z$ broadcasts message ``$\ell(x)$'' to all descendants, all descendants (including $z$) update their leader to be $\ell(x)$.
    \ENDFOR
    \STATE Update partition $\cP$ based on leader function.
    \STATE \textbf{Return} $(\cP, \ell, H, E).$ \COMMENT{If Edges is False, we return False instead of $E$.}
\end{algorithmic}
\end{algorithm}

\begin{algorithm}[tb]
   \caption{$\textsc{Part1}(X, t, Q, \tilde{G}_t)$: Generating $\hat{\cP}_t$ for $t = \alpha^k$.}
   \label{alg:part-1}
\begin{algorithmic}[1]
    \STATE Generate partition $\cP_t^{(0)}$ of $X$ based on the quadtree at level $t/\beta$, i.e., two points in $X$ are in the same connected component in $\cP_t^{(0)}$ if and only if they are in the same cell at level $t/\beta$.
    \STATE Choose an arbitrary leader for each (nonempty) component in $\cP_t^{(0)}$, and let the corresponding leader mapping with respect to $\cP_t^{(0)}$ to be $\ell_t^{(0)}$.
    \FOR{$i = 1$ to $h = O(\log \log n)$}
        \STATE $(\cP_t^{(i)}, \ell^{(i)}_t, \tilde{G}_t, \text{False}) \leftarrow \textsc{LeaderCompression}(\cP_t^{(i-1)}, \ell^{(i-1)}_t, \tilde{G}_t, \text{False})$.
    \ENDFOR
    \FOR{$e = (x, y) \in \tilde{G}_t$}
        \IF{$\ell^{(h)}_t(x) \neq \ell^{(h)}_t(y)$}
            \STATE $(x, y)$ sends \textsc{Incomplete} to both $x$ and $y$
        \ENDIF
    \ENDFOR
    \FOR{$x \in X$ that received \textsc{Incomplete} message}
        \STATE $x$ sends \textsc{Incomplete} to $\ell^{(h)}_t(x)$.
    \ENDFOR
    \STATE Merge all Incomplete components with leaders in the same cell in $Q$ at level $\alpha^{k+1}/\beta$, to form $\hat{\cP}_t$. 
\end{algorithmic}
\end{algorithm}

\begin{algorithm}[tb]
   \caption{$\textsc{Part2}(X, t, \tilde{G}_t, \hat{\cP}_{t/\kappa}, \hat{\cP}_{t \cdot \kappa})$: Generating $\hat{\cP}_t$, given $\hat{\cP}_{t/\kappa}$ and $\hat{\cP}_{t \cdot \kappa}$.}
   \label{alg:part-2}
\begin{algorithmic}[1]
    \STATE Initialize $\cP_t^{(0)} \leftarrow \hat{\cP}_{t/\kappa}$, and choose an arbitrary leader for each component in $\cP_t^{(0)}$.
    Let the corresponding leader mapping with respect to $\cP_t^{(0)}$ to be $\ell_t^{(0)}$
    \FOR{$i = 1$ to $h = O(\log \log n)$}
        \STATE $(\cP_t^{(i)}, \ell^{(i)}_t, \tilde{G}_t, \text{False}) \leftarrow \textsc{LeaderCompression}(\cP_t^{(i-1)}, \ell^{(i-1)}_t, \tilde{G}_t, \text{False})$.
    \ENDFOR
    \FOR{$e = (x, y) \in \tilde{G}_t$}
        \IF{$\ell^{(h)}_t(x) \neq \ell^{(h)}_t(y)$}
            \STATE $(x, y)$ sends \textsc{Incomplete} to both $x$ and $y$
        \ENDIF
    \ENDFOR
    \FOR{$x \in X$ that received \textsc{Incomplete} message}
        \STATE $x$ sends \textsc{Incomplete} to $\ell^{(h)}_t(x)$.
    \ENDFOR
    \STATE Merge all Incomplete components with leaders in the same component in $\hat{\cP}_{t \cdot \kappa}$, to form $\hat{\cP}_t$. 
\end{algorithmic}
\end{algorithm}

\begin{algorithm}[tb]
   \caption{$\textsc{Part3}(X, t, \tilde{G}_t, \hat{\cP}_{t/2}, \hat{\cP}_{t})$: Generating the edges merging $\hat{\cP}_{t/2}$ into $\hat{\cP}_{t}$.}
   \label{alg:part-3}
\begin{algorithmic}[1]
    \STATE Initialize $\cP_t^{(0)} \leftarrow \hat{\cP}_{t/2}$, and choose an arbitrary leader for each component in $\cP_t^{(0)}$. Let the corresponding leader mapping with respect to $\cP_t^{(0)}$ to be $\ell_t^{(0)}$.
    \FOR{$i = 1$ to $h = O(\log \log n)$}
        \STATE $(\cP_t^{(i)}, \ell^{(i)}_t, \tilde{G}_t, E_t^{(i)}) \leftarrow \textsc{LeaderCompression}(\cP_t^{(i-1)}, \ell^{(i-1)}_t, \tilde{G}_t, \text{True})$.
    \ENDFOR
    \STATE Create an arbitrary star on the set of leaders $\{\ell^{(h)}_t(x)\}$ in the same component in $\hat{\cP}_t$, to create a forest of stars $F_t$.
    \STATE \textbf{Return} $F_t \cup \bigcup_{i=1}^h E_t^{(i)}$.
\end{algorithmic}
\end{algorithm}

\section{Analysis of the Offline Algorithm}\label{sec:offline_analysis}

\subsection{Comparison to the Spanner Graph}

Given the spanners $\tilde{G}_t$ for each $t$ a power of $2$ (as described in Section~\ref{sec:offline_algorithm}), we create the graph $\tilde{G}$ as follows. For each pair of distinct vertices $(y, z)$, we assign $(y, z)$ the weight $t$ for $t$ the smallest power of $2$ such that $(y, z) \in \tilde{G}_t$ (if such a $t$ exists). Otherwise, we do not connect $(y, z)$.

We first note that any spanning tree in $\tilde{G}$ does not increase in cost by too much after converting to the corresponding Euclidean spanning tree.

\begin{proposition} \label{prop:MST_spanner_not_too_small}
    Suppose that $T$ is a spanning tree in $\tilde{G}$. Then, the cost of $T$ in $X$ (i.e., over the true Euclidean distance) is at most $O(1/\eps)$ times the cost of $T$ in $\tilde{G}$.
\end{proposition}

\begin{proof}
    We recall that $(x, y) \in \tilde{G}_t$ only if $\|x-y\|_2 \le O(t/\eps)$. Therefore, if $T = \{(y_i, z_i)\}_{i=1}^{n-1}$, the cost of $T$ in $X$ is $\sum_{i=1}^{n-1} \|y_i-z_i\|$. However, if the edge $(y_i, z_i)$ had weight $t_i$, then $\|y_i-z_i\| \le O(t_i/\eps)$, so $t_i \ge \Omega(\eps) \cdot \|y_i-z_i\|$. Thus, the cost of $T$ in $\tilde{G}$ is at least $\sum_{i=1}^{n-1} \Omega(\eps) \cdot \|y_i-z_i\|$. This completes the proof.
\end{proof}

Next, we show that the minimum spanning tree cost in $\tilde{G}$ is not much more than the Euclidean spanning tree cost. 
To do so, we make use of the following two simple but important propositions.

\begin{proposition}[Folklore] \label{prop:quadtree-basic}
    For any edge $e$ connecting two points $p, q \in [0, \Delta]^d$ of length $w = \|p-q\|_2$, the probability that a randomly shifted grid of side length $L \le \Delta$ splits the edge (i.e., $p, q$ are not in the same cell in this grid) is at most $\frac{w \cdot \sqrt{d}}{L}$
\end{proposition}

\begin{proposition}~\cite[restated]{czumaj2009estimating} \label{lem:czumajsohler}
    Let $G$ be any weighted graph on a dataset $X$, with all edges having weight at most $\Delta$ and at least some sufficiently large constant.
    Let $\MST(G)$ denote the weight of the minimum spanning tree of $G$. For each $t$, let $\cP_t$ represent the partition of $X$ representing the connected components of the threshold graph $G_t$ of edges with weight at most $t$.
    Then,
\[\MST(G) = \Theta(1) \cdot \left(\sum_{t = 1}^\Delta (|\cP_t|-1) \cdot t\right),\]
    where the sum ranges over all $1 \le t \le \Delta$ where $t$ is a power of $2$.
\end{proposition}

Let $\MST$ represent the true minimum spanning tree cost of $X$, and for any weighted graph $G$, let $\MST(G)$ represent the minimum spanning tree cost of $G$. First, we prove the following proposition.

\begin{proposition} \label{prop:spanner_cuts_few_edges}
    Suppose that $\alpha^{k-1} < t \le \alpha^{k}.$ Then, $\BE[|\tilde{\cP}_{t}|] \le |\cP_t| + \MST \cdot \frac{\sqrt{d} \cdot \beta}{\alpha^{k+1}}$.
\end{proposition}

\begin{proof}
    Let $L := \alpha^{k+1}/\beta$.
    For each component $S \in \cP_t$, let $M_S$ represent the minimum spanning tree of $X_S$, and $\MST(S)$ represent the cost of $M_S$. Note that $\MST \ge \sum_{S \in \cP_t} \MST(S),$ by Kruskal's algorithm. Now, for each edge $e \in M_S$ of weight $w(e)$, the probability that the randomly shifted grid at level $L$ splits $e$ is at most $\frac{w(e) \cdot \sqrt{d}}{L}$, by \Cref{prop:quadtree-basic}. This means that the expected number of edges in $M_S$ across all $S \in \cP_t$ that are cut is at most $\sum_{S \in \cP_t} \frac{\MST(S) \cdot \sqrt{d}}{L} \le \MST \cdot \frac{\sqrt{d}}{L} = \MST \cdot \frac{\sqrt{d} \cdot \beta}{\alpha^{k+1}}.$ 
    
    Next, for any piece of a tree that has not been cut, all of the points in this piece will be in the same connected component in the spanner $\tilde{G}_t$. Therefore, the additional number of connected components is at most the number of cut edges, which completes the proof.
\end{proof}

Hence, we have the following corollary.

\begin{corollary} \label{cor:MST_spanner_not_too_big}
    The minimum spanning tree cost in $\tilde{G}$, in expectation, is at most $O(1) \cdot \MST(X)$.
\end{corollary}

\begin{proof}
    By~\Cref{lem:czumajsohler}, we can write $\MST(X) = \Theta(1) \cdot \left(\sum_{t} (|\cP_t|-1) \cdot t\right),$ where $t$ ranges as powers of $2$ from $1$ to $\Delta$. Likewise, $\MST(\tilde{G}) = \Theta(1) \cdot \left(\sum_{t} (|\tilde{\cP}_t|-1) \cdot t\right).$ So,
\[\BE[\MST(\tilde{G})] \le O(1) \cdot \left(\MST(X) + \sum_t t \cdot (\BE[|\tilde{\cP}_t|]-|\cP_t|)\right).\]
    Using \Cref{prop:spanner_cuts_few_edges}, and the fact that $t \le \alpha^k$ in \Cref{prop:spanner_cuts_few_edges}, this is at most
\begin{align*}
    O(1) \cdot \left(\MST(X) + \sum_t t \cdot \MST(X) \cdot \frac{\sqrt{d} \cdot \beta}{t \cdot \alpha}\right) &= O(1) \cdot \left(\MST(X) + \MST(X) \cdot \sum_t \frac{\sqrt{d} \cdot \beta}{\alpha}\right) \\
    &= O(1) \cdot \MST(X),
\end{align*}
    as long as $\alpha \ge \sqrt{d} \cdot \beta \cdot \log n$.
\end{proof}

Hence, it suffices to find an $O(1)$-approximate MST in the graph $\tilde{G}$. This tree will have cost at most $O(1) \cdot \MST(X)$ in $\tilde{G}$ by \Cref{cor:MST_spanner_not_too_big}, so by \Cref{prop:MST_spanner_not_too_small}, it also has Euclidean cost at most $O(1/\eps) \cdot \MST(X)$. The rest of the analysis will go into showing the algorithm finds an $O(1)$-approximate MST in $\tilde{G}$.

\subsection{Important Properties of Leader Compression with Early Termination}

Here, we note some simple but important properties of leader compression, and some general properties of the approximate connected components we form.

\begin{definition}
    Given two partitions $\cP$ and $\cQ$ on $X$, we define $\cP \oplus \cQ$ to be the finest partition $\cR$ such that $\cR \sqsubseteq \cP$ and $\cR \sqsubseteq \cQ$. Equivalently, it is the partition generated by merging a spanning forest of $\cP$ and of $\cQ$, and taking the connected components.
    
    Given a graph $H$ on $X$, we define $\cP_H$ as the set of connected components of the graph $H$. We abuse notation and write $\cP \oplus H$ to mean $\cP \oplus \cP_H$.
    %
\end{definition}

First, we note the following basic proposition.

\begin{proposition} \label{prop:leader-compression-basic}
    Let $\cP^{(0)}$ be a starting partition, with a graph $H$. After $h$ rounds of leader compression, let $\cP^{(h)}$ be the set of connected components. Then, $\cP^{(0)} \oplus H \sqsubseteq \cP^{(h)} \sqsubseteq \cP^{(0)}$.
\end{proposition}

\begin{proof}
    Since we are only connecting connected components together, we trivially have that $\cP^{(h)} \sqsubseteq \cP^{(0)}.$ To prove that $\cP^{(0)} \oplus H \sqsubseteq \cP^{(h)}$, we first consider the case that $h = 1$. In this case, we never merge two connected components in $\cP^{(0)}$ unless they had an edge in $H$, so the proof is clear.
    
    For general $h \ge 2$, we proceed by induction (base case $h = 1$ is already done). We know that $\cP^{(0)} \oplus H \sqsubseteq \cP^{(h-1)} \sqsubseteq \cP^{(0)}$. Since $\cP^{(h-1)} \sqsubseteq \cP^{(0)}$, this implies that $\cP^{(h-1)} \oplus H \sqsubseteq \cP^{(0)} \oplus H$. However, we also know that $\cP^{(h-1)} \sqsupseteq \cP^{(0)} \oplus H$, so this clearly implies that $\cP^{(h-1)} \oplus H \sqsupseteq \cP^{(0)} \oplus H$. So, in fact $\cP^{(h-1)} \oplus H = \cP^{(0)} \oplus H$. Therefore, by the base case, if we performing a single round of leader compression on $\cP^{(h-1)}$ to obtain $\cP^{(h)}$, we have that $\cP^{(0)} \oplus H = \cP^{(h-1)} \oplus H \sqsubseteq \cP^{(h)}$, which completes the proof.
\end{proof}

We next note a simple proposition about $\bar{\cP}_t$.

\begin{proposition} \label{prop:bar-P-single-cell}
    For $t = \alpha^k$, we have that $\bar{\cP}_t \sqsupseteq \tilde{\cP}_t$.
    Hence, every connected component $C \in \bar{\cP}_t$ is in a single cell at level $\alpha^{k+1}/\beta$.
\end{proposition}

\begin{proof}
    Initially, $\cP_t^{(0)}$ is based on the quadtree at level $\alpha^k/\beta$, and every two points in the same cell have distance at most $\alpha^k \cdot \sqrt{d}/\beta \le t$ and are also in the same cell at level $\alpha^{k+1}/\beta$, which means they must be in the same connected component in $\tilde{\cP}_t$. Hence $\tilde{\cP}_t \sqsubseteq \cP_t^{(0)}$.
    Then, we perform $h$ rounds of leader compression on $\cP_t^{(0)}$ with respect to $\tilde{G}_t$, to obtain $\bar{\cP}_t \sqsupseteq \cP_t^{(0)} \oplus \tilde{G}_t,$ by \Cref{prop:leader-compression-basic}. However, since $\tilde{\cP}_t \sqsubseteq \cP_t^{(0)}$, this means $\cP_t^{(0)} \oplus \tilde{G}_t = \cP_t^{(0)} \oplus \tilde{\cP}_t = \tilde{\cP}_t$.

    Finally, for $t = \alpha^k$, $\tilde{G}_t$ never crosses a cell at level $\alpha^{k+1}/\beta$. Therefore, any component in $\bar{\cP}_t$ is contained in a single cell at level $\alpha^{k+1}/\beta$.
\end{proof}

The following important lemma helps us understand the approximate connected components $\tilde{\cP}_t$.

\begin{lemma} \label{lem:basic-properties-main}
    For every $t$ a power of $2$, the following hold.
\begin{enumerate}
    \item If $t \le \alpha^k$, then every connected component in $\hat{\cP}_t$ is contained in the same cell at level $\alpha^{k+1}/\beta$.
    \item $\hat{\cP}_t \sqsubseteq \tilde{\cP}_t$.
    \item $\hat{\cP}_t \sqsubseteq \hat{\cP}_{t/2}.$
\end{enumerate}
\end{lemma}

\begin{proof}
    We shall prove these three properties in an inductive manner. In the base case, we prove the first two claims for $t = \alpha^k$, and that $\hat{\cP}_{\alpha^{k+1}} \sqsubseteq \hat{\cP}_{\alpha^k}$.
    For the inductive step, we consider creating $\hat{\cP}_t$, given $\hat{\cP}_{t/\kappa}$ and $\hat{\cP}_{t \cdot \kappa}$, where $\alpha^k \le t/\kappa$ and $t \cdot \kappa \le \alpha^{k+1}$. We inductively assume that every connected component in $\hat{\cP}_{t/\kappa}$ and $\hat{\cP}_{t \cdot \kappa}$ is in a single cell at level $\alpha^{k+2}/\beta$, that $\hat{\cP}_{t/\kappa} \sqsubseteq \tilde{\cP}_{t/\kappa}$ and $\hat{\cP}_{t \cdot \kappa} \sqsubseteq \tilde{\cP}_{t \cdot \kappa}$, and finally, $\hat{\cP}_{t \cdot \kappa} \sqsubseteq \hat{\cP}_{t/\kappa}$. Then, the inductive step proves that every connected component in $\hat{\cP}_{t}$ is in a single cell at level $\alpha^{k+2}/\beta$, that $\hat{\cP}_{t} \sqsubseteq \tilde{\cP}_{t}$, and that $\hat{\cP}_{t \cdot \kappa} \sqsubseteq \hat{\cP}_t \sqsubseteq \hat{\cP}_{t/\kappa}$.

    \medskip
    \textbf{Base case:} $g = 0$, i.e., $t = \alpha^k$.
    Suppose that $x, y \in X$ are in the same connected component in $\tilde{G}_t$. Then, either $x, y$ are in the same component in $\bar{\cP}_t$ (and thus in $\hat{\cP}_t$), or $x, y$ are both in incomplete connected components in $\bar{\cP}_t$. Since $\bar{\cP}_t \sqsupseteq \tilde{\cP}_t$ by \Cref{prop:bar-P-single-cell}, this means the leaders of $x$ and $y$ (in $\bar{\cP}_t$) are in the same connected component in $\tilde{G}_t$, which means they are in the same cell at level $\alpha^{k+1}/\beta$. So, these two components are merged, and thus $x, y$ are in the same component in $\hat{\cP}_t$. Hence, $\hat{\cP}_t \sqsubseteq \tilde{\cP}_t$.
    In addition, every connected component in $\bar{\cP}_t$ is contained in a single cell at level $\alpha^{k+1}/\beta$, by \Cref{prop:bar-P-single-cell}. Therefore, every connected component in $\hat{\cP}_t$ is fully contained in a single cell at level $\alpha^{k+1}/\beta$.
    Finally, all points in a single cell at level $\alpha^{k+1}/\beta$ start off connected as $\cP^{(0)}_{\alpha^{k+1}}$, so this also implies that $\hat{\cP}_{\alpha^{k+1}} \sqsubseteq \hat{\cP}_{\alpha^k}$.

    \medskip
    \textbf{Inductive step:}
    First, we show that $\hat{\cP}_{t \cdot \kappa} \sqsubseteq \hat{\cP}_t \sqsubseteq \hat{\cP}_{t/\kappa}$. The claim that $\hat{\cP}_t \sqsubseteq \hat{\cP}_{t/\kappa}$ is immediate, because we generate $\hat{\cP}_t$ by performing leader compression on $\hat{\cP}_{t/\kappa},$ and then doing additional merging. So, we prove that $\hat{\cP}_{t \cdot \kappa} \sqsubseteq \hat{\cP}_t$.
    By our inductive hypothesis, $\hat{\cP}_{t/\kappa} \sqsupseteq \hat{\cP}_{t \cdot \kappa}$. To generate $\bar{\cP}_t$, we perform rounds of leader compression on $\hat{\cP}_{t/\kappa}$ with respect to $\tilde{G}_t$. But, we know that $\tilde{\cP}_t \sqsupseteq \tilde{\cP}_{t \cdot \kappa} \sqsupseteq \hat{\cP}_{t \cdot \kappa}$ (the last part following from our inductive hypothesis), so $\bar{\cP}_t \sqsupseteq \hat{\cP}_{t \cdot \kappa}$.
    Finally, we only connect components of $\bar{\cP}_t$ to form $\hat{\cP}_t$ if they are incomplete and their leaders are in the same component of $\hat{\cP}_{t \cdot \kappa}$, and since $\bar{\cP}_t \sqsupseteq \hat{\cP}_{t \cdot \kappa}$, this implies that $\hat{\cP}_t \sqsupseteq \hat{\cP}_{t \cdot \kappa},$ as desired.

    Next, by the inductive hypothesis, we have that $\hat{\cP}_{t \cdot \kappa} \sqsupseteq \hat{\cP}_{\alpha^{k+1}}$, and by the base case, every connected component in $\hat{\cP}_{\alpha^{k+1}}$ is contained in a single cell at level $\alpha^{k+2}/\beta.$ Therefore, the same holds for $\hat{\cP}_{t}$, since $\hat{\cP}_{t} \sqsupseteq \hat{\cP}_{t \cdot \kappa} \sqsupseteq \hat{\cP}_{\alpha^{k+1}}$.
    
    To finish the proof, we show that $\hat{\cP}_t \sqsubseteq \tilde{\cP}_t$. Consider any $x, y$ in the same component in $\tilde{\cP}_t$. If they are in the same component in $\bar{\cP}_t$, then they are also in the same component in $\hat{\cP}_t$. Otherwise, $x$ and $y$ are in two different incomplete components in $\bar{\cP}_t$, but are in the same connected component in $\hat{\cP}_{t \cdot \kappa}$, since $\tilde{\cP}_t \sqsupseteq \tilde{\cP}_{t \cdot \kappa} \sqsupseteq \hat{\cP}_{t \cdot \kappa}$ by the inductive hypothesis. Recall that we showed in the first paragraph of the proof that $\bar{\cP}_t \sqsupseteq \hat{\cP}_{t \cdot \kappa}$. Therefore, since $x$ and $y$ are in the same connected comonent in $\hat{\cP}_{t \cdot \kappa}$, so are the leaders of $x$ and $y$ in $\bar{\cP}_t$. Therefore, the components get merged together, so $x$ and $y$ are in the same connected component in $\hat{\cP}_t$. 
\end{proof}

%

%

Next, we prove an probabilistic result about leader compression.

\begin{lemma} \label{lem:leader_compression_main}
    Consider any partitioning $\cP^{(0)}$ of $X$, and a graph $H$. After $h$ rounds of leader compression, $\BE\left[|\cP^{(h)}| - |\cP^{(0)} \oplus H|\right] \le (3/4)^h \cdot (|\cP^{(0)}|-|\cP^{(0)} \oplus H|).$
\end{lemma}

\begin{proof}
    We first prove the lemma for $h = 1$. Fix a connected component $C \in \cP^{(0)} \oplus H$, and suppose that $C$ splits into $C_1, \dots, C_r$ in $\cP^{(0)}$, for some $r \ge 1.$ If $r = 1$, then $C$ still splits into $1$ connected component in $\cP^{(1)}$, since $\cP^{(0)} \oplus H \sqsubseteq \cP^{(1)} \sqsubseteq \cP^{(0)}$. Otherwise, for every $C_j$, there must exist some $C_i$, $i \neq j$, such that there is an edge in $H$ connecting $C_i$ to $C_j$ (or else $C_j$ is not connected to the rest of $C$). Hence, with probability at least $1/4$, $C_j$ has value $0$, $C_i$ has value $1$, and there must be a message sent from some point in $C_i$ to some point in $C_j$. This means that $C_j$ will become merged to some connected component (not necessarily $C_i$) and will have a new head. Hence, if $r > 1$, the value of $r$ multiplies by at most $3/4$ in expectation after a single round of compression. 
    
    Hence, in all cases, the expectation of $r-1$ multiplies by at most $3/4$. However, $|\cP^{(0)}| - |\cP^{(0)} \oplus H|$ precisely equals the sum of $r-1$ across all connected components in $\cP^{(0)} \oplus H$ at the beginning, and $|\cP^{(1)}| - |\cP^{(0)} \oplus H|$ precisely equals the sum of $r-1$ across all connected components in $\cP^{(0)} \oplus H$ at the end. Therefore, $\BE\left[|\cP^{(1)}| - |\cP^{(0)} \oplus H|\right] \le (3/4) \cdot (|\cP^{(0)}|-|\cP^{(0)} \oplus H|).$
    
    For general values of $h$, note that $\cP^{(0)} \oplus H = \cP^{(h-1)} \oplus H.$ Therefore, we have $\BE\left[|\cP^{(h)}| - |\cP^{(0)} \oplus H|\right] \le (3/4) \cdot (|\cP^{(h-1)}|-|\cP^{(0)} \oplus H|),$ which, after an inductive argument, completes the proof.
\end{proof}

\subsection{Part 1 of the Algorithm}

In this subsection, we analyze Part 1 of the algorithm for some fixed $t = \alpha^k$.

First, we recall the following, which is immediate by combining \Cref{prop:bar-P-single-cell} and \Cref{lem:basic-properties-main}.

\begin{proposition} \label{prop:basic-pt-1}
    For every $t$ a power of $\alpha$, $\bar{\cP}_t \sqsupseteq \tilde{\cP}_t \sqsupseteq \hat{\cP}_t$.
\end{proposition}

    %

Recall that by our algorithmic construction (Part 1), and because $\beta > \sqrt{d}$, all points in a cell at level $t/\beta$ must be in the same connected component for any of $\tilde{\cP}_t, \bar{\cP}_t, \hat{\cP}_t$.

\begin{lemma} \label{lem:incomplete_cc_bound_part_1}
    For some fixed $t = \alpha^k$, let $\cP_t^{(h)}$ be the partitioning after $h$ rounds of leader compression on $\cP_t^{(0)}$ with respect to $\tilde{G}_t$.
    Then, the expected number of incomplete connected components in $\bar{\cP}_t$ is at most $(3/4)^h \cdot \alpha \cdot \frac{\MST(\tilde{G})}{t}.$
\end{lemma}

\begin{proof}
    First, note that $\cP_t^{(0)} \oplus \tilde{G}_t = \tilde{\cP}_t$, since $\cP_t^{(0)} \sqsupseteq \tilde{\cP}_t$. Therefore, by \Cref{lem:leader_compression_main}, $\BE[|\cP_t^{(h)}| - |\tilde{\cP}_t|] \le (3/4)^h \cdot (|\cP_t^{(0)}|-|\tilde{\cP}_t|) \le (3/4)^h \cdot (|\cP_t^{(0)}|-1)$. Next, note that $|\cP_t^{(0)}|$ equals the number of distinct nonempty cells in the quadtree $Q$ at side length $t/\beta$. Any points in distinct cells can only be connected by an edge in the spanner of length at least $2 \alpha^{k-1}$. Hence, $\MST(\tilde{G}) \ge 2 \alpha^{k-1} \cdot (|\cP_t^{(0)}|-1),$ which means that $|\cP_t^{(0)}|-1 \le \frac{\MST(\tilde{G})}{2 \alpha^{k-1}}$. Hence, $\BE[|\cP_t^{(h)}| - |\tilde{\cP}_t|] \le (3/4)^h \cdot \frac{\MST(\tilde{G})}{2 \alpha^{k-1}} = (3/4)^h \cdot \alpha \cdot \frac{\MST(\tilde{G})}{2t}$.
    
    Next, we note that the number of incomplete connected components in $\cP_t^{(h)}$ is at most $2 \cdot (|\cP_t^{(h)}| - |\tilde{\cP}_t|)$. To see why, if any connected component $C \in \tilde{\cP}_t$ is split into $r_C$ connected components in $\cP_t^{(h)}$, then $|\cP_t^{(h)}| - |\tilde{\cP}_t| = \sum_{C \in \tilde{\cP}_t} (r_C-1)$, but the number of incomplete connected components is $\sum_{C \in \tilde{\cP}_t} r_C \cdot \textbf{1}(r_C \ge 2)$. Since $2(r_C-1) \ge r_C$ whenever $r_C \ge 2$, this means that $2 \cdot (|\cP_t^{(h)}| - |\tilde{\cP}_t|)$ is at least the number of incomplete connected components. Hence, the expected number of incomplete connected components is at most $(3/4)^h \cdot \alpha \cdot \frac{\MST(\tilde{G})}{t}.$
\end{proof}

Hence, the following holds.

\begin{corollary} \label{cor:overcomplete-bound}
    For some fixed $t = \alpha^k$, $\BE[|\tilde{\cP}_t|-|\hat{\cP}_t|] \le (3/4)^h \cdot \frac{\alpha}{t} \cdot \MST(\tilde{G})$.
\end{corollary}

\begin{proof}
    By \Cref{prop:basic-pt-1}, we have that $|\tilde{\cP}_t| - |\hat{\cP}_t| \le |\bar{\cP}_t| - |\hat{\cP}_t|$. However, to form $\hat{\cP}_t$ from $\bar{\cP}_t$, we only merge incomplete connected components together. So, $|\bar{\cP}_t| - |\hat{\cP}_t|$ is at most the number of incomplete connected components, minus $1$. Hence, by \Cref{lem:incomplete_cc_bound_part_1}, the expectation is at most $(3/4)^h \cdot \frac{\alpha}{t} \cdot \MST(\tilde{G}).$ 
\end{proof}

\subsection{Part 2 of the Algorithm}
In this section, we consider constructing $\hat{\cP}_t$ where $\alpha^{k} < t < \alpha^{k+1}$.

\begin{proposition}
    Let $\cP$ be any partition, and $\cP_1, \cP_2 \sqsubseteq \cP$ be coarser partitions. Then, $|\cP|+|\cP_1 \oplus \cP_2| \ge |\cP_1|+|\cP_2|$.
\end{proposition}

\begin{proof}
    Let $M = |\cP|$. Consider graphs on $[M]$, where each vertex represents a partition in $\cP$. Let $G_1$ be the spanning forest on $[M]$ that generates the partition $\cP_1$ (via the connected components of $G_1$) and $G_2$ be the graph that generates the partition $\cP_2$ on $[M]$. Then, $G_1 \cup G_2$ generates the partition $\cP_1 \oplus \cP_2$. Since $G_1, G_2$ are forests, the number of edges $|G_1|, |G_2|$ equal $M - |\cP_1|, M - |\cP_2|$, respectively. Hence, $|G_1 \cup G_2| \le 2M - |\cP_1|-|\cP_2|$. The number of connected components in $G_1 \cup G_2$ is at least $M - |G_1 \cup G_2|$, with equality if and only if $G_1 \cup G_2$ is also a forest. Therefore, $|\cP_1 \oplus \cP_2| \ge M - (2M - |\cP_1|-|\cP_2|) = |\cP_1|+|\cP_2| - M,$ which completes the proof.
\end{proof}

Define $T := |\tilde{\cP}_{t/\kappa}| - |\hat{\cP}_{t/\kappa}|$.
Since $\tilde{\cP}_t, \hat{\cP}_{t/\kappa} \sqsubseteq \tilde{\cP}_{t/\kappa}$ by \Cref{lem:basic-properties-main}, we have the following corollary.

\begin{corollary} \label{cor:pt-tilde-bound}
    We have that $|\tilde{\cP}_t| - |\hat{\cP}_{t/\kappa} \oplus \tilde{\cP}_t| \le T$.
\end{corollary}

We recall that $\tilde{\cP}_t \sqsupseteq \hat{\cP}_t$.
The rest of this section is devoted to bouding
$\BE[|\tilde{\cP}_t|-|\hat{\cP}_t|]$,
which generalizes \Cref{cor:overcomplete-bound} to $t$ not necessarily a power of $\alpha$.

Next, we bound the number of incomplete connected components, similar to \Cref{lem:incomplete_cc_bound_part_1}.

\begin{lemma} \label{lem:incomplete_cc_bound_part_2}
    For some fixed $\alpha^k < t < \alpha^{k+1}$, the expected number of incomplete connected components in $\bar{\cP}_t$ is at most $(3/4)^h \cdot \frac{\alpha^2}{t} \cdot \MST(\tilde{G})$.
\end{lemma}

\begin{proof}
    Define $\cP_t^{(0)} := \hat{\cP}_{t/\kappa}$ as the starting partition before leader compression. Let $\cP_t^{(h)}$ be the partition after performing $h$ rounds of leader compression on $\cP_t^{(0)}$ with respect to $\tilde{G}_t$. By \Cref{lem:leader_compression_main}, we know that $\BE\left[|\cP_t^{(h)}| - |\cP_t^{(0)} \oplus \tilde{\cP}_t| \big\vert \cP_t^{(0)} \right] \le (3/4)^h \cdot \left(|\cP_t^{(0)}| - |\cP_t^{(0)} \oplus \tilde{\cP}_t|\right)$. Since $\cP_t^{(0)} = \hat{\cP}_{t/\kappa},$ this means that 
\[\BE\left[|\cP_t^{(h)}| - |\hat{\cP}_{t/\kappa} \oplus \tilde{\cP}_t| \big\vert \hat{\cP}_{t/\kappa} \right] \le (3/4)^h \cdot \left(|\hat{\cP}_{t/\kappa}| - |\hat{\cP}_{t/\kappa} \oplus \tilde{\cP}_t|\right).\]
    By Part 3 of \Cref{lem:basic-properties-main}, know that $\hat{\cP}_{t/\kappa} \sqsubseteq \hat{\cP}_{\alpha^k}$, so we can bound $|\hat{\cP}_{t/\kappa}|$ as at most the number of nonempty cells at level $\alpha^k/\beta$. 
    Moreover, $|\hat{\cP}_{t/\kappa} \oplus \tilde{\cP}_t| \ge 1$, which means that $|\hat{\cP}_{t/\kappa}| - |\hat{\cP}_{t/\kappa} \oplus \tilde{\cP}_t|$ is at most the number of nonempty cells at level $\alpha^k/\beta$ minus $1$, which as we saw in \Cref{lem:incomplete_cc_bound_part_1} is at most $\frac{\MST(\tilde{G})}{2 \alpha^{k-1}} \le \alpha^2 \cdot \frac{\MST(\tilde{G})}{2t},$ because $t \le \alpha^{k+1}$. So, by removing the conditioning on $\hat{\cP}_{t/\kappa}$, we have that $\BE\left[|\cP_t^{(h)}| - |\hat{\cP}_{t/\kappa} \oplus \tilde{\cP}_t|\right] \le (3/4)^h \cdot \alpha^2 \cdot \frac{\MST(\tilde{G})}{2t}.$
    
    Next, recall that $\cP_t^{(h)} \sqsupseteq \hat{\cP}_{t/\kappa} \oplus \tilde{\cP}_t$, and every incomplete connected component in $\cP_t^{(h)}$ has an edge leaving it in $\tilde{G}_t$.
    So, if any component $C \in \hat{\cP}_{t/\kappa} \oplus \tilde{\cP}_t$ is split into $r_C$ connected components in $\cP_t^{(h)}$, then $|\cP_t^{(h)}| - |\hat{\cP}_{t/\kappa} \oplus \tilde{\cP}_t| = \sum_{C \in \hat{\cP}_{t/\kappa} \oplus \tilde{\cP}_t} (r_C-1)$, but the number of incomplete connected components is $\sum_{C \in \hat{\cP}_{t/\kappa} \oplus \tilde{\cP}_t} r_C \cdot \textbf{1}(r_C \ge 2)$. Since $2(r_C-1) \ge r_C$ whenever $r_C \ge 2$, this means that 
    the number of incomplete connected components is at most $2 \cdot (|\cP_t^{(h)}| - |\hat{\cP}_{t/\kappa} \oplus \tilde{\cP}_t|).$ Hence, the expected number of incomplete connected components is at most $(3/4)^h \cdot \alpha^2 \cdot \frac{\MST(\tilde{G})}{t}$.
\end{proof}

We next show that $|\tilde{\cP}_t| - |\hat{\cP}_t|$ is small, in expectation.

\begin{lemma} \label{lem:overcomplete-bound-part-2}
    Suppose $t$ is a power of $2$ with $\alpha^k < t < \alpha^{k+1}$, and let $\kappa = 2^{2^{v_2(\log_2 t)}}$. Then, $\BE\left[|\tilde{\cP}_t| - |\hat{\cP}_t|\right] \le \BE\left[|\tilde{\cP}_{t/\kappa}| - |\hat{\cP}_{t/\kappa}|\right] + (3/4)^h \cdot \frac{\alpha^2}{t} \cdot \MST(\tilde{G})$.
\end{lemma}

\begin{proof}
    Since we only merge incomplete connected components together to go from $\bar{\cP}_t = \cP_t^{(h)}$ to $\hat{\cP}_t$, we know that
\begin{equation} \label{eq:b1}
    \BE[|\cP_t^{(h)}| - |\hat{\cP}_t|] \le (3/4)^h \cdot \frac{\alpha^2}{t} \cdot \MST(\tilde{G}),
\end{equation}
by \Cref{lem:incomplete_cc_bound_part_2}.
    Next,
\begin{equation} \label{eq:b2}
    \BE[|\tilde{\cP}_t| - |\hat{\cP}_{t/\kappa} \oplus \tilde{\cP}_t|] \le \BE[|\tilde{\cP}_{t/\kappa}| - |\hat{\cP}_{t/\kappa}|],
\end{equation}
    by \Cref{cor:pt-tilde-bound}.
    Finally, $\cP_t^{(h)} \sqsupseteq \hat{\cP}_{t/\kappa} \oplus \tilde{\cP}_t,$ by \Cref{prop:leader-compression-basic}. Therefore,
\begin{equation} \label{eq:b3}
    \BE[|\hat{\cP}_{t/\kappa} \oplus \tilde{\cP}_t| - |\cP_t^{(h)}|] \le 0.
\end{equation}
    Adding Equations \eqref{eq:b1}, \eqref{eq:b2}, and \eqref{eq:b3} together completes the proof.
\end{proof}

By a simple inductive argument, we have $|\tilde{\cP}_t| - |\hat{\cP}_t|$ is small in expectation. Specifically, we have the following corollary.

\begin{corollary} \label{cor:overcount-main}
    For any $t$ a power of $2$, $\hat{\cP}_t \sqsubseteq \tilde{\cP}_t$, and $\BE\left[|\tilde{\cP}_t| - |\hat{\cP}_t|\right] \le (3/4)^h \cdot \frac{\alpha^3}{t} \cdot \MST(\tilde{G})$.
\end{corollary}

\begin{proof}
    Suppose $\alpha^k \le t < \alpha^{k+1}$. We can construct a sequence $t = t_0 > \cdots > t_r = \alpha^k$, where $t_{i+1} = t_i/2^{2^{v_2(\log_2 t_i)}}$. Then, we can form a telescoping sum
\begin{align*}
    \BE\left[|\tilde{\cP}_t| - |\hat{\cP}_t|\right] &= \BE\left[|\tilde{\cP}_{t_r}| - |\hat{\cP}_{t_r}|\right] + \sum_{i = 0}^{r-1} \left(\BE\left[|\tilde{\cP}_{t_i}| - |\hat{\cP}_{t_i}|\right] - \BE\left[|\tilde{\cP}_{t_{i+1}}| - |\hat{\cP}_{t_{i+1}}|\right]\right).
\intertext{Using \Cref{cor:overcomplete-bound} and \Cref{lem:overcomplete-bound-part-2}, this is at most}
    &\le \left(\frac{3}{4}\right)^h \cdot \frac{\alpha}{t_r} \cdot \MST(\tilde{G}) + \sum_{i = 0}^{r-1} \left(\frac{3}{4}\right)^h \cdot \frac{\alpha^2}{t_i} \cdot \MST(\tilde{G}) \\
    &\le \left(\frac{3}{4}\right)^h \cdot \alpha^2 \cdot \MST(\tilde{G}) \cdot \sum_{i=0}^r \frac{1}{t_i} \\
    &\le \left(\frac{3}{4}\right)^h \cdot \alpha^2 \cdot \MST(\tilde{G}) \cdot \left(\frac{1}{t} + \frac{2}{t} + \frac{4}{t} + \cdots + \frac{1}{\alpha^k}\right) \\
    &\le \left(\frac{3}{4}\right)^h \cdot \alpha^2 \cdot \MST(\tilde{G}) \cdot \frac{2}{\alpha^k} \\
    &\le \left(\frac{3}{4}\right)^h \cdot \frac{\alpha^3}{t} \cdot \MST(\tilde{G}),
\end{align*}
    since $t \le \alpha^{k+1}/2$ so $\frac{2}{\alpha^k} \le \frac{\alpha}{t}$.
\end{proof}

\subsection{Part 3 of the Algorithm}

From \Cref{lem:basic-properties-main}, we have generated a nested set of partitions $\hat{\cP}_1 \sqsupseteq \hat{\cP}_2 \sqsupseteq \hat{\cP}_4 \sqsupseteq \cdots$. In Part 3, we generate the edges: we must show that the cost of the edges we produce is at most $O(1) \cdot \MST(\tilde{G})$, and that the edges we produce generate a spanning tree.

When connecting $\hat{\cP}_{t/2}$ into $\hat{\cP}_t$, we generate a new partition ${\cP_t'}^{(h)}$ after $h$ rounds of leader compression on $\hat{\cP}_{t/2}$ with respect to $\tilde{G}_t$. First, we note that ${\cP_t'}^{(h)}$ still refines $\hat{\cP}_t$.

\begin{proposition} \label{prop:edges-basic}
    We have ${\cP_t'}^{(h)} \sqsupseteq \hat{\cP}_t$.
\end{proposition}
    
\begin{proof}
    By \Cref{prop:leader-compression-basic}, ${\cP_t'}^{(h)} \sqsupseteq \hat{\cP}_{t/2} \oplus \tilde{\cP}_t$. Moreover, by \Cref{lem:basic-properties-main}, $\hat{\cP}_{t/2} \sqsupseteq \hat{\cP}_t$ and $\tilde{\cP}_{t} \sqsupseteq \hat{\cP}_t$, which means that $\hat{\cP}_{t/2} \oplus \tilde{\cP}_t \sqsupseteq \hat{\cP}_t$. In summary, ${\cP_t'}^{(h)} \sqsupseteq \hat{\cP}_{t/2} \oplus \tilde{\cP}_t \sqsupseteq \hat{\cP}_t$.
\end{proof}

Let $\hat{G}_t$ represent the edges generated when connecting $\hat{\cP}_{t/2}$ into $\hat{\cP}_t$ (and $\hat{G}_1 = \emptyset$).
\begin{lemma}
    $\bigcup_{t' \le \alpha^H} \hat{G}_{t'}$ is a spanning tree of $X$.
\end{lemma}

\begin{proof}
     We prove inductively that $\bigcup_{t' \le t} \hat{G}_{t'}$ forms a spanning forest for $\hat{\cP}_t$.

    For the base case $t = 1$, 
    recall that we assume the minimum distance between any two points in $X$ is at least $\alpha^{100}$, which means by Part 1 of \Cref{lem:basic-properties-main}, every connected component in $\hat{\cP}_1$ consists of a single element. Since $\hat{G}_1$ has no edges, the base case follows.

    Define $\hat{G}_{\le t} := \bigcup_{t' \le t} \hat{G}_{t'}$. For the inductive step, assume that $\hat{G}_{\le t/2}$ generates a spanning forest for $\hat{\cP}_{t/2}$. Then, leader compression only adds edges to connect two distinct connected components together, so when adding the edges generated by $h$ rounds of leader compression on $\hat{\cP}_{t/2}$ with respect to $\tilde{G}_t$, we obtain a spanning forest for ${\cP_t'}^{(h)}$, where ${\cP_t'}^{(h)} \sqsupseteq \hat{\cP}_t$ by \Cref{prop:edges-basic}. Finally, we connect all disconnected components in each partition of $\hat{\cP}_t$, which generates a spanning forest for $\hat{\cP}_t$.

    Due to Part 2 of~\Cref{lem:basic-properties-main}, $\hat{\cP}_t$ refines $\tilde{\cP}_t$. 
    Since $\tilde{\cP}_t$ contains only one part $X$ for $t=\alpha^H$, $\bigcup_{t' \le \alpha^H} \hat{G}_{t'}$ is a spanning tree of $X$.
\end{proof}

Now, we fix a value $t$, and consider the edges generated to connect $\hat{\cP}_{t/2}$ into $\hat{\cP}_t$. We will bound the sum of the weights of these edges.

\begin{lemma}\label{lem:total_weight_of_one_level}
    The sum of the weights of all edges in $\hat{G}_t$, in expectation, is at most $t \cdot \left(|\tilde{\cP}_{t/2}| - |\tilde{\cP}_t|\right) + (3/4)^h \cdot \left(\alpha^3 + \frac{2 \alpha^4 \sqrt{d}}{\beta}\right) \cdot \MST(\tilde{G}).$
\end{lemma}

\begin{proof}
    By \Cref{prop:leader-compression-basic}, we have $\hat{\cP}_{t/2} \oplus \tilde{\cP}_t \sqsubseteq {\cP_t'}^{(h)} \sqsubseteq \hat{\cP}_{t/2},$ and by \Cref{lem:leader_compression_main}, $\BE\left[|{\cP_t'}^{(h)}| - |\hat{\cP}_{t/2} \oplus \tilde{\cP}_t|\right] \le (3/4)^h \cdot \left(|\hat{\cP}_{t/2}| - |\hat{\cP}_{t/2} \oplus \tilde{\cP}_t|\right)$.
    As in \Cref{lem:incomplete_cc_bound_part_2},
    we know that 
    $|\hat{\cP}_{t/2}| - |\hat{\cP}_{t/2} \oplus \tilde{\cP}_t| \le \alpha^2 \cdot \frac{\MST(\tilde{G})}{t},$ so $\BE\left[|{\cP_t'}^{(h)}| - |\hat{\cP}_{t/2} \oplus \tilde{\cP}_t|\right] \le (3/4)^h \cdot \alpha^2 \cdot \frac{\MST(\tilde{G})}{t}.$ In addition, $|\hat{\cP}_{t/2} \oplus \tilde{\cP}_t| - |\tilde{\cP}_t| \le 0$, and $\BE\left[|\tilde{\cP}_t| - |\hat{\cP}_t|\right] \le (3/4)^h \cdot \frac{\alpha^3}{t} \cdot \MST(\tilde{G})$ by \Cref{cor:overcount-main}. So, $\BE\left[|{\cP_t'}^{(h)}| - |\hat{\cP}_t|\right] \le (3/4)^h \cdot \frac{2 \alpha^3}{t} \cdot \MST(\tilde{G})$.

    Assume $\alpha^k \le t < \alpha^{k+1}$.
    At level $t$, we produce $|\hat{\cP}_{t/2}| - |{\cP_t'}^{(h)}|$ edges of weight $t$, and then $|{\cP_t'}^{(h)}| - |\hat{\cP}_t|$ edges to combine all of $\hat{\cP}_t$ together. But by Part 1 of \Cref{lem:basic-properties-main}, the diameter of any cell in $\hat{\cP}_t$ is at most $\alpha^{k+1} \cdot \sqrt{d}/\beta,$ so the total weight of all the edges, in expectation, is at most
\begin{align*}
    t \cdot \BE\left[|\hat{\cP}_{t/2}| - |{\cP_t'}^{(h)}|\right] + \frac{\alpha^{k+1} \sqrt{d}}{\beta} \cdot \BE\left[|{\cP_t'}^{(h)}| - |\hat{\cP}_t|\right]
    &\le t \cdot \BE\left[|\hat{\cP}_{t/2}| - |{\cP_t'}^{(h)}|\right] + \frac{\alpha^{k+1} \sqrt{d}}{\beta} \cdot \left(\frac{3}{4}\right)^h \cdot \frac{2 \alpha^3}{t} \cdot \MST(\tilde{G}) \\
    &\le t \cdot \BE\left[|\hat{\cP}_{t/2}| - |{\cP_t'}^{(h)}|\right] + \frac{2 \alpha^{4} \sqrt{d}}{\beta} \cdot \left(\frac{3}{4}\right)^h \cdot \MST(\tilde{G}) \\
    &\le t \cdot \BE\left[|\tilde{\cP}_{t/2}| - |\hat{\cP}_t|\right] + \frac{2 \alpha^{4} \sqrt{d}}{\beta} \cdot \left(\frac{3}{4}\right)^h \cdot \MST(\tilde{G}).
\end{align*}
    Moreover, $\BE[|\tilde{\cP}_t| - |\hat{\cP}_t|] \le (3/4)^h \cdot \frac{\alpha^3}{t} \cdot \MST(\tilde{G})$ by \Cref{cor:overcount-main}, so $\BE[|\tilde{\cP}_{t/2}| - |\hat{\cP}_t|] \le |\tilde{\cP}_{t/2}| - |\tilde{\cP}_t| + (3/4)^h \cdot \frac{\alpha^3}{t} \cdot \MST(\tilde{G})$. So overall, the weight of the edges we added, in expectation is at most
\[t \cdot \left(|\tilde{\cP}_{t/2}| - |\tilde{\cP}_t|\right) + (3/4)^h \cdot \left(\alpha^3 + \frac{2 \alpha^4 \sqrt{d}}{\beta}\right) \cdot \MST(\tilde{G}). \qedhere\]
\end{proof}

    Adding this over all levels $t$, we get this is at most $\MST(\tilde{G}) + (3/4)^h \cdot \log n \cdot \left(\alpha^3 + \frac{2 \alpha^4 \sqrt{d}}{\beta}\right) \cdot \MST(\tilde{G})$. So, we just need $h = O(\log \log n)$.

\begin{lemma}\label{lem:approximation_gaurantee}
If $h$ is a sufficiently large $\Theta(\log\log n)$, $\bigcup_{t' \le \alpha^H} \hat{G}_{t'}$ is an $O(1)$-approximate $\MST$ for $X$ in expecation.
\end{lemma}
\begin{proof}
According to Lemma~\ref{lem:total_weight_of_one_level}, the expected total weights of $\bigcup_{t' \le \alpha^H} \hat{G}_{t'}$ is at most
\begin{align*}
\MST(\tilde{G}) + (3/4)^h \cdot \log n \cdot \left(\alpha^3 + \frac{2 \alpha^4 \sqrt{d}}{\beta}\right) \cdot \MST(\tilde{G}) = O(1)\cdot \MST(\tilde{G}).
\end{align*}
According to Corollary~\ref{cor:MST_spanner_not_too_big} and Proposition~\ref{prop:MST_spanner_not_too_small}, $\bigcup_{t' \le \alpha^H} \hat{G}_{t'}$ is a $O(1)$-approximate MST for $X$ in expectation.
\end{proof}

\section{Euclidean MST in the MPC Model}\label{sec:mpc_mst}
In this section, we describe how to implement our algorithms in the Massively Parallel Computation (MPC) model.
In Section~\ref{sec:existing_mpc_algorithms}, we briefly introduce some useful exsiting algorithmic primitives in the MPC model.
In Section~\ref{sec:MPC_MST_implementation}, we show the detailed implementation of our Euclidean MST algorithm in the MPC model.

\subsection{Existing Algorithmic Primitives in the MPC model}\label{sec:existing_mpc_algorithms}

One of the most basic MPC algorithmic primitive is sorting.
\begin{theorem}[\cite{goodrich1999communication,goodrich2011sorting}]\label{thm:sorting}
Given size $N$ input data where each data item has size at most $N^{o(1)}$, there is a fully scalable MPC algorithm to sort and index (rank) these data items in $O(1)$ rounds and using $O(N)$ total space.
\end{theorem}
Sorting can be used to build other MPC subroutines. 
One important use case of sorting is to simulate a PRAM algorithm in the MPC model.
The PRAM model is a classic model of parallel computation. 
Roughly speaking, there is a shared memory and multiple processors in the PRAM model.
The computation proceeds in rounds, and each processor does at most one algorithmic operation in one round.
In addition, processors can simutaneously read and write the shared memory cells in a round.
The efficiency of a PRAM algorithm is measured by \emph{depth} (longest chain of dependencies of the computation) and \emph{work} (total running time over processors).
A CRCW (Concurrent Read Concurrent Write) PRAM algorithm means that processors can simutaneously read and write the same shared memory cell in a round.
It was shown that any PRAM algorithm can be simulated in the MPC model (see e.g.~\cite{goodrich2011sorting,andoni2018parallel}).
We refer readers to Appendix E of~\cite{andoni2018parallel} for detailed implementations.
We formally state the gaurantees of the simulation as the following.
\begin{theorem}[\cite{goodrich2011sorting,andoni2018parallel}]\label{thm:pram_simulation}
Given a CRCW PRAM algorithm with $O(D)$ depth and $O(W)$ work, it can be simulated by a fully scalable MPC algorithm with $O(D)$ rounds and $O(W)$ space.
\end{theorem}
One application of Theorem~\ref{thm:pram_simulation} is the following simultaneous access problem: several arrays with total size $O(m)$ are stored distributedly on the machines, and $p$ machines want to simultaneously access some entries of these arrays where each machine has $s$ non-adaptive queries.
This operation can be done in $O(1)$ rounds and total space $O(m+p\cdot s)$, and it is fully scalable.
One application of simultaneous access is duplicating and broadcasting a small size message.
\begin{theorem}[See e.g., \cite{andoni2018parallel}]\label{thm:broadcasting}
Given a data item with size $N^{o(1)}$ where $N$ is the target output size, there is a fully scalable MPC algorithm\footnote{Here the ``fully scalable'' means that the algorithm works for local memory size $N^{\delta}$ for any constant $\delta > 0$.} uses $O(1)$ rounds and $O(N)$ space that duplicates the data item such that the total size over all duplications is $N$.
\end{theorem}

Another important application of sorting is to duplicate the input data small number of times and rearange them. 
This allows us to run multiple tasks which requires the same input data in parallel.
\begin{theorem}[See e.g., \cite{andoni2018parallel}]\label{thm:duplication}
Given size $N$ input data and a parameter $m=N^{o(1)}$, there is a fully scalable algorithm which duplicates the input data $m$ times.
The algorithm takes $O(1)$ rounds and $O(N\cdot m)$ total space.
\end{theorem}

The following theorem gives a subroutine of applying dimensionality reduction efficiently in the MPC model.
Note that each input data point is not necessarily small enough to fit into the memory of a single machine and the entries of each data point can distributed arbitrarily over machines.
\begin{theorem}[Theorem 3 of \cite{ahanchi2023massively}]\label{thm:dim_reduce}
Given a set of points $x_1,x_2,\cdots,x_n\in\mathbb{R}^d$ whose entries are distributed arbitrarily over the machines in the MPC model, there is a fully scalable MPC algorithm which outputs $x_1',x_2',\cdots,x_n'\in\mathbb{R}^{d'}$ where $d'= O(\log n)$ and with probability at least $1-1/\poly(n)$, $\forall i,j\in[n],0.99 \|x_i-x_j\|_2^2\leq \|x'_i-x'_j\|_2^2\leq 1.01 \|x_i-x_j\|_2^2$.
In addition, $\forall i\in[n],x'_i$ is held entirely by a single machine.
The algorithm takes $O(1)$ rounds, and the total space needed is $O(nd + n\log^3 n)$. 
\end{theorem}

The following theorem states that there is an efficient MPC algorithm which computes a constant approximate Euclidean spanner.
Note that a graph can be represented by a set of edges, and we store the edges in an arbitrarily distributed way over machines in the MPC model.

\begin{theorem}[See e.g.,~\cite{epasto2022massively}]\label{thm:twohopspanner}
Given parameters $C>1, t>0$ and a set of points $X=\{x_1,x_2,\cdots,x_n\}\subset \mathbb{R}^d$ distributed arbitrarily on several machines in the MPC model with $d=O(\log n)$, there is a fully scalable MPC algorithm which outputs a graph $G=(X, E)$ such that with probability at least $1-1/\poly(n)$:
\begin{enumerate}
\item For any $x, y \in X$, if there is an edge in $G$ between $x$ and $y$, $\|x-y\|_2\leq C\cdot t$.
\item For any $x, y \in X$ with $\|x-y\|_2\leq t$, there is a path in $G$ connecting $x, y$ using at most $2$ edges.
\end{enumerate}
In addition, $G$ has at most $n^{1+1/C^2+o(1)}$ number of edges.
The algorithm uses $O(1)$ rounds and needs total space $n^{1+1/C^2+o(1)}$.
\end{theorem}

Another algorithmic primitive in the MPC model is to find predecessors.
In particular, given pairs $(x_1,y_1),(x_2,y_2),\cdots,(x_n,y_n)$ where $x_i$ have a total ordering and $y_i\in \{0,1\}$. 
The total size of all pairs is $N$ and each $x_i$ has size $N^{o(1)}$.
The goal is to output $(x_1,i_1),(x_2,i_2),\cdots,(x_n,i_n)$ where $\forall j\in [n], i_j$ satisfies that (1) $y_{i_j}=1$, (2) $x_{i_j}$ is the largest one such that $x_{i_j}\leq x_j$.

\begin{theorem}[\cite{goodrich2011sorting,andoni2018parallel}]\label{thm:predecessor}
There is a fully scalable MPC algorithm that solves the predecessor problem in $O(1)$ rounds and $O(N)$ total space.
\end{theorem}

We also state an MPC algorithmic primitive to index elements in sets.

\begin{theorem}[\cite{andoni2018parallel}]\label{thm:index_in_sets}
Given sets $S_1,S_2,\cdots,S_k$ of items with total size $N$, there is a fully scalable MPC algorithm which runs in $O(1)$ rounds and uses total space $O(N)$ to output the size of each set $S_i$, and for each $x\in S_i$, outputs its ranking / index within the set $S_i$.
\end{theorem}

Similarly, the above primitive can be easily extended to compute a prefix sum for each element in each set.

\begin{theorem}[\cite{andoni2018parallel}]\label{thm:prefix_sum_in_sets}
Given sets $S_1,S_2,\cdots,S_k$ of ranked items with total size $N$, where each item has a weight, there is a fully scalable MPC algorithm which runs in $O(1)$ rounds and uses total space $O(N)$ to output a prefix sum $\sum_{j\in S_l:\text{item }j\text{ has smaller rank than item }i} (\text{weight of }j)$ for each item $i\in S_l$ for each $S_l,l\in[k]$.
\end{theorem}

\subsection{Implementation of Euclidean MST in the MPC Model}\label{sec:MPC_MST_implementation}
In this section, we provide details of the MPC implementation of our Euclidean MST algorithm. 

\begin{theorem}[Constant approximate Euclidean MST in the MPC model]\label{thm:mpc_mst}
Given $n$ points from $\mathbb{R}^d$, there is a fully scalable MPC algorithm which outputs an $O(1)$-approximate MST with probability at least $0.99$.
The number of rounds of the algorithm is $O(\log\log(n)\cdot \log\log\log(n))$.
The total space required is at most $O(nd+n^{1+\varepsilon})$ where $\varepsilon>0$ is an arbitrary small constant.
\end{theorem}

\begin{proof}
The approximation gaurantee is shown by Lemma~\ref{lem:approximation_gaurantee} in Section~\ref{sec:offline_analysis}.
In the remaining of the proof, we show how to implement the algorithm described in Section~\ref{sec:offline_algorithm} in the MPC model.
\paragraph{Preprocessing.}
We firstly need to run dimensionality reduction.
By applying the algorithm mentioned in Theorem~\ref{thm:dim_reduce}, we are able to reduce the dimension of each point to $d'=\Theta(\log n)$ while preserve all pairwise distances up to a constant factor with probability at least $0.999$.
This step requires $O(1)$ rounds and $O(nd+n\log^3n)$ total space.
In addition, the algorithm is fully scalable.
Then we can use sorting (Theorem~\ref{thm:sorting}) to sort all values in each dimension separtely and parallelly to learn the maximum and the minimum values in each dimension.
Then according to Theorem~\ref{thm:broadcasting}, we are able to use $O(1)$ rounds to make every machine learn the maximum and the minimum value of each dimension.
Thus, we are able to shift the point set such that the minimum value in each dimension is $0$.
Note that base on the maximum value and minimum value of each coordinate, we can verify whether all points are at the same location. 
If all points are at the same location, the MST cost is $0$ and we can output an arbitrary spanning tree.
Otherwise we linearly scale the entire dataset such that the maximum value of the point set $\Delta = \Theta(n^2/\log n)$.

Next, we need to build the randomly shifted grid and round the coordinates of each point.
To do this, we let one machine generate a randomly shifted vector $(a_1,a_2,\cdots,a_{d'})\in \mathbb{R}^{d'}$ where each coordinate is drawn uniformly from $[0,\Delta)$ (this is the same vector as described in Section~\ref{sec:offline_algorithm}).
Since $d'=\Theta(\log n)$, we apply Theorem~\ref{thm:broadcasting} again, we can use $O(1)$ rounds to make every machine learns $(a_1,a_2,\cdots,a_{d'})$.
Therefore, all machines agree on the same randomly shifted grid.
Each machine checks each point stored in its memory, and moves the point to the closest center of a cell at level $t=\alpha^{100}$.
Note that each movement might change the cost of MST by $\sqrt{d'}\cdot \alpha^{100}=\polylog n$, and because we have $n$ points, the total change will be at most $n\polylog n$.
However, since the MST cost is at least the maximum distance which is at least $\Delta = \Theta(n^2/\log n)$, the MST cost is changed by at most a constant factor by above movements.
After doing above steps, our point set satifies that the maximum distance is at most $n^2$ and the minimum distance is at least $\alpha^{100}$ which are the same as mentioned in Section~\ref{sec:offline_algorithm}.

Now we describe how to generate $\tilde{G}_t$ for $t=1,2,4,\cdots,\alpha^H$.
Firstly, we apply Theorem~\ref{thm:duplication} to duplicate input data $\log(\alpha^H)$ times.
This step is fully scalable and takes $O(1)$ rounds and $O(n\log n)\cdot \log(\alpha^H)$ total space.
Since we duplicated the input data, we can handle each $t=1,2,4,\cdots,\alpha^H$ in parallel.
Consider a fixed $t$ and let $k$ be such $\alpha^{k-1}<t\leq \alpha^k$.
As discussed above, all machines agree on the same randomly shifted grid, thus each cell at level $\alpha^{k+1}/\beta$ can be uniquely determined by a vector in $\mathbb{Z}^{d'}$ (see the discussion in Section~\ref{sec:offline_algorithm}).
Let $c(x)$ be the cell at level $\alpha^{k+1}/\beta$ containing $x$.
We create $(c(x),x)$ for each point $x$.
These steps only require local computations.
Let $\varepsilon>0$ be an arbitray constant.
Then we sort all $(c(x),x)$ (Theorem~\ref{thm:sorting}) but we restrict that each machine only uses memory $\Theta(s^{1/(1+\varepsilon/2)})$ where $s$ is the local memory size of each machine.
The algorithm is fully scalable, takes $O(1)$ rounds and $O(n^{1+\varepsilon/2}\log n)$ total space.
Then the points $X_c$ in the same cell $c$ are rearranged into consecutive machines.
We set $C$ used in Theorem~\ref{thm:twohopspanner} to be $\xi/\sqrt{\varepsilon}$ where $\xi\geq 1$ is a sufficiently large constant.
If $X_c$ is entirely contained by a single machine, then $\tilde{G}'_t(X_c)$ can be computed locally and with probability at least $1-1/\poly(n)$, 
\begin{enumerate}
    \item For any $p,q\in X_c$, if there is an edge in $\tilde{G}'_t(X_c)$ between $p$ and $q$, $\|p-q\|_2\leq C\cdot t$.
\item For any $p,q\in X_c$ with $\|p-q\|_2\leq t$, there is a path in $\tilde{G}'_t(X_c)$ connecting $p,q$ using at most $2$ edges.
\end{enumerate}
Since $\tilde{G}'_t(X_c)$ only has $O(|X_c|^{1+\varepsilon/2})$ edges, $\tilde{G}'_t(X_c)$ can still be stored locally (note each machine only consumes $s^{1+\varepsilon/2}$ space during the sorting).
Otherwise, if $X_c$ is distributed over multiple machines, then only the first machine containing points from $X_c$ and the last machine containing points from $X_c$ might contain points that are not in $X_c$.
Then there are two cases: 
(1) If $X_c$ is only on two consecutive machines, we can send them into an arbitrary machine and handle them locally as discussed above. 
(2) If $X_c$ is on more than two consecutive machines, the first machine containing points in $X_c$ can move them to the second machine containing points in $X_c$, and the last machine containing points in $X_c$ can move them to the second to the last machine that contains points in $X_c$.
By applying the predecessor algorithm (Theorem~\ref{thm:predecessor}), if a machine holds points in $X_c$, it learns the first machine contains points in $X_c$ as well as the last machine that contains points in $X_c$.
Finding predecessor processor requires $O(1)$ rounds and $O(n\log n)$ space, and it is fully scalable in the MPC model.
Then we are able to apply Theorem~\ref{thm:twohopspanner} for each $X_c$ in parallel. 
Therefore, we obtain graphs $\tilde{G}'_t(X_c)$ for all $X_c$ in $O(1)$ rounds and the total space required is $O(n^{1+\varepsilon/2})$.
This step is also fully scalable.
In addition, with probability at least $1-1/\poly(n)$, for every $X_c$, $G'_t(X_c)$ satisfies the properties mentioned in Theorem~\ref{thm:twohopspanner} with respect to $X_c$.
Then we let $G'_t$ be the union of all $G'_t(X_c)$.
This step can be done by only local computation, i.e., renaming the graph.

Then, we duplicate $\tilde{G}'_1,\tilde{G}'_2,\tilde{G}'_4,\cdots, \tilde{G}'_{\alpha^H}$ $\log(\alpha^H)$ times.
According to Theorem~\ref{thm:duplication}, this is fully scalable and only takes $O(1)$ rounds and $O(n^{1+\varepsilon/2}\log n)$ total space.
For each edge of the $i$-th copy of $\tilde{G}'_t$, if $t\leq 2^i$, we add this edge into $\tilde{G}_{2^i}$.
Otherwise, we delete the edge.
Therefore, we obtain that $\tilde{G}_t = \bigcup_{t'\leq t} \tilde{G}'_{t'}$, and thus $\tilde{G}_1\subseteq \tilde{G}_2\subseteq\tilde{G}_4\subseteq\cdots\subseteq \tilde{G}_{\alpha^H}$.
The above step can be done by renaming the edges which only requires local computation.
Finally, we apply sorting (Theorem~\ref{thm:sorting}) again to remove duplicated edges in each $\tilde{G}_t$.

To conclude, the entire preprocessing stage takes $O(1)$ rounds and uses $O(nd+n^{1+\varepsilon/2}\log n)$ total space. 
In addition, these steps are fully scalable in the MPC model.

\paragraph{Leader compression.}
Now we describe how to implement leader compression algorithm (Algorithm~\ref{alg:leader-compression}) in the MPC model.
We firstly describe the format of the input data of of leader compression stored in the system.
Note that the partition $\cP$ is implied by the leader mapping function $\ell$: $\forall x,y\in X$, $x$ and $y$ are in the same part if and only if $\ell(x)=\ell(y)$.
Therefore, we store $\ell$ in the system to denote $\cP$.
In particular we have $|X|$ tuples $(x,\ell(x))$ distributed on the machines.
The graph $H$ is represented a set of edges, and the edges $(x,y)$ are also distributed on the machines.
The boolean value $\text{Edges}$ is known by all machines (this can be done by Theorem~\ref{thm:broadcasting} for broadcasting). 

We firstly need to compute the set of leaders $S$.
This step can be done by looking at each tuple $(x,\ell(x))$.
If $\ell(x)=x$, we generate the random bit $b_x$ for the leader node $x$ (e.g., create a tuple $(``b",x,b_x)$).
Above steps only require local computation.
The next step is to let $b_x\gets b_{\ell(x)}$ for every $x\in X$.
We can regard it as a simultaneous access problem in the PRAM model.
In particular, the vector $b$ is stored in the shared memory.
Each processor corresponds to a tuple $(x,\ell(x))$. 
If $x\not=\ell(x)$, then the processor reads the value of $b_{\ell(x)}$ and writes the value of $b_x$.
This operation has $O(1)$ depth and $O(|X|)$ work in the PRAM model.
Therefore, according to the simulation algorithm (Theorem~\ref{thm:pram_simulation}), we can use $O(1)$ rounds and total space $O(n)$ to compute $b_x$ for all $x\in X$.
In addition, this step is fully scalable.
Then the next step is to find all edges $(x,y)\in H$ such that $b_x=1$ and $b_y=0$, and for each such edge, create a tuple $(\ell(y),(x,y,\ell(x)))$.
This step can also be regarded as a simultaneous access problem in the PRAM model.
In particular, the vector $b$ and the vector $\ell(\cdot)$ is stored in the shared memory.
Each processor corresponds to each edge $(x,y)\in H$.
The process that all processors simulateously read corresponding $b_x,b_y,\ell(x),\ell(y)$ has $O(1)$ depth and $O(|H|+|X|)$ work.
Therefore, by applying the simulation algorithm again (Theorem~\ref{thm:pram_simulation}), we can use $O(1)$ rounds and $O(|H|+|X|)$ total space in the MPC model to create all tuples $(\ell(y),(x,y,\ell(x)))$ for edges $(x,y)\in H$ whose $b_x=1,b_y=0$.
This step is also fully scalable.
Then for each $z=\ell(y)$, if there are some tuples $(z,(\cdot, \cdot,\cdot))$ created, we only keep an arbitrary one $(\ell(y),(x,y,\ell(x)))$.
This deduplication step can be done by sorting (Theorem~\ref{thm:sorting}).
Thus it is fully scalable, takes $O(1)$ rounds and uses total space at most $O(|H|)$.
For every tuple $(\ell(y),(x,y,\ell(x)))$ that is kept, if $\text{Edges}$ is true, we add edge $(x,y)$ into the output edge set $E$.
This operation can be done by local computation only (e.g., create a tuple $(``E",(x,y))$).
Finally, we need to update $\ell(y)$ for all $y\in X$.
In particular, if there is a tuple $(\ell(y),(x,y,\ell(x)))$ kept, we need to update $\ell(y)\gets \ell(x)$. 
Otherwise we keep $\ell(y)$ unchanged.
This operation can also be seen as a simultaneous access in the PRAM model. 
In particular, we have a size $|X|$ array stored in the shared memory, if there is a tuple $(\ell(y),(x,y,\ell(x)))$ kept, then the entry of the array with index $\ell(y)$ has value $(x,y,\ell(x))$. 
Otherwise, the entry of the array with index $\ell(y)$ has value empty.
Each processor corresponds to each tuple $(y,\ell(y))$.
Each processor simutaneously accesses the array and see whether the entry with index $\ell(y)$ is empty or not.
If the entry is empty, the processor keeps $(y,\ell(y))$ unchanged.
Otherwise the entry is $(x,y,\ell(x))$, the processor updates $(y,\ell(y))$ to $(y,\ell(x))$.
These PRAM operations take $O(1)$ depth and $O(|X|)$ total work.
Therefore, according to the simulation theorem again (Theorem~\ref{thm:pram_simulation}), the above steps can also be done using $O(1)$ rounds and $O(|X|)$ total space, and the algorithm is fully scalable in the MPC model.

To conclude, the one leader compression step can be implemented in the MPC model using $O(1)$ rounds and $O(|X|+|H|)$ total space. 
In addition, the algorithm is fully scalable.

\paragraph{Part 1 of the algorithm.}
We consider how to implement Algorithm~\ref{alg:part-1} in the MPC model.
Notice that we will run Algorithm~\ref{alg:part-1} for $t=\alpha,\alpha^2,\alpha^3,\cdots,\alpha^H$ simutenously in parallel.
Firstly, we can apply Theorem~\ref{thm:duplication} to duplicate the point set $H$ times. 
This operation is fully scalable and takes $O(1)$ rounds and $O(H\cdot n\log n)$ total space.
Note that the graphs $\tilde{G}_\alpha,\tilde{G}_{\alpha^2},\tilde{G}_{\alpha^3},\cdots,\tilde{G}_{\alpha^H}$ are already computed and stored distributedly in the system.
We can use sorting (Theorem~\ref{thm:sorting}) to send each $\tilde{G}_t$ and a duplicated copy of the input point set to a group of machines.
Thus we can handle $t=\alpha,\alpha^2,\alpha^3,\cdots,\alpha^H$ simutenously in parallel.
This operation is fully scalable and takes $O(1)$ rounds and $O(n^{1+\varepsilon/2}\log n)$ total space.

Next, we focus on how to implement Algorithm~\ref{alg:part-1} for a fixed $t$.
The first step is to generate partition $\cP^{(0)}_t$.
As we discussed before, we only need to generate the leader mapping function $\ell_t^{(0)}(\cdot)$ such that $x$ and $y$ are in the same part of $\cP^{(0)}_t$ if and only if $\ell_t^{(0)}(x) = \ell_t^{(0)}(y)$.
Recall the definition $\cP^{(0)}_t$, two points $x,y$ are in the same part if and only if they are in the same cell in the grid of level $t/\beta$.
As discussed in the preprocessing step, all machines agree on the same randomly shifted grid, thus each cell at level $t/\beta$ can be uniquely determined by a vector in $\mathbb{Z}^{d'}$ (see the discussion in Section~\ref{sec:offline_algorithm}).
Let $c(x)\in \mathbb{Z}^{d'}$  denote the cell which contains point $x$ in level $t/\beta$.
Then we can compute $(c(x), x)$ for every point $x\in X$.
This step only requires local computation.
We can sort (Theorem~\ref{thm:sorting}) all $(c(x),x)$ such that each point $x$ learns whether it is the first point in the cell $c(x)$.
For each $(c(x),x),$ we create a tuple $((c(x),x),1)$ if $x$ is the first point in the cell $c(x)$ and set $\ell_t^{(0)}(x) = x$.
Otherwise we create a tuple $((c(x),x),0)$
Then, by applying the predecessor algorithm (Theorem~\ref{thm:predecessor}), each point $x$ learns the index of the point $y$ which is the first point within the same cell $c(x)=c(y)$.
By using the index of the point $y$ and simulating the simultaneous access operation in PRAM (Theorem~\ref{thm:pram_simulation}), each point $x$ learns $y$ which is the first point in the same cell.
Then we set $\ell_t^{(0)}(x)=y$.
The above operations only require $O(1)$ rounds and $O(n\log n)$ total space.

Once we obtained $\ell^{(0)}_t$, we are able to run $h=O(\log\log n)$ rounds of leader compression process.
As we discussed previously, running $h$ rounds of leader comression takes $O(\log\log n)$ rounds and $O(|X|+|\tilde{G}_t|)$ space.
The algorithm is fully scalable as well.
For each $(x,y)\in\tilde{G}_t$, we want to access $\ell^{(h)}_t(x)$ and $\ell^{(h)}_t(y)$, this again can be regarded as a simultaneous access problem in the PRAM model where $\ell^{(h)}_t$ is stored in the shared memory and each processor corresponds to each edge $(x,y)\in\tilde{G}_t$.
Each processor requires to read $\ell^{(h)}_t(x)$ and $\ell^{(h)}_t(y)$, and if they are not equal, the processor writes $\textsc{Incomplete}$ to the $\ell^{(h)}_t(x)$-th entry and the $\ell^{(h)}_t(y)$-th entry of an array in the shared memory.
According to the simulation theorem (Theorem~\ref{thm:pram_simulation}), these steps can be done in $O(1)$ rounds in the MPC model and $O(|X|+|\tilde{G}_t|)$ space, and it is fully scalable.
Then for each point $x$, we can check whether it is a leader, i.e., $\ell^{(h)}_t(x)=x$ and whether it is marked as \textsc{Incomplete}.
This operation can be done by simulating simultaneous access in PRAM (Theorem~\ref{thm:pram_simulation}) which takes $O(1)$ rounds and $O(n)$ total space, and is fully scalable.
Let $I$ be the set of all incomplete leaders.
Then for each $x\in I$, we can locally compute $(c(x),x)$ where $c(x)\in\mathbb{Z}^{d'}$ indicates the cell at the level $\alpha^{k+1}/\beta$ containing $x$.
Then we can sort all such $(c(x),x)$, and each point in $I$ learns whether it is the first point among $I$ and in the cell $c(x)$.
By following the similar approach of computing $\ell_t^{(0)}$, we compute $\ell_t^{(h+1)}(x)=y$ for each $x\in I$ where $y\in I$ is the first point in the cell $c(x)$.
Similar as computing $\ell_t^{(0)}$, this step takes $O(1)$ MPC rounds and requires $O(n\log n)$ total space, and it is fully scalable.
For each leader which is not \textsc{Incomplete}, i.e., $x\not\in I$, we set $\ell^{(h+1)}(x)=x$.
This can be done by simulating simultaneous access in PRAM (Theorem~\ref{thm:pram_simulation}).
Thus, it is fully scalable, has $O(1)$ rounds, and uses space $O(n)$.
Finally, for each tuple $(x,\ell^{(h)}_t(x))$, we access $\ell^{(h+1)}_t(\ell^{(h)}_t(x))$ and create a new tuple $(x,\ell^{(h+1)}_t(\ell^{(h)}_t(x)))$ to indicate $\ell^{(h+1)}_t(x)\gets \ell^{(h+1)}_t(\ell^{(h)}_t(x))$.
This can also be done by simulating simultaneous access in PRAM (Theorem~\ref{thm:pram_simulation}).
Thus, it is fully scalable, has $O(1)$ rounds, and uses space $O(n)$.
The output $\hat{\cP}_t$ is represented by $\ell_t(\cdot)\equiv\ell^{(h+1)}_t(\cdot)$.

To conclude, Part 1 of the algorithm (Algorithm~\ref{alg:part-1}) is fully scalable.
It requires $O(\log\log n)$ rounds and requires total space $O(n^{1+\varepsilon/2}\log n)$.

\paragraph{Part 2 of the algorithm.}
We consider how to implement Algorithm~\ref{alg:part-2} for all $t=1,2,4,8,\cdots,\alpha^H$.
Firstly, we can duplicate our input data $\log(\alpha^H)$ times (Theorem~\ref{thm:duplication}) which takes $O(1)$ rounds and $O(n\log n \cdot \log(\alpha^H))$ total space, and the algorithm is fully scalable.
Then we can use disjoint set of machines to handle each $t$ separately.
Note that $\tilde{G}_t$ are already stored in the system due to the preprocessing stage.
We also send $\tilde{G}_t$ to the group of machines which will be used to compute $\hat{\cP}_t$.
This is also fully scalable and takes only $O(1)$ rounds and total space $O(n^{1+\varepsilon/2}\log n)$.

As described in Section~\ref{sec:offline_algorithm}, instead of computing all $t$ at the same time, we handle $t$ in descreasing order of $v_2(\log t)$ (recall Definition~\ref{def:adic_valuation} for $v_2(\cdot)$).
In particular, we already obtained $\hat{\cP}_t$ for all $t=1,\alpha,\alpha^2,\cdots,\alpha^H$ by running Part 1 of the algorithm.
More precisely, we obtained $\ell_t(\cdot)$ for these $t$.
In the first iteration, $\kappa=\sqrt{\alpha}$, we will handle $t=\sqrt{\alpha},\alpha^{1.5},\cdots,\alpha^{H-0.5}$ at the same time in parallel, i.e., we will have $\hat{\cP}_t$ for all $t=1,\alpha^{0.5},\alpha,\alpha^{1.5},\cdots,\alpha^H$ at the end of the first iteration.
In the second iteration, $\kappa=\alpha^{0.25}$, we will have $\hat{\cP}_t$ for all $t=1,\alpha^{0.25},\alpha^{0.5},\alpha,\cdots,\alpha^H$ at the end of the second iteration.
In the $i$-th iteration we will have $\kappa=\alpha^{1/2^i}$, and we will have $\hat{\cP}_t$ for all $t=\kappa^0,\kappa,\kappa^2,\cdots,\alpha^H$ at the end of the $i$-th iteration.
Thus, we will have $O(\log\log(\alpha))=O(\log\log\log(n))$ iterations in total.
Note that if $\hat{\cP}_t$ is computed in the $i$-th iteration, then $\forall j\geq 1$, $\hat{\cP}_{t}$ (more precisly, $\ell_t(\cdot)$) will be the input for computing $\hat{\cP}_{t\cdot \alpha^{1/2^{i+j}}}$ and $\hat{\cP}_{t/ \alpha^{1/2^{i+j}}}$.
Therefore, when all $\ell_t(\cdot)$ are computed after the $i$-th iteration, we can use duplication process (Theorem~\ref{thm:duplication}) to duplicate all $\ell_t(\cdot)$ at most $O(\log\log\log(n))$ times and send each duplicated copy of $\hat{\cP}_t$ (more precisely, $\ell_t(\cdot)$) to the groups of machines that will be used to compute $\hat{\cP}_{t\cdot \alpha^{1/2^j}}$ or $\hat{\cP}_{t/ \alpha^{1/2^j}}$ for $j\geq 1$.

When a group of machines recieves all required inputs $X, \tilde{G}_t,\ell_{t/\kappa}(\cdot),\ell_{t\cdot \kappa}(\cdot)$, it starts to run Algorithm~\ref{alg:part-2} to compute $\hat{\cP}_t$ (i.e., $\ell_t(\cdot)$).
We firstly intialize $\ell_t^{(0)}(\cdot)$ to be $\ell_{t/\kappa}(\cdot)$.
Then, similar as Part 1 of the algorithm, we run $O(\log\log n)$ rounds of leader compression.
This step takes $O(\log\log n)$ MPC rounds and $O(|X|+|\tilde{G}_t|)$ space.
This step is fully scalable.
Then by using the same process described in how to implement Part 1 of the algorithm, we can simulate PRAM operations (Theorem~\ref{thm:pram_simulation}) to find all $x\in X$ that are marked as \textsc{Incomplete}.
For each leader $x$ ($\ell_t^{(h)}(x)=x$) which is marked as \textsc{Incomplete}, let $\ell_{t}^{(h+1)}(x)=\ell_{t\cdot \kappa}(x)$.
For each leader $x$ ($\ell_t^{(h)}(x)=x$) which is not marked as \textsc{Incomplete}, let $\ell_{t}^{(h+1)}(x)=x$.
Finally, for each $x\in X$, let $\ell_{t}^{(h+1)}(x)\gets \ell_{t}^{(h+1)}(\ell_{t}^{(h)}(x))$.
Similar as before, above operations used to compute $\ell_{t}^{(h+1)}(\cdot)$ can be seen as simultaneous accesses in the PRAM model which requires $O(1)$ depth and $O(n)$ total work.
By applying PRAM simulation again (Theorem~\ref{thm:pram_simulation}), the above procedure to compute $\ell_t^{(h+1)}(\cdot)$ requires $O(1)$ MPC rounds and $O(n)$ total space, and it is fully scalable.
Note that the output $\hat{\cP}_t$ is represented by $\ell_t(\cdot)\equiv \ell_{t}^{(h+1)}(\cdot)$.

To conclude, to compute $\hat{\cP}_t$ (more preciesly, $\ell_t(\cdot)$) for all $t=1,2,4,\cdots,\alpha^H$, our algorithm takes $O(\log\log(n)\cdot\log\log\log(n))$ rounds.
The total space required is at most $n^{1+\varepsilon/2}\cdot \polylog(n)$.
In addition, the algorithm is fully scalable.

\paragraph{Part 3 of the algorithm.}
Finally, let us consider how to implement Algorithm~\ref{alg:part-3} in parallel for all $t$ in the MPC model.

Similar as before, we use a disjoint group of machines to compute $\textsc{Part3}(X,t,\tilde{G}_{t},\hat{\cP}_{t/2},\hat{\cP}_t)$ for each $t$.
To do this, we need to duplicate $X$ $\log(\alpha^H)$ times and duplicate $\hat{\cP}_t$ (i.e., $\ell_t(\cdot)$) $2$ times ($\hat{\cP}_t$ is used to compute $\textsc{Part3}(X,t,\tilde{G}_{t},\hat{\cP}_{t/2},\hat{\cP}_t)$ and $\textsc{Part3}(X,2t,\tilde{G}_{2t},\hat{\cP}_{t},\hat{\cP}_{2t})$).
Then we send each copy of the data and $\tilde{G}_t$ to the corresponding group of machines.
According to Theorem~\ref{thm:duplication}, above process can be done using $O(1)$ rounds and $n^{1+\varepsilon/2}\polylog(n)$ space, and these steps are fully scalable.

In the following, we describe how to implement Algorithm~\ref{alg:part-3} for any fixed $t$.
We initialize ${\ell'}^{(0)}_t(\cdot)$ to be $\ell_{t/2}(\cdot)$.
This only requires local computation.
Then, as we discussed previously, $O(\log\log n)$ rounds of leader compression can be implemented in $O(\log\log n)$ MPC rounds and $O(n^{1+\varepsilon/2}\log (n))$ total space, and the process is fully scalable.
The $i$-th round of leader compression outputs a set of edges $E_t^{(i)}$.
At the end of the last leader compression process, let the obtained leader mapping be ${\ell'}_t^{(h)}(\cdot)$.
For each $x\in X$, we check whether ${\ell'}_t^{(h)}(x)=\ell_t(x)$, if not, we create a set $S_{\ell_t(x)}$ and add ${\ell'}_t^{(h)}(x)$ into the set $S_{\ell_t(x)}$.
Note that this operation can be regarded as simultaneous access of ${\ell'}_t^{(h)}(\cdot)$ and $\ell_t(\cdot)$ under the PRAM model, which only has $O(1)$ depth and $O(n)$ work. 
Therefore, according to the simulation theorem (Theorem~\ref{thm:pram_simulation}), the computation of all $S_{\ell_t(x)}$ only requires $O(1)$ MPC rounds and $O(n)$ total space, and it is fully scalable.
Then, for each set $S_c$ if $z$ is in $S_c$, we add an edge $\{z,c\}$ into $F_t$.
Thus $F_t$ will be a forest of stars.
Then we output $F_t\cup \bigcup_{i=1}^h E_t^{(i)}$.

To conclude, to compute $F_t\cup \bigcup_{i=1}^h E_t^{(i)}$ for all $t=1,2,4,\cdots,\alpha^H$, our algorithm takes $O(\log\log(n))$ rounds.
The total space required is at most $n^{1+\varepsilon/2}\cdot \polylog(n)$.
In addition, the algorithm is fully scalable.

\paragraph{Put them together.}
All parts of our algorithm are fully scalable.
Thus, the overall algorithm is fully scalable as well.
The preprocessing stage takes $O(1)$ rounds and $O(nd+n^{1+\varepsilon/2}\log n)$ total space.
Running Algorithm~\ref{alg:part-1} for all $t=1,\alpha,\alpha^2,\cdots,\alpha^H$ takes $O(\log\log(n))$ rounds and $O(nd+n^{1+\varepsilon/2}\log n)$ total space.
Running Algorithm~\ref{alg:part-2} for all remaining $t$ takes $O(\log\log(n)\cdot \log\log\log(n))$ rounds and $n^{1+\varepsilon/2}\cdot\polylog n$ total space.
Running Algorithm~\ref{alg:part-3} for all $t$ takes $O(\log\log(n))$ rounds and $n^{1+\varepsilon/2}\cdot\polylog n$ total space.
Therefore, the entire algorithm takes $O(\log\log(n)\log\log\log(n))$ rounds and $O(nd+n^{1+\varepsilon})$ total space.
\end{proof}
\section{Euler Tour of Approximate Euclidean MST in MPC}\label{sec:tsp}
In this section, we show how to compute an Euler tour of the approximate MST that we obtained.
Note that the actual diameter of the tree that we obtained can be very large (for example, the tree that we obtained can be a path which has diameter $\Theta(n)$).
Therefore, we cannot use the algorithm of \cite{andoni2018parallel} which requires $\Omega(\log(\text{diameter}))$ rounds.
We propose a new method.
In particular, given the approximate MST and the hierarchy of clusters that we used to obtain the tree, we are able to use linear total space and $O(1)$ MPC rounds to compute an Euler tour of the tree, and our algorithm is fully scalable.

\subsection{Additional Existing Algorithmic Primitives in the MPC Model}

A standard way to store a sequence $A=(a_1,a_2,\cdots,a_n)$ in the MPC model is that we store tuples $(1,a_1),(2,a_2),(3,a_3),\cdots,(n,a_n)$ arbitrarily on the machines in a distributed manner (see e.g.,~\cite{andoni2018parallel}).
In this way, one can easily use sorting (Theorem~\ref{thm:sorting}) to reorder the elements such that the sequence is stored in consecutive machines and each machine stores corresponding consecutive elements in order. Similarly, a standard way to store a mapping $f:[n]\rightarrow [m]$ in the MPC model is that we store tuples $(1,f(1)),(2,f(2)),\cdots,(n,f(n))$.

The sequence insertion problem is stated as the following: 
Given $k+1$ sequences  $A = (a_1, a_2, \cdots , a_m),$ $A_1, A_2, \cdots , A_k$ and a mapping $f:[k]\rightarrow \{0\}\cup[m]$, the goal is to output a sequence $A'$ which is obtained by inserting every $A_i,(i\in[k])$ into $A$ such that $A_i$ is between the element $a_{f(i)}$ and $a_{f(i)+1}$.
Let the total size of the input sequences is $|A|+|A_1|+|A_2|+\cdots+|A_k|=N$.
The size of $A'$ is also $N$, and the space to store the input mapping $f$ is $O(k)=O(N)$ as well.

\begin{theorem}[\cite{andoni2018parallel}]\label{thm:sequence_insertion}
There is a fully scalable MPC algorithm solving the sequence insertion problem in $O(1)$ rounds and $O(N)$ space.
\end{theorem}

\begin{theorem}[\cite{andoni2018parallel}]\label{thm:euler_tour_low_diameter}
Given a tree of $n$ nodes with diameter at most $\Lambda$, there is a fully scalable which takes $O(\log(\Lambda))$ rounds and $O(n^{1+\varepsilon})$ total space to output an Euler tour (see Definition~\ref{def:euler_tour}) of the tree, where $\varepsilon>0$ is an arbitrarily small constant.
\end{theorem}

\subsection{Join Euler Tours of Spanning Tree of Sub-clusters}
Given a tree, we treat each edge $\{x,y\}$ as two directed edges $(x,y)$ and $(y,x)$.
An Euler tour of a tree is a traverse sequence of tree edges such that each directed edge $(x,y)$ appeared in the sequence exactly once.
The following gives a formal definition of an Euler tour of a tree.
\begin{definition}[Euler tour of a tree]\label{def:euler_tour}
Given a tree $T=(V,E)$ with $|V|=n$ nodes, the Euler tour of $T$ is a sequence $((v_1,v_2),(v_2,v_3),\cdots,(v_{2n-2},v_1))$ of length $2n-2$ where $\forall i\in[2n-2], \{v_{i},v_{(i\bmod (2n-2)) + 1}\}\in E$, and $\forall \{x,y\}\in E$, there exists exactly one $i\in [2n-2]$ and exactly one $i'\in[2n-2]$ such that $(v_i,v_{(i\bmod (2n-2)) + 1})=(x,y)$ and $(v_{i'},v_{(i'\bmod (2n-2)) + 1})=(y,x)$.
\end{definition}

\subsubsection{Properties of An Euler Tour of a Tree}
In this section, we give some formal definitions for Euler tour of a tree.
Most observations stated in this section can also be found in the textbooks of algorithms such as~\cite{cormen2022introduction}.

\begin{definition}[Parent pointers]
Given a rooted tree $T$ over $V$, we use $\p:V\rightarrow V$ to denote a set of parent pointers which is a mapping from each node to its parent, and we set $\p(v)=v$ for the root $v$ of $T$.
\end{definition}
Note that $\p(\cdot)$ contains the information of all tree edges and the root that we select.
\begin{definition}[Size of subtree]
Given parent points $\p(\cdot)$ of tree with root $v$, let $\size_{\p}(u)$ denote the number of nodes in the subtree rooted at $u$.
If $\p(\cdot)$ is clear in the context, we omit the subscript $\p$ and use $\size(u)$ for short.
\end{definition}

\begin{definition}[Ordering of children]\label{def:order_children}
Given a rooted tree represented by parent pointers $\p$, and an ordering of children of each node, we use $\child_{\p}(v)$ to denote the set of children of $v$, $\child_{\p}(v,i),(i\in[|\child_{\p}(v)|])$ to denote the $i$-th child of $v$, and $\rank_{\p}(v)$ to denote its rank among its siblings, i.e., $\rank_{\p}(v)$ satisfies $\child_{\p}(\p(v),\rank_{\p}(v))=v$ if $v$ is not a root and $\rank_{\p}(v)=1$ otherwise.
\end{definition}
If $\p$ is clear in the context, we use $\child(\cdot)$ and $\rank(\cdot)$ without subscript $\p$ for short.

\begin{observation}[Euler tour for different root]\label{obs:circular_shift}
Let $A=((v_1,v_2),(v_2,v_3),(v_3,v_4),\cdots,(v_{2n-2},v_1))$ be any Euler tour of a tree $T$, then any circular shift of $A$, $((v_j,v_{j+1}),(v_{j+1},v_{j+2}),\cdots,(v_{2n-2},v_1),(v_1,v_2),(v_2,v_3),\cdots,(v_{j-1},v_j))$ for $j\in[2n-2]$, is still an Euler tour of $T$.
\end{observation}

\begin{definition}
If $A=((v_1,v_2),(v_2,v_3),(v_3,v_4),\cdots,(v_{2n-2},v_1))$ is an Euler tour of a tree $T=(V,E)$ and $v_1=u\in V$, then $A$ is an Euler tour with respect to the root $u$.
\end{definition}

\begin{definition}
Let $A=(e_1,e_2,\cdots,e_{2n-2})$ be an Euler tour of a tree $T$ where $e_i=(v_i,v_{i+1})$.
Given any node $v$, we say that the first appearance of $v$ is at $\first_A(v)=i$ of $A$ if $i$ is the smallest value such that $v_i=v$, and we say that the last appearance of $v$ is at $\last_A(v)=i$ if $i$ is the largest value such that $v_{i+1}=v$.
If the Euler tour is clear in the context, we omit the subscript $A$ of $\first$ and $\last$.
\end{definition}

\begin{observation}[Parent via Euler tour]\label{obs:parent_by_euler_tour}
Let $A=((v_1,v_2),(v_2,v_3),(v_3,v_4),\cdots,(v_{2n-2},v_1))$ be an Euler tour of a tree $T$ with root $v$.
If $u$ is not a root, then $v_{\first(u)-1}$ is the parent of $u$.
\end{observation}

\begin{observation}[Euler tour of subtree]\label{obs:subtree_size}
Let $A=(e_1,e_2,\cdots,e_{2n-2})$ be an Euler tour of a tree $T$ with respect to the root $u$.
Let $v$ be any node in tree $T$, let $i$ be the first appearance of $v$ in $A$, and let $j$ be the last appearance of $v$ in $A$.
Then $(e_i,e_{i+1},\cdots,e_j)$ is an Euler tour of the subtree rooted at $v$ (note that $j=i-1$ when $v$ is a leaf and thus the tour is an empty sequence), and thus $(j-i+1)/2+1$ is the size of the subtree rooted at $v$.
\end{observation}

\begin{lemma}[Euler tour from the ordering of children]\label{lem:euler_tour_by_ordered_chlildren}
Given a tree $T$ rooted at $v$ with parent pointers $\p$, and given a mapping $\child_{\p}(\cdot,\cdot)$ (Definition~\ref{def:order_children}), then we are able to construct an Euler tour $A$ with respect to the root $v$ such that for any node $u$ in the tree, we have $\first(\child(u,1))<\first(\child(u,2))<\first(\child(u,3))<\cdots<\first(\child(u,|\child(u)|))$.
In addition, for any non-root node $u$, let the path $((x_1,x_2),(x_2,x_3),\cdots,(x_{k-1},x_k))$ be the unique simple path from the root $v$ to $u$ (i.e., $x_1=v,x_k=u$), then the position of the (directed) edge $(\p(u),u)$ in $A$ is:
\begin{align*}
p=\sum_{l=2}^k \left(1+\sum_{i=1}^{\rank(x_l)-1}\size(\child(x_{l-1},i))\right),
\end{align*}
and the position of the (directed) edge $(u,\p(u))$ is $p+2\cdot (\size(x_k)) -1$
\end{lemma}
\begin{proof}
The proof is by induction and construction.
If $T$ only has one node, then $\child(v)=\emptyset$ and we just output $A$ to be an empty sequence.

Suppose the statement is true for any tree whose depth is at most $l-1$, now we consider how to construct $A$ for a tree with depth $l$.
Suppose $T$ has root $v$ and $\forall i\in[|\child(v)|],\child(v,i)=v_i$.
We construct $A$ to be following:
\begin{align*}
(v,v_1), A_1, (v_1,v), (v, v_2), A_2, (v_2, v),\cdots,(v,v_{|\child(v)|}), A_{|\child(v)|}, (v_{|\child(v)|},v),
\end{align*}
where $A_i$ is an Euler tour of the subtree rooted at $v_i$ via our induction hypothesis.
Then, if $u$ is in a subtree of $v_i$, then we still have $\first(\child(u,1))<\first(\child(u,2))<\first(\child(u,3))<\cdots<\first(\child(u,|\child(u)|))$ by our induction hypothesis.
Otherwise, if $u=v$, it is clear that $\first(v_1)<\first(v_2)<\first(v_3)<\cdots<\first(v_{|\child(v)}|)$.
In the remaining of the proof, we need to determine the position of an edge.
If an edge is $(v,v_i)$, then its position is clearly $1+\sum_{j=1}^{i-1}2\cdot\size(v_j)$.
If an edge is $(v_i,v)$, then its position is $\sum_{j=1}^{i}2\cdot\size(v_j)$.
Consider an edge $(\p(u),u)$ is in the subtree rooted at $v_i$.
Let $((x_1,x_2),(x_2,x_3),\cdots,(x_{k-1},x_k))$ be the path from $v$ to $u$, i.e., $x_1=v,x_2=v_i,x_{k-1}=\p(u),x_k=u$.
Then by our induction hypothesis, the position of $(\p(u),u)$ is
\begin{align*}
&1+\left(\sum_{j=1}^{\rank(v_i)-1}2\cdot\size(v_j)\right)+\sum_{l=3}^k \left(1+\sum_{j=1}^{\rank(x_l)-1}2\cdot\size(\child(x_{l-1},j))\right)\\
&=\sum_{l=2}^k \left(1+\sum_{j=1}^{\rank(x_l)-1}2\cdot\size(\child(x_{l-1},j))\right).
\end{align*}
Similarly, the position of $(u,\p(u))$ is 
\begin{align*}
&1+\left(\sum_{j=1}^{\rank(v_i)-1}2\cdot\size(v_j)\right)+\sum_{l=3}^k \left(1+\sum_{j=1}^{\rank(x_l)-1}2\cdot\size(\child(x_{l-1},j))\right)+2\cdot\size(x_k) - 1\\
&=\sum_{l=2}^k \left(1+\sum_{j=1}^{\rank(x_l)-1}2\cdot\size(\child(x_{l-1},j))\right)+2\cdot\size(x_k) - 1.
\end{align*}
Therefore, we conclude the proof of the induction.
\end{proof}

\begin{observation}[Euler tour and a path from the root]\label{obs:path_euler_tour}
Let $A=(e_1,e_2,\cdots,e_{2n-2})$ be an Euler tour of a tree $T$ with respect to the root $u$.
Consider any node $v$. 
If $i$ is the first appearance of $v$, and $B=((x_1,x_2),(x_2,x_3),\cdots,(x_{k-1},x_k))$ is a simple directed path from $u$ to $v$ (i.e., $x_1=u,x_k=v$), then we have:
\begin{enumerate}
\item $\forall j\in [k-1]$, $(x_{j},x_{j+1})$ appears in $(e_1,\cdots,e_{i-1})$ in $A$, and $(x_{j+1},x_j)$ appears in $(e_i,\cdots,e_{2n-2})$.
\item If $\{p,q\}$ is a tree edge but neither $(p,q)$ nor $(q,p)$ appears in the path $B$, then both $(p,q)$ and $(q,p)$ appear simultaneously in $(e_1,\cdots,e_{i-1})$ or appear simultaneously in $(e_i,\cdots,e_{2n-2})$.
\end{enumerate}
\end{observation}

\begin{corollary}[Total weight on the path via Euler tour]\label{cor:weight_on_the_path}
Let $A=(e_1,e_2,\cdots,e_{2n-2})$ be an Euler tour of a tree $T$ with respect to the root $u$.
Let $w$ be a weight function such that every edge $\{p,q\}$ has a weight $w(\{p,q\})$ in the tree.
Let $w'$ be the weight function of directed edge such that $\forall \text{edge }\{p,q\}$, if $p=\p(q)$, then $w'(p,q)=w(\{p,q\})$ and $w'(q,p)=-(w\{p,q\})$.
Let $B=((x_1,x_2),(x_2,x_3),\cdots,(x_{k-1},x_k))$ be a simple directed path from $u$ to $v$, then the total edge weight on the path $\sum_{i=2}^k w(\{x_{i-1},x_{i}\})=\sum_{j=1}^{\first(x_k)-1} w'(e_j)$, i.e., the prefix sum over $A$ with weights $w'$.
\end{corollary}
\begin{proof}
Due to Observation~\ref{obs:path_euler_tour}, for the prefix of $A$, if an edge is on the path,  we will count its edge weight once, otherwise, the non-path edge weight is either cancelled since both directions appeared, or none of the directions is appeared so we did not count it.
\end{proof}

\subsubsection{MPC Algorithms for Trees via Euler Tour}

\begin{lemma}[Change root of the Euler tour]\label{lem:change_root}
Given an arbitrary Euler tour $A$ of a tree $T$ over $n$ nodes $V$, and given any node $v\in V$, there is a fully scalable MPC algorithm which outputs an Euler tour $A'$ of tree $T$ with respect to the root $v$ in $O(1)$ rounds and $O(n)$ total space.
\end{lemma}
\begin{proof}
We use sorting (Theorem~\ref{thm:sorting}) to find the first appearance of $v$ in $A=\{e_1,e_2,\cdots,e_{2n-2}\}$, and send $i=\first_A(v)$ to all machines by broadcasting algorithm (Theorem~\ref{thm:broadcasting}).
These steps only take $O(1)$ rounds and $O(n)$ total space, and they are fully scalable.
Then compute the circular shift $A'=(e_i,e_{i+1},\cdots,e_{2n-2},e_1,\cdots,e_{i-1})$ of $A$.
Since every machine learns $\first_A(v)$, computing the circular shift only requires local computation.
According to Observation~\ref{obs:circular_shift}, the circular shift $A'$ is also an Euler tour of $T$.
Since $A'$ starts from node $v$, $A'$ is an Euler tour of $T$ with respects to the root $v$.
\end{proof}

\begin{lemma}[Subtree sizes]\label{lem:subtree_sizes}
Given an arbitrary Euler tour $A$ of a tree $T$ over $n$ nodes $V$ with root $v$, there is a fully scalable MPC algorithm which outputs the size of each subtree $u$ in $O(1)$ rounds and $O(n)$ total space.
\end{lemma}
\begin{proof}
We suppose $A$ is an Euler tour of $T$ with respect to the root $v$.
Otherwise, we can apply Lemma~\ref{lem:change_root} to make $A$ satisfiy above condition, and the operation is fully scalable, takes $O(1)$ rounds and $O(n)$ space.

We use sorting (Theorem~\ref{thm:sorting}) to compute $\first(u)$ and $\last(u)$ for each $u\in V$.
It is fully scalable and only takes $O(1)$ rounds and $O(n)$ total space.
Then, for each $u$ in parallel, we simultaneously query both $\first(u)$ and $\last(u)$, which can be done in the fully scalable setting by Theorem~\ref{thm:pram_simulation} in $O(1)$ rounds and $O(n)$ total space.
According to Observation~\ref{obs:subtree_size}, $\size(u)=(\last(u)-\first(u)+1)/2+1$.
\end{proof}

\begin{lemma}[Euler tour in consistent with ordering of children]\label{lem:mpc_euler_by_children}
Given a tree $T$ rooted at $v$ over $n$ nodes $V$ with parent pointers $\p$, and given a mapping $\child_{\p}(\cdot,\cdot)$ (Definition~\ref{def:order_children}), if we also have an arbitrary Euler tour $A'$ of $T$, there is an MPC algorithm which constructs an Euler tour $A$ with respect to the root $v$ such that for any node $u$ in the tree, it satisfies $\first(\child(u,1))<\first(\child(u,2))<\first(\child(u,3))<\cdots<\first(\child(u,|\child(u)|))$.
In addition, for any non-root node $u$, let the path $((x_1,x_2),(x_2,x_3),\cdots,(x_{k-1},x_k))$ be the unique simple path from the root $v$ to $u$ (i.e., $x_1=v,x_k=u$), then the position of the (directed) edge $(\p(u),u)$ in $A$ is:
\begin{align*}
p=\sum_{l=2}^k \left(1+\sum_{i=1}^{\rank(x_l)-1}\size(\child(x_{l-1},i))\right),
\end{align*}
and the position of the (directed) edge $(u,\p(u))$ is $p+2\cdot (\size(x_k)) -1$.
In addition, the algorithm is fully scalable, and it takes $O(1)$ rounds and $O(n)$ total space.
\end{lemma}
\begin{proof}
Note that the formulation of the position of each directed edge $(\p(u),u),(u,\p(u))$ is given, and Lemma~\ref{lem:euler_tour_by_ordered_chlildren} already showed that the resulting sequence is indeed an Euler tour with stated properties.
Therefore, our goal is to provide an MPC algorithm to efficiently compute the position of each directed edge in the resulting Euler tour $A$.

We can assume $A'$ is an Euler tour of $T$ with respect to the root $v$.
Otherwise, we apply Lemma~\ref{lem:change_root} to make $A'$ be with respect to the root $v$ which takes $O(1)$ rounds, $O(n)$ total space in the fully scalable setting.
We apply Lemma~\ref{lem:subtree_sizes} on $A'$ to compute the size of each subtree in $O(1)$ rounds, $O(n)$ total space in the fully scalable setting.
Next, for each node $u$, we want to compute $w(\{u,\p(u)\})=1+\sum_{i=1}^{\rank(u)-1}\size(\child(\p(u),i))$ for each non-root $u$.
To achieve this, for each mapping $y=\child(x,i)$ in parallel, we create $\rank(y)=i$, add $y$ into a set $S_x$, and simultaneously access the value $\size(y)$.
The simultaneous access operation can be done by Theorem~\ref{thm:pram_simulation}.
Then we apply the subset prefix sum algorithm (Theorem~\ref{thm:prefix_sum_in_sets}) over $\{S_x\}$ and thus each non-root $u$ learns the value $\sum_{i=1}^{\rank(u)-1}\size(\child(\p(u),i))$, and it sets the edge weight $w(\{u,\p(u)\})=1+\sum_{i=1}^{\rank(u)-1}\size(\child(\p(u),i))$.
The overall number of rounds needed to compute $w(\cdot)$ is $O(1)$, and it requires $O(n)$ total space, and is fully scalable.

Finally, we want to compute the position of each directed edge $(\p(u),u)$, $(u,\p(u))$ in the final sequence $A$.
Note that the position of $(\p(u),u)$ is $\sum_{x\in V:x\not=v\text{ and } x\text{ is on the path from the root }v\text{ to }u}w(x)$.
Suppose $A'=(e'_1,e'_2,\cdots,e'_{2n-2})$.
If $e'_i=(v'_i,v'_{i+1})$ satisfies $v'_i=\p(v'_{i+1})$, let $w'_i=w(\{v'_i,v'_{i+1}\})$, otherwise, let $w'_i=-w(\{v'_i,v'_{i+1}\})$.
Then according to Corollary~\ref{cor:weight_on_the_path}, the position of $(\p(u),u)$ is a prefix sum $\sum_{i=1}^{\first_{A'}(u)} w'_i$.
Note that computing $w'$ only requires simultaneous accesses which can be done in $O(1)$ rounds, $O(n)$ total space, and in fully scalable setting according to Theorem~\ref{thm:pram_simulation}.
Then computing prefix sum $\sum_{i=1}^j w'_i$ can be done using Theorem~\ref{thm:prefix_sum_in_sets}.
It takes $O(1)$ rounds, $O(n)$ total space, and is fully scalable.
Note that the position of $(u,\p(u))$ is the position of $(\p(u),u)$ plus $2\cdot\size(u)-1$.
Since $\size(u)$ is already computed, all computations of positions $(\p(u),u)$ can be done in parallel with an additional call of simultaneous access (Theorem~\ref{thm:pram_simulation}).

Therefore, the overall algorithm only takes $O(1)$ rounds, $O(n)$ total space.
The algorithm is fully scalable.
\end{proof}

\subsubsection{Euler Tour Join}

Now consider the following problem.
We define the problem of joining Euler tours as following:
\begin{itemize}
    \item \textbf{Inputs:}
    \begin{enumerate}
        \item A partition $\cC=\{C_1,C_2,\cdots,C_k\}$ of a node set $V$, represented by a leader mapping $\ell:V\rightarrow V$ where two nodes $u,v$ are in the same component iff $\ell(u)=\ell(v)$.
        Each component $C_i$ is represented by its leader node, i.e., $u\in C_i$ and $\ell(u)=u$. 
        \item A spanning tree $T=(\cC, E)$ over components $C_1,C_2,\cdots,C_k$ where each node in $T$ corresponds to a component $C_i$. 
        An Euler tour $A$ of $T$ is also given.
        \item An edge mapping $g,$ where $\forall e=\{C_i,C_j\}\in E$, $g(e)=\{x,y\}\subseteq V$ and $x\in C_i,y\in C_j$, we sometimes also abuse the notation and use directed version $g(C_i,C_j)=(x,y)$ and $g(C_j,C_i)=(y,x)$.
        We use $g^{-1}$ to indicate the inverse mapping.
        \item A spanning tree $T_i=(C_i,E_i)$ for each component $i\in[k]$, and an Euler tour $A_i$ of $T_i$.
    \end{enumerate}
    \item \textbf{Outputs:}
    \begin{enumerate}
        \item A spanning tree $T'=\left(V,\{g(e)\mid e\in E\}\cup\bigcup_{i\in[k]}E_i\right)$ over $V$, and an Euler tour $A'$ of $T'$.
    \end{enumerate}
\end{itemize}
In Algorithm~\ref{alg:euler_join}, we show a novel MPC algorithm that solves the Euler tour join problem efficiently.

\begin{algorithm}[tb]
   \caption{$\textsc{EulerTourJoin}(\cC, T,A, g, \{T_i\}, \{A_i\})$: Generating a spanning tree of $V$ and the Euler tour of the spanning tree.}
   \label{alg:euler_join}
\begin{algorithmic}[1]
\small
    \STATE \textbf{Inputs:} $\cC=\{C_1,C_2,\cdots,C_k\}$ is a partition over a node set $V$, $T=(\cC, E)$ is a spanning tree over $\cC$, $A$ is an arbitrary Euler tour of $T$, $g$ is a mapping such that $\forall e=\{C_i,C_j\}\in E, g(e)=\{x,y\}$ where $x\in C_i, y\in C_j$, and $T_i=(C_i,E_i)$ is a spanning tree of $C_i$ and $A_i$ is an Euler tour of $T_i$.
    \STATE Let $\hat{E}=\{g(e)\mid e\in E\}$. 
    Let $T'=(V,\hat{E}\cup\bigcup_{i\in [k]}E_i)$ be the spanning tree over $V$.
    \STATE Let $\hat{V} = \{x \in V \mid \exists y, \{x,y\}\in \hat{E}\}$.
    \STATE Choose an arbitrary root $C\in\cC$ for $T$, and shift $A$ to make it be an Euler tour of $T$ with respect to the root $C$.
    \STATE For each $C_i,i\in[k]$, compute its parent $\p_T(C_i)$ in $T$ when root is $C$. 
    \STATE For each $i\in[k]$, let $(x,y)=g(C_i,\p_T(C_i))$, and
    circular shift $A_i$ to make $A_i$ be an Euler tour of $T_i$ with root $x$.
    If $C_i$ is a root, we circular shift $A_i$ to make $A_i$ be an Euler tour of $T_i$ with an arbitrary node from $\hat{V}$ as a root.
    \STATE For each $i\in [k]$, sort all children of $C_i$ in $T$: 
    Consider two children $C_j$ and $C_q$ of $C_i$, if $(x,y)=g(C_i,C_j),(x',y')=g(C_i,C_q)$ and $\first_{A_i}(x)<\first_{A_i}(x')$, then we regard that $C_j$ has smaller rank than $C_q$ among children of $C_i$. If $x=x'$, we can given an arbitrary ordering for $C_j$ and $C_q$.
    \STATE Compute a new Euler tour $\bar{A}$ for $T$ with respect to the root $C$, satisfying that $\forall C_i\in\cC$, $\first_{\bar{A}}(\child( C_i,1))<\first_{\bar{A}}(\child( C_i,2))<\cdots<\first_{\bar{A}}(\child( C_i,|\child( C_i)|))$.
    \STATE Compute $\hat{A}$: for each $(C_p,C_q)$ appeared in the sequence of $\bar{A}$, replace it with $g(C_p,C_q)$.
    \STATE Compute $A_x$ for each $x\in \hat{V}$: Suppose $x$ is in $C_i$, and $A_i=((v_1,v_2),(v_2,v_3),\cdots,(v_{2\cdot|C_i|-2},v_1))$, let $A_x=((v_{\first_{A_i}(x)},v_{\first_{A_i}(x)+1}),(v_{\first_{A_i}(x)+1},v_{\first_{A_i}(x)+2}),\cdots,(v_{j},v_{(j\bmod (2\cdot|C_i|-2))+1}))$ where $j> \first_{A_i}(x)$ is the smallest value such that either $j=2\cdot|C_i|-2$ or $j=\first_{A_i}(x')$ for some $x'\in \hat{V}$.
    \STATE Compute a mapping $f:\hat{V}\rightarrow [2\cdot |\mathcal{C}|-2]$.
    Suppose $\hat{A}=((v'_1,v'_2),(v'_2,v'_3),\cdots,(v'_{2\cdot |\mathcal{C}|-2},v'_1))$.
    For each $x\in \hat{V}$, compute $f(x)=j$ where $j\in [2\cdot |\mathcal{C}|-2]$ is the largest value such that $v'_{(j\bmod(2\cdot |\cC|-2 ))+1}=x$.
 \STATE Compute a sequence $A'$ by plugging $A_x$ for each $x\in \hat{V}$ into $\hat{A}$, where $A_x$ should be inserted directly after $(v'_{f(x)},v'_{(f(x)\bmod(2\cdot |\cC|-2 ))+1})$ in $\hat{A}$.
 \STATE \textbf{Output:} $T'$ and its Euler tour $A'$.
\end{algorithmic}
\end{algorithm}

\begin{lemma}[Correctness of Algorithm~\ref{alg:euler_join}]\label{lem:euler_join_correctness}
The output $A'$ of Algorithm~\ref{alg:euler_join} is an Euler tour of the output $T'$, and $T'$ is a spanning tree of $V$, where the edges of $T'$ is the union of edges in $T_1,T_2,\cdots,T_k$ and edges $\{g(e)\mid e\in E\}$.
\end{lemma}
\begin{proof}
Since each $T_i$ is a spanning tree of $C_i$, and $T$ is a spanning tree of $\{C_1,C_2,\cdots,C_k\}$ and $\hat{E}$ are inter-cluster edges with respect to $T$, we know that $T'$ is a spanning tree of $\bigcup_{i\in[k]}C_i = V$, and the edges of $T'$ is union of edges in $T_1,T_2,\cdots,T_k$ and $\hat{E}=\{g(e)\mid e\in E\}$.

Since $\bar{A}$ is an Euler tour of $T$ which means that both directions of each edge in $\hat{E}$ must appear in $\hat{A}$, which means that every node in $\hat{V}$ also appears in $\hat{A}$.
According to our construction of $A'$, both directions of every edge in $T'$ must appear in $A'$ and $|A'|=|A|+\sum_{x\in\hat{V}}|A_x|=|A|+\sum_{i\in[k]]}|A_i|=2|V|-2$ which implies that each direction of each edge in $T'$ also appeared exactly once.

In the remaining of the proof, we only need to show that $A'$ is indeed a cycle.
Consider an arbitrary $(u,v)$ appeared in $\hat{A}$.
By our construction of $\hat{A}$, $\{u,v\}$ must be an inter-cluster edge.
There are several cases.
Let $u$ be in the cluster $C_i$ and $v$ be in the cluster $C_j$.
\paragraph{Case 1.1: $C_j=\p_{T}(C_i)$ and there is some $(p,q)$ after $(u,v)$ in $\hat{A}$ satisfying $q=v$.} 
Suppose $C_i=\child_{\p_T}(C_j,l)$ for some $l$.
If there is some $(p,q)$ after $(u,v)$ in $\hat{A}$ satisfying $q=v$, then there must be a $C_r$ which contains $p$, and $\p_T(C_r)=C_j$, and $C_r=\child_{\p_T}(C_j,l')$ for some $l'>l$.
This means that $(x,y)=g(C_j,\child_{\p_T}(C_j,l+1))$ satisfies that $x=v$.
Since in $\bar{A}$, $(C_j,\child_{\p_T}(C_j,l+1))$ directly follows $(C_i,C_j)$, $(x,y)$ directly follows $(u,v)$ in $\hat{A}$, and thus $(x,y)$ directly follows $(u,v)$ in $A'$ as well.
Since $x=v$, $(u,v),(x,y)$ is a connected path in $A'$.

\paragraph{Case 1.2: $C_j=\p_{T}(C_i)$ and there is no $(p,q)$ after $(u,v)$ in $\hat{A}$ satisfying $q=v$.}
Suppose $C_i$ is the last child of $C_j$. 
Then $A_v$ must be a suffix of $A_j$.
If $C_j$ is the root of $T$, then $A_v$ is a suffix of $A'$, and since $A'$ starts with a node which is also the first node in $A_j$, and $A_j$ itself is a cycle, we know that it is valid to put $A_v$ at the end of $A'$.
Otherwise, the edge follow $(C_i,C_j)$ in $\bar{A}$ is $(C_j,\p_T(C_j))$
Let $(x,y)=g(\p_T(C_j),C_j)$.
Since we circularly shifted $A_j$ to make $A_j$ be an Euler tour of $T_j$ with respect to the root $y$, it means that the first entry in $A_j$ should be $(y,\cdot)$ (i.e., started from $y$), and the last entry in $A_j$ should be $(\cdot,y)$ (i.e., ended with $y$).
Since $(C_j,\p_T(C_j))$ follows $(C_i,C_j)$ in $\bar{A}$, we have $(y,x)$ follows $(u,v)$ in $\hat{A}$.
By our construction of $A'$, we inserted $A_v$ between $(u,v)$ and $(y,x)$.
Since $A_v$ starts from $v$ (by our construction) and ends in $y$ (since $A_v$ is a suffix of $A_j$), $(u,v),A_v,(y,x)$ is a connected path in $A'$.

Suppose $C_i$ is not the last child of $C_j$ and $C_i=\child_{\p_T}(C_j,l)$ for some $l$.
Then the edge that follows $(C_i,C_j)$ in $\bar{A}$ must be $(C_j,\child_{\p_T}(C_j,l+1))$.
Let $(x,y)=g(C_j,\child_{\p_T}(C_j,l+1))$.
By our construction of $A_v$, we know that $A_v$ must end at some node in $\hat{V}$.
In the following, we show that $A_v$ must end in $x$.
Firstly, we have $\first_{A_j}(x)>\first_{A_j}(v)$ since this is how we used to sort all children of $C_j$ in $T$.
Then, if $A_x$ ends in $z\in\hat{V}$ which is between $\first_{A_j}(x)$ and $\first_{A_j}(v)$, then it means that there is some cluster $C_r$ containing $z$ such that $\p_{T}(C_r)=C_j$ and $g(C_j,C_r)=(z,z')$, and in addition, we have $\first_{A_j}(x)<\first_{A_j}(z)<\first_{A_j}(v)$.
This implies that the rank of $C_r$ among children of $C_j$ should be between $l$ and $l+1$ which leads to a contradiction.
Therefore, $(u,v),A_v,(x,y)$ is a connected path in $A'$.

\paragraph{Case 2.1: $C_i=\p_{T}(C_j)$ and there is some $(p,q)$ after $(u,v)$ in $\hat{A}$ satisfying $q=v$.} 
In this case, it implies that $C_j$ cannot be a leaf because we should be able to find $C_r$ such that $C_r$ contains $p$, $g(C_r,C_j)=(p,q)$ and $\p_T(C_r)=C_j$.
Note that we circular shifted $A_j$ such that $A_j$ starts from $v$.
Based on how we sorted the children of $C_j$, we now that $g(C_j,\child(C_j,1))$ should be some $(x,y)$ where $x=v$.
In $\bar{A}$, $(C_j,\child(C_j,1))$ should follow $(C_i,C_j)$ immediately, which means $(x,y)$ follows $(u,v)$ $\hat{A}$ as well as $A'$.
Since $v=x$, $(u,v),(x,y)$ is a connected path in $A'$.

\paragraph{Case 2.2: $C_i=\p_{T}(C_j)$ and there is no $(p,q)$ after $(u,v)$ in $\hat{A}$ satisfying $q=v$.}
Note that we circular shifted $A_j$ such that $A_j$ starts from $v$.
If $C_j$ is a leaf in $T$, then the edge in $\bar{A}$ following $(C_i,C_j)$ is $(C_j,C_i)$ and thus, the edge in $\hat{A}$ following $(u,v)$ is $(v,u)$.
Since $A_j$ is a cycle starting from $v$, we have $(u,v),A_j,(v,u)$ in $A'$ which is a connected path (or cycle).

If $C_j$ is not a leaf in $T$, the edge $(C_j,\child_{\p_T}(C_j,1))$ should directly follow $(C_i,C_j)$ in $\bar{A}$ and thus $(x,y) = g(C_j,\child_{\p_T}(C_j,1))$ directly follows $(u,v)$ in $\hat{A}$.
Since $A_j$ starts from $v$, we have $A_v$ is a prefix of $A_j$.
If $A_v$ ends in some $z\in\hat{V}$ which is before $\first_{A_j}(x)$, then there must be a cluster $C_r$ such that $g(C_j,C_r)=(z,z')$ which implies that $C_r$ should have smaller than rank than $\child_{\p_T}(C_j,1)$ which contradicts to the definition of $\child_{\p_T}(C_j,1)$.
Therefore $A_v$ must ends in $x$ which implies that $(u,v),A_j,(x,y)$ is a connected path in $A'$.

\paragraph{Put them together.}
Since each $(u,v)$ in $\hat{A}$ can find a valid path in $A'$ to the following $(v,p)$ in $\hat{A}$, $A'$ is indeed a cycle and thus $A'$ is Euler tour of $T'$.
\end{proof}

\begin{theorem}[Euler tour join]\label{thm:euler_join}
Given (1) a partition $\cC=\{C_1,C_2,\cdots, C_k\}$ over a set of $n$ nodes $V$, (2) a spanning tree $T=(\cC,E)$ over $\cC$ and an arbitrary Euler tour $A$ of $T$, (3) an edge mapping $g$ where $\forall \{C_i,C_j\}\in E,g(e)=\{x,y\}$ and $x\in C_i,y\in C_j$, and (4) spanning trees $T_1,T_2,\cdots,T_k$ and corresponding Euler tours $A_1,A_2,\cdots,A_k$ for $C_1,C_2,\cdots, C_k$ respectively, there is an MPC algorithm (Algorithm~\ref{alg:euler_join}) that outputs a spanning tree $T'$ over $V$ and an Euler tour $A'$ of $T'$ using $O(1)$ rounds and $O(n)$ total space where the edges of $T'$ is the union of edges in $T_1,T_2,\cdots,T_k$ and edges $\{g(e)\mid e\in E\}$.
In addition, the algorithm is fully scalable.
\end{theorem}
\begin{proof}
The correctness of Algorithm~\ref{alg:euler_join} is proved by Lemma~\ref{lem:euler_join_correctness}.
In the remaining of the proof, we show how to implement Algorithm~\ref{alg:euler_join} in the MPC model using $O(1)$ rounds, $O(n)$ total space, and make the implemented algorithm fully scalable.

Computing $\hat{E}$ only require simultaneously access $g(e)$ for each $e\in E$.
According to Theorem~\ref{thm:pram_simulation}, this step only takes $O(1)$ rounds and $O(n)$ total space.
Computing $T'$ only requires local computation (i.e., renaming the tuples).
To compute $\hat{V}$, for each $\{x,y\}\in E$, it generates $x$ and $y$.
This only requires local computation.
Then we run sorting (Theorem~\ref{thm:sorting}) to remove duplicates.
To choose an arbitrary $C\in\mathcal{C}$ to be a root, we can sort (Theorem~\ref{thm:sorting}) all $C_1,C_2,\cdots,C_k$ via an arbitrary sorting key and broadcast (Theorem~\ref{thm:broadcasting}) the name of the cluster with the smallest sorting key.
We apply lemma~\ref{lem:change_root} to circular shift $A$ to make $A$ be an Euler tour of $T$ with respect to the root $C$.
These steps take $O(1)$ rounds, $O(n)$ total space, and are fully scalable.
To compute $\p_T(C_i)$ for all $C_i$ simultaneously, we use sorting (Theorem~\ref{thm:sorting}) to compute $\first_{A}(C_i)$, then we use Observation~\ref{obs:parent_by_euler_tour} to compute $\p_T(C_i)$ which only requires another call of sorting.
Therefore computing $\p_T(C_i)$ for all $C_i$ can be done using $O(1)$ rounds, $O(n)$ total space, and in fully scalable setting.
Computing $g(C_i,\p_T(C_i))$ only requires simultaneous accesses (Theorem~\ref{thm:pram_simulation}), and circular shifting all $A_i$ at the same time can be done in parallel using Lemma~\ref{lem:change_root}.
These steps take $O(1)$ rounds, $O(n)$ total space, and are fully scalable.
To sort the children of each $C_i$, we do the following: 
\begin{enumerate}
    \item For all $i\in[k],x\in C_i$, simultaneously compute $\first_{A_i}(x)$. 
    This step can be done via sorting (Theorem~\ref{thm:sorting}) which takes $O(1)$ rounds, $O(n)$ total space and is fully scalable.
    \item For all $j\in[k]$, compute $(x,y)=g(C_j,C_i)$ where $C_i=\p_T(C_j)$, add $C_j$ into a set $S_{C_i}$ and give $C_j$ a sorting key $\first_{A_i}(y)$.
    These steps only require simultaneous accesses (Theorem~\ref{thm:pram_simulation}), and thus they can be finished in $O(1)$ rounds and $O(n)$ total space. 
    The computation is also fully scalable.
    \item We use Theorem~\ref{thm:index_in_sets} to simultaneously compute the rank of each $C_j$ in the set $S_{C_i}$ where $C_i=\p_T(C_j)$.
    Suppose $C_j$ is the $l$-th element in $S_{C_i}$, then we set $\child(C_i,l)=C_j$.
    This step takes $O(1)$ rounds, $O(n)$ total space, and is fully scalable.
\end{enumerate}
Now we have parent pointers $\p_T(\cdot)$ for $T$, a mapping $\child(\cdot,\cdot)$ and an Euler tour $A$ of $T$.
We apply Lemma~\ref{lem:mpc_euler_by_children} to compute a new Euler tour $\bar{A}$ of $T$ with respect to the root $C$, such that $\forall C_i\in\cC$, $\first_{\bar{A}}(\child(C_i,1))<\first_{\bar{A}}(\child(C_i,2))<\cdots<\first_{\bar{A}}(\child(C_i,|\child(C_i)|))$.
According to Lemma~\ref{lem:mpc_euler_by_children}, this step only takes $O(1)$ rounds, $O(n)$ total space, and it is fully scalable.
For each $(C_p,C_q)$ appeared in $\bar{A}$, we simultaneously access $g(C_p,C_q)$ and replace $(C_p,C_q)$ with $g(C_p,C_q)$ to obtain $\hat{A}$.
This step can be done by Theorem~\ref{thm:pram_simulation}, and it takes $O(1)$ rounds, $O(n)$ total space, and it is fully scalable.
Next, we need to compute $A_x$ for $x\in\hat{V}$.
To do it, for each $(u,v)$ appeared in each $A_i$, we check whether its position is equal to $\first_{A_i}(u)$, if it is, we mark $(u,v)$ as $1$ otherwise, we mark $(u,v)$ as $0$.
Then, for each $(u,v)$ we find the closest $(u',v')$ appeared in $A_i$ such that $(u',v')$ is marked as $1$ and is appeared before $(u,v)$.
Then we put $(u,v)$ into $A_{u'}$ and its index can be derived from the distance between $(u,v)$ and $(u',v')$ in $A_i$.
Above steps can be implemented by simultaneous accesses (Theorem~\ref{thm:pram_simulation}) and predecessor algorithm (Theorem~\ref{thm:predecessor}).
Therefore, computing $A_x$ for all $x\in \hat{V}$ simultaneously only requires $O(1)$ rounds and $O(n)$ total space.
Computing $f$ can be done using sorting (Theorem~\ref{thm:sorting}) which can be done in $O(1)$ rounds and $O(n)$ total space and is fully scalable.
Finally, we run sequence insertion algorithm (Theorem~\ref{thm:sequence_insertion}) on $\hat{A},f,\{A_x\}$, to obtain $A'$, and it takes $O(1)$ rounds and $O(n)$ total space and is fully scalble as well.
\end{proof}

\subsection{Euler Tour via Hierarchical Decomposition}
Let $V$ be a set of nodes.
Let $\cC_0 \sqsupseteq \cC_1 \sqsupseteq \cdots \sqsupseteq \cC_L$ be a hierarchy of partitions on $V$ where $\cC_0=\{\{v\}\mid v\in V\}$ and $\cC_L=\{V\}$.
Let $E_1,E_2,\cdots,E_L$ be arbitrary sets of edges between nodes in $V$ such that $\forall l\in[L],\cC_{l-1}\oplus E_l = \cC_l$ and $|\cC_{l-1}|-|E_l|=|\cC_{l}|$, i.e., any edge in $E_l$ only connects two different components in $\cC_{l-1}$.
Then it is easy to see that $T=(V,\bigcup_{l\in [L]}E_l)$ is a spanning tree of $V$.
If $L$ is small, and for every $l\in [L]$, the tree obtained by regarding each component in $\cC_{l-1}$ as a node and regarding each pair of components in $\cC_{l-1}$ that are connected by an edge in $E_l$ as an edge has a low diameter, then we provide a novel MPC algorithm to efficiently compute an Euler tour of $T$.

In the MPC model, we use a leader mapping $\ell_l:V\rightarrow V$ to denote the partitioning $\cC_l$, i.e., $x,y$ are in the same component in $\cC_l$ iff $\ell_l(x)=\ell_l(y)$.
The edges $E_l$ are distributed arbitrarily on the machines.

\begin{algorithm}[tb]
   \caption{$\textsc{EulerTourViaHierarchicalDecomposition}(\{\cC_l\}\text{ (represented by $\ell_l$)}, \{E_l\})$: Generating an Euler tour of the spanning tree $T=(V,E_1\cup E_2\cup \cdots \cup E_L)$.}
   \label{alg:euler_tour}
\begin{algorithmic}[1]
\small
    \STATE \textbf{Inputs:} 
    Leader mappings $\ell_0(\cdot),\ell_1(\cdot),\cdots,\ell_L(\cdot)$ representing a hierarchical decomposition $\cC_0 \sqsupseteq \cC_1 \sqsupseteq \cdots \sqsupseteq \cC_L$ of $n$ nodes $V$ where $\cC_0=\{\{v\}\mid v\in V\}$ and $\cC_L=\{V\}$, and sets of edges $E_1,E_2,\cdots,E_L$ between nodes in $V$ such that $\forall l\in[L],\cC_{l-1}\oplus E_l = \cC_l$ and $|\cC_{l-1}|-|E_l|=|\cC_{l}|$.
    \STATE Let $R=\log L$. \COMMENT{We suppose $L$ is a power of $2$.}
    \STATE \COMMENT{Notation:  $\forall l\in[L]\cup \{0\},\forall C\in \cC_l,\forall l'\leq l$, we denote 
    $$C^{(l')}=\{C'\in \cC_{l'}\mid C'\subseteq C\},$$ i.e., interpret $C$ as a set of clusters at level $l'$.
    Similarly, $\forall l\in[L]\cup \{0\},\forall E\subseteq \bigcup_{j=l}^{L} E_j,\forall l'< l,$ we denote 
    $$E^{(l')}=\{\{C_x,C_y\}\mid \{x,y\}\in E, x\in C_x,y\in C_y,\text{ and }C_x,C_y\in \cC_{l'}\},$$
    i.e., interpret $E$ as a set of edges between nodes at level $l'$, where each node denotes a cluster at level $l'$.
    Let $E(C)$ denotes the subset of edges where both end nodes are in $C$.
    For $0\leq l'\leq l\leq L$, we denote $g^{(l\rightarrow l')}$ as following, if there is an edge $\{x,y\}\in \bigcup_{j=l}^{L} E_j,$ then $g^{(l\rightarrow l')}(\{C_x,C_y\})=\{C'_x,C'_y\}$ where $C_x,C_y\in \cC_l,C'_x,C'_y\in\cC_{l'},x\in C'_x\subseteq C_x,y\in C'_y\subseteq C_y$. 
    }
    \FOR{$l\in[L]$ in parallel}
        \FOR{$C\in \cC_l$ in parallel}
            \STATE Compute the Euler tour $A_l^{(0)}(C)$ of the tree $T_l^{(0)}(C)=(C^{(l-1)},E_l(C)^{(l-1)})$
        \ENDFOR
    \ENDFOR
    \FOR{$r:=1\rightarrow R$}
        \FOR{$l\in\{1\cdot 2^r,2\cdot 2^r,3\cdot 2^r,4\cdot 2^r, \cdots, L\}$ in parallel}
            \FOR{$C\in \cC_l$ in parallel}
                \STATE $T_l^{(r)}(C), A_l^{(r)}(C)\gets \textsc{EulerTourJoin}(\mathcal{C},T,A,g,\{T_i\},\{A_i\})$ \COMMENT{Algorithm~\ref{alg:euler_join}} where: 
                
                     $\mathcal{C}=C^{(l-2^{r-1})}$,
                     $T = T_l^{(r-1)}(C)$,
                     $A = A_l^{(r-1)}(C)$,
                     $g = g^{(l-2^{r-1}\rightarrow l-2^r)}$,
                     
                     $\left\{T_i\right\}=\left\{T_{l-2^{r-1}}^{(r-1)}(C')\mid C'\in C^{(l-2^{r-1})}\right\}$,
                     $\left\{A_i\right\}=\left\{A_{l-2^{r-1}}^{(r-1)}(C')\mid C'\in C^{(l-2^{r-1})}\right\}$
            \ENDFOR
        \ENDFOR
    \ENDFOR
 \STATE \textbf{Output:} Euler tour $A_L^{(R)}(V)$ of spanning tree $T_L^{(R)}(V)$.
\end{algorithmic}
\end{algorithm}

\begin{lemma}(correctness of Algortihm~\ref{alg:euler_tour})\label{lem:correctness_intermediate_euler_tour}
For $r\in \{0\}\cup [R],l\in [L]$ with $l\bmod 2^r = 0$, $\forall C\in \cC_l$, $T_l^{(r)}(C)$ and $A_l^{(r)}(C)$ satisfy following properties:
\begin{enumerate}
\item $T_l^{(r)}(C)$ is a spanning tree where each node in the tree denotes a cluster $C'\in\cC_{l-2^r}$ and $C'\subseteq C$.
\item The edge set of $T_l^{(r)}(C)$ is $\{\{C_x,C_y\} \mid \{x,y\}\in \bigcup_{j=l-2^r+1}^l E_j, x \in C_x,y\in C_y, C_x,C_y\in \cC_{l-2^r},C_x,C_y\subseteq C\}$.
\item $A^{(r)}_l(C)$ is an Euler tour of $T_l^{(r)}(C)$.
\end{enumerate}
\end{lemma}
\begin{proof}
Our proof is by induction.
The base case is $r=0$, the claimed properties of $T_l^{(0)}(C),A_l^{(0)}(C)$ obviously hold for all $l\in [L],C\in\cC_l$.

Now consider the case that the claimed properties hold for $T_l^{(r-1)}(C),A_l^{(r-1)}(C)$  for all $l\in [L]\text{ with }l \bmod 2^{r-1}=0$ and all $C\in\cC_l$.
Consider a cluster $C\in\cC_l$ and how we compute $T_l^{(r)}(C)$ and $A_l^{(r)}(C)$.
By our induction hypothesis, 
$T_l^{(r-1)}(C)$ is a spanning tree where each node in the tree denotes a cluster $C'\in\cC_{l-2^{r-1}}$ and $C'\subseteq C$.
$A_l^{(r-1)}(C)$ is an Euler tour of $T_l^{(r-1)}(C)$.
In addition, the edge set of $T_l^{(r-1)}(C)$ is $\{\{C_x,C_y\} \mid \{x,y\}\in \bigcup_{j=l-2^{r-1}+1}^l E_j, x \in C_x,y\in C_y, C_x,C_y\in \cC_{l-2^{r-1}},C_x,C_y\subseteq C\}$.
Then for each $C'\in \cC_{l-2^{r-1}}$ and $C'\subseteq C$, we have that $T_{l-2^{r-1}}^{(r-1)}(C')$ is a spanning tree where each node in the tree denotes a cluster $C''\in\cC_{l-2^{r}}$ and $C''\subseteq C'$.
$A_{l-2^{r-1}}^{(r-1)}(C')$ is an Euler tour of $T_{l-2^{r-1}}^{(r-1)}(C')$.
The edges of $T_{l-2^{r-1}}^{(r-1)}(C')$ is $\{\{C_x,C_y\} \mid \{x,y\}\in \bigcup_{j=l-2^{r}+1}^{l-2^{r-1}} E_j, x \in C_x,y\in C_y, C_x,C_y\in \cC_{l-2^{r}},C_x,C_y\subseteq C'\}$.
Since $g^{(l-2^{r-1}\rightarrow l-2^r)}$ maps each edge in $\{\{C_x,C_y\} \mid \{x,y\}\in \bigcup_{j=l-2^{r-1}+1}^l E_j, x \in C_x,y\in C_y, C_x,C_y\in \cC_{l-2^{r-1}},C_x,C_y\subseteq C\}$ to a corresponding edge in $\{\{C_x,C_y\} \mid \{x,y\}\in \bigcup_{j=l-2^{r-1}+1}^l E_j, x \in C_x,y\in C_y, C_x,C_y\in \cC_{l-2^{r}},C_x,C_y\subseteq C\}$.
By applying Theorem~\ref{thm:euler_join}, we have that $T_l^{(r)}(C)$ is a spanning tree where each node in the tree denotes a cluster $C'\in\cC_{l-2^r}$ and $C'\subseteq C$.
$A^{(r)}_l(C)$ is an Euler tour of $T_l^{(r)}(C)$.
In addition, The edge set of $T_l^{(r)}(C)$ is $\{\{C_x,C_y\} \mid \{x,y\}\in \bigcup_{j=l-2^r+1}^l E_j, x \in C_x,y\in C_y, C_x,C_y\in \cC_{l-2^r},C_x,C_y\subseteq C\}$.
Therefore, the stated properties also hold for $T_l^{(r)}(C)$ and $A_l^{(r)}(C)$.
\end{proof}

If we regard each singleton cluster as the node itself, we get the following corollary.
\begin{corollary}[Correctness of final output of Algorithm~\ref{alg:euler_tour}]\label{cor:correctness_euler_tour}
At the end of Algorithm~\ref{alg:euler_tour},
 $T_L^{(R)}(V)$ is a spanning tree over $V$.
 The edge set of $T_L^{(R)}(V)$ is $\bigcup_{j=1}^L E_j$.
 $A^{(R)}_L(V)$ is an Euler tour of $T_L^{(R)}(V)$.
\end{corollary}

\begin{theorem}[Euler tour via hierarchical decomposition]\label{thm:euler_tour_via_hc}
Consider leader mappings $\ell_0(\cdot),\ell_1(\cdot),\cdots,\ell_L(\cdot)$ representing a hierarchical decomposition $\cC_0 \sqsupseteq \cC_1 \sqsupseteq \cdots \sqsupseteq \cC_L$ of $n$ nodes $V$ where $\cC_0=\{\{v\}\mid v\in V\}$ and $\cC_L=\{V\}$, and arbitrary sets of edges $E_1,E_2,\cdots,E_L$ between nodes in $V$ such that $\forall l\in[L],\cC_{l-1}\oplus E_l = \cC_l$ and $|\cC_{l-1}|-|E_l|=|\cC_{l}|$.
Let $\Lambda$ be an upper bound such that $\forall l\in [L]$, if we regard each cluster in $\cC_{l-1}$ as a node and each pair of clusters with an edge in $E_l$ connecting them as an edge, the diameter of the tree is at most $\Lambda$.
Then there is a fully scalable MPC algorithm (Algorithm~\ref{alg:euler_tour}) which takes $O(\log(L)+\log(\Lambda))$ rounds and $O(nL+n^{1+\varepsilon})$ total space outputting an Euler tour of the spanning tree $T=\left(V,\bigcup_{l\in [L]} E_l\right)$ where $\varepsilon>0$ is an arbitrarily small constant.
\end{theorem}
\begin{proof}
The correctness is proved by Corollary~\ref{cor:correctness_euler_tour}.
In the following, we show how to implement Algorithm~\ref{alg:euler_tour} in the MPC model.
Note that we store all leader mappings on machines and thus it takes total space $O(nL)$.
The inputs for computing $T_l^{(0)}(C)$ over all $l\in[L],C\in\cC_l$ are disjoint, thus we can use sorting (Theorem~\ref{thm:sorting}) to assign each Euler tour computation to a disjoint group of machines.
Then we independently apply Corollary~\ref{thm:euler_tour_low_diameter} for each task in parallel.
Therefore, we use $O(\log\Lambda)$ rounds and $O(n^{1+\varepsilon})$ total space to compute $T_l^{(0)}(C)$ for all $l\in[L],C\in\cC_l$.
The algorithm is fully scalable.

Then, we run $R$ iterations.
In each iteration, we run for all considered $l$ and all $C\in\cC_l$ in parallel. 
We need to call \textsc{EulerTourJoin} (Algorithm~\ref{alg:euler_join}) for each $l$ and $C$.
Due to Lemma~\ref{lem:correctness_intermediate_euler_tour}, the inputs for subroutines of \textsc{EulerTourJoin} in one iteration are disjoint, and the total size of inputs are at most $O(n)$.
Therefore, we can use sorting (Theorem~\ref{thm:sorting}) to assign each subroutine to a disjoint group of machines.
Note that, we only need to compute one $g^{(l-2^{r-1}\rightarrow l-2^r)}$ for all $C\in\cC_l$.
We only need simultaneous accesses (Theorem~\ref{thm:pram_simulation}) to $\ell_{l-2^{r-1}}$ and $\ell_{l-2^{r-1}}$ to compute $g^{(l-2^{r-1}\rightarrow l-2^r)}$.
Therefore, to prepare the inputs for each subroutine and assign machines for each subtask, we need $O(1)$ rounds, $O(nL)$ total space and the process is fully scalable.
Then, for each subtask of calling \textsc{EulerTourJoin}, we apply Theorem~\ref{thm:euler_join}, which takes $O(n)$ space in total. 
It takes $O(1)$ rounds and is fully scalalbe.
Since we have $O(R)=O(\log L)$ iterations, it takes $O(R)$ rounds.

Therefore, the total number of rounds is $O(\log R + \log \Lambda)$, and the total space required is at most $O(n^{1+\varepsilon}+nL)$ where $\varepsilon>0$ is an arbitrarily small constant.
The entire algorithm is fully scalable.
\end{proof}

\subsection{Euler Tour of Approximate Euclidean MST}
Then, by applying Algorithm~\ref{alg:euler_tour} on the approximate MST and the corresponding hierarchical decomposition $\hat{\cP}_1,\hat{\cP}_2,\hat{\cP}_4,\cdots,\hat{\cP}_{\alpha^H}$ that we obtained in Section~\ref{sec:offline_algorithm} and Section~\ref{sec:mpc_mst}, we are able to output an Euler tour of our approximate MST.

\begin{theorem}[Approximate Euclidean MST with Euler tour]\label{thm:main_mst}
Given $n$ points from $\mathbb{R}^d$, there is a fully scalable MPC algorithm which outputs an $O(1)$-approximate MST with probability at least $0.99$.
In addition, the algorithm also outputs an Euler tour of the outputted approximate MST.
The number of rounds of the algorithm is $O(\log\log(n)\cdot \log\log\log(n))$.
The total space required is at most $O(nd+n^{1+\varepsilon})$ where $\varepsilon>0$ is an arbitrary small constant.
\end{theorem}
\begin{proof}
The approximate Euclidean MST is shown by Theorem~\ref{thm:mpc_mst}. 
Note that at the end of Part 3 (Algorithm~\ref{alg:part-3}) of our algorithm, we also obtain a hierarchical decomposition: $\hat{\cP}_1 \sqsupseteq \hat{\cP}_2 \sqsupseteq \hat{\cP}_4 \sqsupseteq \cdots \sqsupseteq \hat{\cP}_{\alpha^H}$ which has at most $\polylog(n)$ levels.
In addition, we showed that $\hat{\cP}_{t/2}\oplus (F_t\cup \bigcup_{i\in[h]} E_t^{(i)}) = \hat{\cP}_t$ and $|\hat{\cP}_{t/2}| - |F_t\cup \bigcup_{i\in[h]} E_t^{(i)}| = |\hat{\cP}_t|$.
\begin{claim}
For any $t$, if we regard each cluster in $\hat{\cP}_{t/2}$ as a node and each pair of clusters in $\hat{\cP}_{t/2}$ that is connected by an edge in $F_t\cup \bigcup_{i\in[h]} E_t^{(i)}$ as an edge, then each tree has diameter at most $2^{O(h)}=\polylog(n)$.
\end{claim}
\begin{proof}
If we just run one round leader compression, each tree can have diameter at most $2$.
Since in every round of leader compression, we merge clusters using a star which may blow up the diameter by at most a factor of $5$.
After leader compression, we create stars to merge incomplete components.
This operation can blow up the diameter by a factor of $5$ as well.
Therefore, the diameter can be at most $5^{O(h)}=\polylog(n)$.
\end{proof}
Then, by plugging $\{\hat{\cP}_t \mid t=2^j,j\in[O(\log n)]\}$ and $\{F_t\cup \bigcup_{i\in[h]} E_t^{(i)}\mid t=2^j,j\in[O(\log n)]\}$ into Theorem~\ref{thm:euler_tour_via_hc}, we can use additional $O(n^{1+\varepsilon})$ total space and $O(\log\log(n))$ number of rounds to compute an Euler tour of our approximate MST.
The algorithm is fully scalable.
\end{proof}

The Euclidean travelling salesman problem (TSP) is stated as the following: Given $n$ points, the goal is to output a cycle over points such that each point is visited exactly once such that the total length of the cycle is minimized.

\begin{corollary}[Approximate Euclidean TSP]\label{cor:tsp}
Given $n$ points from $\mathbb{R}^d$, there is a fully scalable MPC algorithm which outputs an $O(1)$-approximate TSP solution with probability at least $0.99$.
The number of rounds of the algorithm is $O(\log\log(n)\cdot \log\log\log(n))$.
The total space required is at most $O(nd+n^{1+\varepsilon})$ where $\varepsilon>0$ is an arbitrary small constant.
\end{corollary}
\begin{proof}
It is easy to observe that the MST cost and the length of Euler tour of the MST are the same up to a factor of $2$.
Since optimal TSP cost is greater than MST cost and less than the cost of Euler tour, we only need to output a shortcut Euler tour of a constant approximate MST to get a constant approximate TSP solution.

By applying Theorem~\ref{thm:main_mst}, we obtain an Euler tour of $O(1)$ approximate MST.
Then we can use sorting and (re)indexing (Theorem~\ref{thm:sorting}) to firstly only keep the first appearance of each point on the Euler tour and then recompute the index of each point in the deduplicated tour sequence to provide an $O(1)$-approximate TSP solution.
Above steps only takes $O(1)$ additional MPC rounds and $O(n)$ additional total space.
These steps are fully scalable.
\end{proof}

\bibliography{cluster}

\newcommand{\etalchar}[1]{$^{#1}$}
\begin{thebibliography}{BBD{\etalchar{+}}17b}

\bibitem[AAH{\etalchar{+}}23]{ahanchi2023massively}
AmirMohsen Ahanchi, Alexandr Andoni, MohammadTaghi Hajiaghayi, Marina Knittel,
  and Peilin Zhong.
\newblock Massively parallel tree embeddings for high dimensional spaces.
\newblock In {\em Proceedings of the 35th ACM Symposium on Parallelism in
  Algorithms and Architectures}, pages 77--88, 2023.

\bibitem[ACK{\etalchar{+}}16]{andoni2016sketching}
Alexandr Andoni, Jiecao Chen, Robert Krauthgamer, Bo~Qin, David~P Woodruff, and
  Qin Zhang.
\newblock On sketching quadratic forms.
\newblock In {\em Proceedings of the 2016 ACM Conference on Innovations in
  Theoretical Computer Science}, pages 311--319. ACM, 2016.

\bibitem[AESW90]{agarwal1990euclidean}
Pankaj~K Agarwal, Herbert Edelsbrunner, Otfried Schwarzkopf, and Emo Welzl.
\newblock Euclidean minimum spanning trees and bichromatic closest pairs.
\newblock In {\em Proceedings of the sixth annual symposium on Computational
  geometry}, pages 203--210, 1990.

\bibitem[AIK08]{andoni2008earth}
Alexandr Andoni, Piotr Indyk, and Robert Krauthgamer.
\newblock Earth mover distance over high-dimensional spaces.
\newblock In {\em SODA}, volume~8, pages 343--352. Citeseer, 2008.

\bibitem[ANOY14]{AndoniNikolov}
Alexandr Andoni, Aleksandar Nikolov, Krzysztof Onak, and Grigory Yaroslavtsev.
\newblock Parallel algorithms for geometric graph problems.
\newblock In {\em Proceedings of the Forty-Sixth Annual ACM Symposium on Theory
  of Computing}, STOC '14, page 574–583, New York, NY, USA, 2014. Association
  for Computing Machinery.

\bibitem[ASS{\etalchar{+}}18]{andoni2018parallel}
Alexandr Andoni, Zhao Song, Clifford Stein, Zhengyu Wang, and Peilin Zhong.
\newblock Parallel graph connectivity in log diameter rounds.
\newblock In {\em 2018 IEEE 59th Annual Symposium on Foundations of Computer
  Science (FOCS)}, pages 674--685. IEEE, 2018.

\bibitem[ASW19]{assadi2019massively}
Sepehr Assadi, Xiaorui Sun, and Omri Weinstein.
\newblock Massively parallel algorithms for finding well-connected components
  in sparse graphs.
\newblock In {\em Proceedings of the 2019 ACM Symposium on principles of
  distributed computing}, pages 461--470, 2019.

\bibitem[Bar96]{bartal1996probabilistic}
Yair Bartal.
\newblock Probabilistic approximation of metric spaces and its algorithmic
  applications.
\newblock In {\em Proceedings of 37th Conference on Foundations of Computer
  Science}, pages 184--193. IEEE, 1996.

\bibitem[BBD{\etalchar{+}}17a]{bateni}
Mohammadhossein Bateni, Soheil Behnezhad, Mahsa Derakhshan, MohammadTaghi
  Hajiaghayi, Raimondas Kiveris, Silvio Lattanzi, and Vahab Mirrokni.
\newblock Affinity clustering: Hierarchical clustering at scale.
\newblock In I.~Guyon, U.~Von Luxburg, S.~Bengio, H.~Wallach, R.~Fergus,
  S.~Vishwanathan, and R.~Garnett, editors, {\em Advances in Neural Information
  Processing Systems}, volume~30. Curran Associates, Inc., 2017.

\bibitem[BBD{\etalchar{+}}17b]{bateni2017affinity}
MohammadHossein Bateni, Soheil Behnezhad, Mahsa Derakhshan, MohammadTaghi
  Hajiaghayi, Raimondas Kiveris, Silvio Lattanzi, and Vahab Mirrokni.
\newblock Affinity clustering: Hierarchical clustering at scale.
\newblock {\em Advances in Neural Information Processing Systems}, 30, 2017.

\bibitem[BDE{\etalchar{+}}19]{behnezhad2019near}
Soheil Behnezhad, Laxman Dhulipala, Hossein Esfandiari, Jakub Lacki, and Vahab
  Mirrokni.
\newblock Near-optimal massively parallel graph connectivity.
\newblock In {\em 2019 IEEE 60th Annual Symposium on Foundations of Computer
  Science (FOCS)}, pages 1615--1636. IEEE, 2019.

\bibitem[BKS17]{beame2017communication}
Paul Beame, Paraschos Koutris, and Dan Suciu.
\newblock Communication steps for parallel query processing.
\newblock {\em Journal of the ACM (JACM)}, 64(6):1--58, 2017.

\bibitem[Bor26]{boruuvka1926jistem}
Otakar Boruvka.
\newblock O jist{\'e}m probl{\'e}mu minim{\'a}ln{\'\i}m.
\newblock 1926.

\bibitem[CAMZ22]{cohen2022massively}
Vincent Cohen-Addad, Vahab Mirrokni, and Peilin Zhong.
\newblock Massively parallel $ k $-means clustering for perturbation resilient
  instances.
\newblock In {\em International Conference on Machine Learning}, pages
  4180--4201. PMLR, 2022.

\bibitem[CC22]{coy2022deterministic}
Sam Coy and Artur Czumaj.
\newblock Deterministic massively parallel connectivity.
\newblock In {\em Proceedings of the 54th Annual ACM SIGACT Symposium on Theory
  of Computing}, pages 162--175, 2022.

\bibitem[CCAJ{\etalchar{+}}23]{chen2023streaming}
Xi~Chen, Vincent Cohen-Addad, Rajesh Jayaram, Amit Levi, and Erik Waingarten.
\newblock Streaming euclidean mst to a constant factor.
\newblock In {\em Proceedings of the 55th Annual ACM Symposium on Theory of
  Computing}, STOC 2023, page 156–169, New York, NY, USA, 2023. Association
  for Computing Machinery.

\bibitem[CCFC02]{Charikar:2002}
Moses Charikar, Kevin Chen, and Martin Farach-Colton.
\newblock Finding frequent items in data streams.
\newblock In {\em Proceedings of the 29th International Colloquium on Automata,
  Languages and Programming}, ICALP '02, pages 693--703, London, UK, UK, 2002.
  Springer-Verlag.

\bibitem[CEF{\etalchar{+}}05]{CzumajEFMNRS05}
Artur Czumaj, Funda Erg{\"{u}}n, Lance Fortnow, Avner Magen, Ilan Newman,
  Ronitt Rubinfeld, and Christian Sohler.
\newblock Approximating the weight of the euclidean minimum spanning tree in
  sublinear time.
\newblock {\em {SIAM} J. Comput.}, 35(1):91--109, 2005.

\bibitem[Cha00]{chazelle00}
Bernard Chazelle.
\newblock A minimum spanning tree algorithm with inverse-ackermann type
  complexity.
\newblock {\em J. ACM}, 47(6):1028–1047, nov 2000.

\bibitem[CJLW22]{chen2022new}
Xi~Chen, Rajesh Jayaram, Amit Levi, and Erik Waingarten.
\newblock New streaming algorithms for high dimensional emd and mst.
\newblock In {\em Proceedings of the 54th Annual ACM SIGACT Symposium on Theory
  of Computing}, pages 222--233, 2022.

\bibitem[CLRS22]{cormen2022introduction}
Thomas~H Cormen, Charles~E Leiserson, Ronald~L Rivest, and Clifford Stein.
\newblock {\em Introduction to algorithms}.
\newblock MIT press, 2022.

\bibitem[CRT05]{chazellerubinfeld}
Bernard Chazelle, Ronitt Rubinfeld, and Luca Trevisan.
\newblock Approximating the minimum spanning tree weight in sublinear time.
\newblock {\em SIAM Journal on Computing}, 34(6):1370--1379, 2005.

\bibitem[CS04]{CzumajS04}
Artur Czumaj and Christian Sohler.
\newblock Estimating the weight of metric minimum spanning trees in
  sublinear-time.
\newblock In L{\'{a}}szl{\'{o}} Babai, editor, {\em Proceedings of the 36th
  Annual {ACM} Symposium on Theory of Computing, Chicago, IL, USA, June 13-16,
  2004}, pages 175--183. {ACM}, 2004.

\bibitem[CS09]{czumaj2009estimating}
Artur Czumaj and Christian Sohler.
\newblock Estimating the weight of metric minimum spanning trees in sublinear
  time.
\newblock {\em SIAM Journal on Computing}, 39(3):904--922, 2009.

\bibitem[DCLT18]{devlin2018bert}
Jacob Devlin, Ming-Wei Chang, Kenton Lee, and Kristina Toutanova.
\newblock Bert: Pre-training of deep bidirectional transformers for language
  understanding.
\newblock {\em arXiv preprint arXiv:1810.04805}, 2018.

\bibitem[DG04]{dean2004mapreduce}
Jeffrey Dean and Sanjay Ghemawat.
\newblock Mapreduce: Simplified data processing on large clusters.
\newblock 2004.

\bibitem[DG08]{dean2008mapreduce}
Jeffrey Dean and Sanjay Ghemawat.
\newblock Mapreduce: simplified data processing on large clusters.
\newblock {\em Communications of the ACM}, 51(1):107--113, 2008.

\bibitem[EMMZ22]{epasto2022massively}
Alessandro Epasto, Mohammad Mahdian, Vahab Mirrokni, and Peilin Zhong.
\newblock Massively parallel and dynamic algorithms for minimum size
  clustering.
\newblock In {\em Proceedings of the 2022 Annual ACM-SIAM Symposium on Discrete
  Algorithms (SODA)}, pages 1613--1660. SIAM, 2022.

\bibitem[Epp00]{eppstein2000spanning}
David Eppstein.
\newblock Spanning trees and spanners., 2000.

\bibitem[FIS05]{10.1145/1064092.1064116}
Gereon Frahling, Piotr Indyk, and Christian Sohler.
\newblock Sampling in dynamic data streams and applications.
\newblock In {\em Proceedings of the Twenty-First Annual Symposium on
  Computational Geometry}, SCG '05, page 142–149, New York, NY, USA, 2005.
  Association for Computing Machinery.

\bibitem[Goo99]{goodrich1999communication}
Michael~T Goodrich.
\newblock Communication-efficient parallel sorting.
\newblock {\em SIAM Journal on Computing}, 29(2):416--432, 1999.

\bibitem[GSZ11]{goodrich2011sorting}
Michael~T Goodrich, Nodari Sitchinava, and Qin Zhang.
\newblock Sorting, searching, and simulation in the mapreduce framework.
\newblock In {\em International Symposium on Algorithms and Computation}, pages
  374--383. Springer, 2011.

\bibitem[GZJ06]{grygorash2006minimum}
Oleksandr Grygorash, Yan Zhou, and Zach Jorgensen.
\newblock Minimum spanning tree based clustering algorithms.
\newblock In {\em 2006 18th IEEE International Conference on Tools with
  Artificial Intelligence (ICTAI'06)}, pages 73--81. IEEE, 2006.

\bibitem[HIS13]{Har-PeledIS13}
Sariel Har{-}Peled, Piotr Indyk, and Anastasios Sidiropoulos.
\newblock Euclidean spanners in high dimensions.
\newblock In Sanjeev Khanna, editor, {\em Proceedings of the Twenty-Fourth
  Annual {ACM-SIAM} Symposium on Discrete Algorithms, {SODA} 2013, New Orleans,
  Louisiana, USA, January 6-8, 2013}, pages 804--809. {SIAM}, 2013.

\bibitem[HPIM12]{har2012approximate}
Sariel Har-Peled, Piotr Indyk, and Rajeev Motwani.
\newblock Approximate nearest neighbor: Towards removing the curse of
  dimensionality.
\newblock {\em Theory of computing}, 8(1):321--350, 2012.

\bibitem[HZRS16]{he2016deep}
Kaiming He, Xiangyu Zhang, Shaoqing Ren, and Jian Sun.
\newblock Deep residual learning for image recognition.
\newblock In {\em Proceedings of the IEEE conference on computer vision and
  pattern recognition}, pages 770--778, 2016.

\bibitem[IBY{\etalchar{+}}07]{isard2007dryad}
Michael Isard, Mihai Budiu, Yuan Yu, Andrew Birrell, and Dennis Fetterly.
\newblock Dryad: distributed data-parallel programs from sequential building
  blocks.
\newblock In {\em Proceedings of the 2nd ACM SIGOPS/EuroSys European Conference
  on Computer Systems 2007}, pages 59--72, 2007.

\bibitem[Ind99]{indyk1999sublinear}
Piotr Indyk.
\newblock Sublinear time algorithms for metric space problems.
\newblock In {\em Proceedings of the thirty-first annual ACM symposium on
  Theory of computing}, pages 428--434, 1999.

\bibitem[Ind04]{indyk2004algorithms}
Piotr Indyk.
\newblock Algorithms for dynamic geometric problems over data streams.
\newblock In {\em Proceedings of the thirty-sixth annual ACM symposium on
  Theory of computing}, pages 373--380. ACM, 2004.

\bibitem[IT03]{IT03}
Piotr Indyk and Nitin Thaper.
\newblock Fast color image retrieval via embeddings.
\newblock In {\em Workshop on Statistical and Computational Theories of Vision
  (at ICCV)}, 2003.

\bibitem[JL84]{JL}
William~B. Johnson and Joram Lindenstrauss.
\newblock Extensions of lipschitz mappings into a hilbert space.
\newblock {\em Contemporary Mathematics}, 26:189--206, 1984.

\bibitem[JN18]{jurdzinski2018mst}
Tomasz Jurdzi{\'n}ski and Krzysztof Nowicki.
\newblock Mst in o (1) rounds of congested clique.
\newblock In {\em Proceedings of the Twenty-Ninth Annual ACM-SIAM Symposium on
  Discrete Algorithms}, pages 2620--2632. SIAM, 2018.

\bibitem[KSV10]{karloff2010model}
Howard Karloff, Siddharth Suri, and Sergei Vassilvitskii.
\newblock A model of computation for mapreduce.
\newblock In {\em Proceedings of the twenty-first annual ACM-SIAM symposium on
  Discrete Algorithms}, pages 938--948. SIAM, 2010.

\bibitem[LMSV11]{lattanzi2011filtering}
Silvio Lattanzi, Benjamin Moseley, Siddharth Suri, and Sergei Vassilvitskii.
\newblock Filtering: a method for solving graph problems in mapreduce.
\newblock In {\em Proceedings of the twenty-third annual ACM symposium on
  Parallelism in algorithms and architectures}, pages 85--94, 2011.

\bibitem[LMW18]{lkacki2018connected}
Jakub Lacki, Vahab Mirrokni, and Micha{\l} W{\l}odarczyk.
\newblock Connected components at scale via local contractions.
\newblock {\em arXiv preprint arXiv:1807.10727}, 2018.

\bibitem[LRN09]{lai2009approximate}
Chih Lai, Taras Rafa, and Dwight~E Nelson.
\newblock Approximate minimum spanning tree clustering in high-dimensional
  space.
\newblock {\em Intelligent Data Analysis}, 13(4):575--597, 2009.

\bibitem[MCCD13]{mikolov2013efficient}
Tomas Mikolov, Kai Chen, Greg Corrado, and Jeffrey Dean.
\newblock Efficient estimation of word representations in vector space.
\newblock {\em arXiv preprint arXiv:1301.3781}, 2013.

\bibitem[RVW18]{roughgarden2018shuffles}
Tim Roughgarden, Sergei Vassilvitskii, and Joshua~R Wang.
\newblock Shuffles and circuits (on lower bounds for modern parallel
  computation).
\newblock {\em Journal of the ACM (JACM)}, 65(6):1--24, 2018.

\bibitem[VdMH08]{van2008visualizing}
Laurens Van~der Maaten and Geoffrey Hinton.
\newblock Visualizing data using t-sne.
\newblock {\em Journal of machine learning research}, 9(11), 2008.

\bibitem[Whi12]{white2012hadoop}
Tom White.
\newblock {\em Hadoop: The definitive guide}.
\newblock " O'Reilly Media, Inc.", 2012.

\bibitem[WWW09]{wang2009divide}
Xiaochun Wang, Xiali Wang, and D~Mitchell Wilkes.
\newblock A divide-and-conquer approach for minimum spanning tree-based
  clustering.
\newblock {\em IEEE Transactions on Knowledge and Data Engineering},
  21(7):945--958, 2009.

\bibitem[YV18]{yaroslavtsev2018massively}
Grigory Yaroslavtsev and Adithya Vadapalli.
\newblock Massively parallel algorithms and hardness for single-linkage
  clustering under lp distances.
\newblock In {\em International Conference on Machine Learning}, pages
  5600--5609. PMLR, 2018.

\bibitem[ZCF{\etalchar{+}}10]{zaharia2010spark}
Matei Zaharia, Mosharaf Chowdhury, Michael~J Franklin, Scott Shenker, and Ion
  Stoica.
\newblock Spark: Cluster computing with working sets.
\newblock In {\em 2nd USENIX Workshop on Hot Topics in Cloud Computing
  (HotCloud 10)}, 2010.

\bibitem[ZMMF15]{zhong2015fast}
Caiming Zhong, Mikko Malinen, Duoqian Miao, and Pasi Fr{\"a}nti.
\newblock A fast minimum spanning tree algorithm based on k-means.
\newblock {\em Information Sciences}, 295:1--17, 2015.

\end{thebibliography}
\end{document}